\documentclass[11pt]{article}
\usepackage{enumerate}
\usepackage{enumitem}
\usepackage{fixltx2e}
\usepackage{graphicx}
\usepackage{longtable}
\usepackage{float}
\usepackage{setspace}
\usepackage{wrapfig}
\usepackage{soul}
\usepackage{amsmath}
\RequirePackage{amsthm}
\usepackage{mathtools}
\usepackage{mathrsfs}
\usepackage{textcomp}
\usepackage{marvosym}
\usepackage{wasysym}
\usepackage{latexsym}
\usepackage{amssymb}
\usepackage{hyperref}
\usepackage{graphicx}
\usepackage{subfig}
\usepackage{longtable}
\usepackage{float}
\usepackage{array}
\usepackage[ruled]{algorithm}
\usepackage[noend]{algpseudocode}
\usepackage[table]{xcolor}
\usepackage{arydshln}
\usepackage{cancel}
\usepackage{url}
\usepackage{comment}
\usepackage[square, comma]{natbib}
\setcitestyle{numbers}
\usepackage[margin=1in]{geometry}
\usepackage[affil-it]{authblk}
\usepackage{titlesec}
\usepackage[flushleft]{threeparttable}
\usepackage{color}
\titleformat{\paragraph}
{\normalfont\normalsize\bfseries}{\theparagraph}{1em}{}
\titlespacing*{\paragraph}
{0pt}{3.25ex plus 1ex minus .2ex}{1.5ex plus .2ex}

\newcommand{\subscript}[2]{$#1 _ #2$}

\def\T{{ \mathrm{\scriptscriptstyle T} }}

\allowdisplaybreaks

\numberwithin{equation}{section}

\newtheorem{theorem}{Theorem}[section]
\newtheorem{proposition}{Proposition}[section]
\newtheorem{lemma}{Lemma}[section]

\input{./Definitions}

\allowdisplaybreaks

\title{Divide-and-Conquer Bayesian Inference in Hidden Markov Models}

\author[1]{Chunlei Wang \thanks{\url{chunlei-wang@uiowa.edu}}}
\author[1]{Sanvesh Srivastava \thanks{\url{sanvesh-srivastava@uiowa.edu}}}

\affil[1]{Department of Statistics and Actuarial Science, The University of Iowa}

\date{\today}

\begin{document}

\maketitle

\begin{abstract}
  Divide-and-conquer Bayesian methods consist of three steps: dividing the data into smaller computationally manageable subsets, running a sampling algorithm in parallel on all the subsets, and combining parameter draws from all the subsets. The combined parameter draws are used for efficient posterior inference in massive data settings. A major restriction of existing divide-and-conquer methods is that their first two steps assume that the observations are independent. We address this problem by developing a divide-and-conquer method for Bayesian inference in parametric hidden Markov models, where the state space is known and finite.  Our main contributions are two-fold. First, after partitioning the data into smaller blocks of consecutive observations, we modify the likelihood for performing posterior computations on the subsets such that the posterior variances of the subset and true posterior distributions have the same asymptotic order. Second, if the number of subsets is chosen appropriately depending on the mixing properties of the hidden Markov chain, then we show that the subset posterior distributions defined using the modified likelihood are asymptotically normal as the subset sample size tends to infinity. The latter result also implies that we can use any existing combination algorithm in the third step. We show that the combined posterior distribution obtained using one such algorithm is close to the true posterior distribution in 1-Wasserstein distance under widely used regularity assumptions. Our numerical results show that the proposed method provides an accurate approximation of the true posterior distribution than its competitors in diverse simulation studies and a real data analysis. 
\end{abstract}
 
\noindent%

\section{Introduction}

\label{sec:intro}
Hidden Markov models are widely used for analyzing dependent discrete-time data. Our focus is on hidden Markov models with a known state space and parametric emission distributions, which we shorten as HMMs. Bayesian inference in HMMs using Markov chain Monte Carlo algorithms has been studied extensively from computational and theoretical perspectives \citep{Capetal06,Fru06}. In massive data settings, however, posterior computations become inefficient due to multiple passes through the full data. We address this limitation using the divide-and-conquer technique.
The resulting posterior distribution is computationally efficient and has similar asymptotic properties as that of the true posterior distribution. 

Consider the problem setup briefly. Let $(Y_{1},\ldots,Y_{n})$ be the observed data from an HMM with parameter $\theta$, $K$ be the number of subsets, and $m = n/K$ be the subset sample size, where $m$ is assumed to be an integer for convenience. The proposed method for divide-and-conquer inference in HMMs has three steps.
\begin{enumerate} 
\item Partition the observed data sequence into $K$ subsets  
\begin{align}
\label{eq:int:step1}
  Y_{[0]} = \emptyset, \, Y_{[1]} = (Y_{1},\ldots,Y_{m}), \, \ldots, \, Y_{[K]} = (Y_{(K-1)m+1},\ldots,Y_{n }). 
\end{align}
\item Given a prior density $\pi(\theta)$ and the modified conditional likelihood $p_{\theta}(Y_{[j]}\mid Y_{[j-1]})$, obtain $\theta$ draws on the subsets in parallel with posterior densities 
\begin{align}
\label{eq:int:step2}
  \pi(\theta\mid Y_{[j]}, Y_{[j-1]}) \propto \pi(\theta)\{p_{\theta}(Y_{[j]}\mid Y_{[j-1]})\}^{K},\quad j =1,\ldots,K.
\end{align}
\item Combine posterior draws of $\theta$ from all the $K$ subsets.
\end{enumerate}
The third step uses an existing combination algorithm, but we solve two major problems that remain unaddressed in the literature on divide-and-conquer Bayesian methods. First, we show that mixing properties of the hidden Markov chain determine an appropriate choice of $K$. Second, we show that $\{p_{\theta}(Y_{[j]}\mid Y_{[j-1]})\}^{K}$ in \eqref{eq:int:step2} accurately approximates the true conditional likelihood using the prediction filter as $m,K$ tend to infinity. Due to the novelty of the first two steps, the combined posterior is called the \textit{block filtered posterior distribution}.

Markov chain Monte Carlo algorithms for posterior inference in HMMs are categorized in two major groups, one based on data augmentation \citep{Robetal99a,Sco02} and the other on the Metropolis-Hastings algorithm  \citep{Celetal00,Robetal00,Capetal03,andrieu2010particle}. The former and latter groups respectively draw and average over the hidden Markov chain in every iteration. The data augmentation-type algorithms are easy to implement using forward-backward recursions, but repeated passes through the full data are time consuming in massive data settings \citep{Sco02,Ryd08}. Metropolis-Hastings-type algorithms bypass this time consuming step, but the proposal tuning required for optimal performance outweighs their advantages in practice. A major reason for the popularity of these algorithms is the availability of efficient software \citep{Ryd08}.  Unfortunately, the software faces computational bottlenecks in massive data applications. Divide-and-conquer Bayesian methods cannot solve this problem directly because they mainly focus on independent observations. Filling this gap, our goal is to scale existing Bayesian methods for HMMs to massive data sets using the divide-and-conquer technique.

Online methods based on stochastic approximation have been widely used to address the inefficiency of data augmentation \citep{CapMou09}. Online Expectation Maximization (EM) has also been extended to HMMs  and is efficient in massive data settings \citep{cappe2011online,LeFor13}. The stochastic versions of online EM are based on sub-sampling \citep{Baretal17,Quietal19} and stochastic gradient descent  \citep{WelTeh11,Bouetal18,Feaetal18,Bieetal19}. These algorithms are scalable and guarantee convergence to a general target, but they also require significant tuning for optimal performance in HMMs. Online EM has also been used to develop efficient variational Bayes algorithms, but such algorithms are known to underestimate posterior uncertainty  \citep{Fotetal14,Giaetal18}.

A variety of divide-and-conquer methods exist for efficient posterior inference using Monte Carlo algorithms \citep{Scoetal16,li2017simple,minsker2017robust,srivastava2018scalable,Robetal18,xue2019double,Joretal19,WuRob19}.  The combination steps in these methods are fairly general and their theoretical guarantees mainly rely on asymptotic normality of the subset posterior distributions. On the other hand, their first two steps assume that the observations are independent, so they are inapplicable to HMMs. The major problem lies in extending and justifying the validity of the second step for dependent observations. Specifically, the second step defines the posterior density on any subset using a quasi-likelihood that raises the subset likelihood to a power of $K$. This modification compensates for the missing $(1 - 1/K)$-fraction of the full data on any subset and ensures that the variances of the true and subset posterior distributions are of the same asymptotic order \citep{Minetal14}. The generalization of this idea to dependent data is unclear and we are unaware of any extensions for HMMs similar to \eqref{eq:int:step2}. 

Our first major contribution is the definition of subset posterior distribution for divide-and-conquer Bayesian inference in HMMs. The first step divides the observed data into $K$ non-overlapping time blocks $Y_{[1]}, \ldots, Y_{[K]}$ defined in \eqref{eq:int:step1}. For $j = 1,\ldots,K$, the second step defines the conditional likelihood of $j$th subset, which is denoted as $p_{\theta}(Y_{[j]}\mid Y_{[j-1]})$ in \eqref{eq:int:step2}, by conditioning only on $Y_{[j-1]}$ instead of $Y_{[1]}, \ldots, Y_{[j-1]}$. The modified conditional likelihood is raised to a power of $K$ and is used as the quasi-likelihood for defining the (quasi) subset posterior density in \eqref{eq:int:step2}.  The quasi-likelihood replaces the true likelihood in the original Monte Carlo algorithm for posterior inference on a subset. The quasi-likelihood on any subset is easily calculated using a simple modification of the forward-backward recursions with access to $(2/K)$-fraction of the full data, which {is} significantly more efficient than using the full data when $m \ll n$.

Our modified conditional likelihood on a subset differs from the composite likelihoods used in HMMs \citep{Varetal11}. Such approaches replace the true likelihood by the composite marginal likelihood  \citep{Ryd94, ryden1997recursive,andrieu2005line,pauli2011bayesian}. This likelihood is a product of lower-dimensional marginal likelihoods of consecutive observation blocks that fail to account for the dependence between the time blocks. Our modified likelihood conditions on the previous time block and is a type of ``fixed-lag conditional likelihood approximation'', where  the lag corresponds to the size of a time block. More importantly, {in a divide-and-conquer setup}, the modified conditional likelihood in \eqref{eq:int:step2} is raised to the power of $K$ for defining subset posterior density and $K$ tends to infinity. Such modifications are absent in composite likelihoods. Recently, there is an increasing interest in developing composite likelihoods for inference in massive data settings using stochastic gradient descent algorithm \citep{aicher2019stochastic,NemFea21}. The block filtered posterior does not compete with any of these algorithms. We are instead motivated to scale them to even larger data sets by using them for efficient sampling on the subsets. For ease of presentation, we have focused on Markov chain Monte Carlo algorithms, but our theoretical guarantees are free of any assumption on the type of sampling algorithm used on the subsets. 

Our second major contribution is to show that the $K$ subset posterior distributions defined in \eqref{eq:int:step2} are asymptotically normal as $m$ and $K$ tend to infinity. For independent observations, this is the first step in justifying asymptotic normality of the combined posterior distribution; see, for example, \citet{li2017simple}  and  \citet{xue2019double}. This step is nontrivial in our case due to the dependence within and between the $K$ time blocks. We show that if the growth of $K$ is chosen appropriately depending on the mixing properties of the hidden Markov chain, then all the subset posterior distributions are asymptotically normal and their covariance matrices have the same asymptotic order as that of the true posterior distribution. Such results that quantify the growth of $K$ on a mixing property are missing in the literature on divide-and-conquer Bayesian methods.

The asymptotic normality of subset posterior distributions greatly simplifies the third step. First, it implies that any existing combination algorithm can be used to combine  parameter draws from all the subsets. We choose the simplest one based on the Double Parallel Monte Carlo algorithm, where the combined parameter draw is a weighted average of the subset posterior draws \citep{xue2019double}. Second, the regularity assumptions required for justifying the asymptotic normality of subset posterior distributions extend naturally to the third step. Specifically, we need only an extra assumption on the weights used in the combination step for showing the asymptotic normality of the block filtered posterior distribution as $m$ and $K$ tend to infinity.

The posterior distribution of parameter in HMMs is asymptotically normal under usual regularity assumptions \citep{bickel1998asymptotic,de2008asymptotic,gassiat2014posterior,vernet2015posterior}. We extend these results to divide-and-conquer Bayesian inference in that the block filtered posterior distribution is asymptotically normal if $K$ and $m$ are chosen properly depending on mixing property of the hidden Markov chain. Our regularity assumptions and convergence rates are natural extensions of existing results in that they reduce to those in \citet{de2008asymptotic} if $K=1$ and to those in \citet{li2017simple} if the data are independent. The additional set of regularity assumptions are required for guaranteeing the accuracy of the subset posterior distribution. Finally, we identify conditions under which inference obtained using the true and block filtered posterior distributions agree asymptotically. 

\section{Divide-and-Conquer Bayesian Inference in HMMs}

\subsection{Setup and Background}
\label{setup-back}

An HMM is a discrete-time stochastic process $\{(X_{t},Y_{t}), t \geq 1\}$, where the $Y$-variables are observed and $X$-variables are latent. In our setup, HMMs satisfy three properties. First, $\{X_{t}\}$ is a time-homogeneous Markov chain with a finite state-space $\Xcal =  \{1,\ldots, S\}$, where $S$ is known, and $Y_t$ takes values in a general space $\Ycal$ for every $t$. Second, $\{Y_{t}, t \geq 1\}$ are conditionally independent given $\{X_{t}, t \geq 1\}$. Third, the distribution of $Y_t$ given $\{X_{u}, u \geq 1\}$ equals to the distribution of $Y_t$ given $X_t$. The initial distribution and transition probability matrix of $\{X_t\}$ are denoted as $r(\cdot)$ and $Q$, respectively, where $r(\cdot)$ is stationary,  $\PP(X_{1} = a) = r(a)$ for $a \in \Xcal$, and 
\begin{align}
  \label{eq:2}
  Q_{ab} = q(a, b), \quad q(a, b) = \PP(X_{t+1} = b  \mid X_{t} = a), \quad a, b \in \Xcal, \quad t \geq 1.
\end{align}
The conditional distributions of $Y_t$ given $X_t$ belong to a common family  ($t=1, \ldots, n$). For any $x \in \Xcal$, the conditional distribution of $Y_t$ given $X_t = x$ is $G(\cdot \mid x)$ and has density  $g(\cdot \mid x)$ with respect to a $\sigma$-finite measure $\mu$ on $\mathcal{Y}$.

Focusing on parametric HMMs, we adopt a widely used setup for studying their asymptotic properties \citep{de2008asymptotic,bickel1998asymptotic,leroux1992maximum}. Let $\Theta \subset \RR^d$ be the parameter space of fixed dimension $d \in \NN$, the dependence of $r(\cdot)$, $q(\cdot,\cdot)$, $g(\cdot \mid x)$ on a parameter $\theta \in \Theta$ be denoted as $r_{\theta}(\cdot)$, $q_{\theta}(\cdot,\cdot)$, $g_{\theta}(\cdot \mid x)$, and $\theta_0 \in \Theta$ be the true parameter value. We use the shorthand $Y^n_m$ for $(Y_m, \ldots, Y_n)$, where $1\leq m \leq n$. Denote the sample size as $n$, the observed data as $(Y_1, \ldots, Y_n) \equiv Y_1^n$, the latent data as $(X_1, \ldots, X_n) \equiv X_1^n$, and the augmented data as $(X_1^n, Y_1^n)$. The joint densities of $(X_1^n, Y_1^n)$ with respect to $(\text{counting measure})^{n} \times \mu^{n}$ and the marginal densities of $Y_1^n$ with respect to $\mu^{n}$ are 
\begin{align}
  \label{eq:3}
  p_{\theta}(X_1^n, Y_1^n) &= r_{\theta}(X_{1}) \prod_{t=1}^{n-1}q_{\theta}(X_{t},X_{t+1}) \prod_{t=1}^{n} g_{\theta}(Y_{t} \mid X_{t}), 
\nonumber
\\
                             p_{\theta}(Y_1^n) &= \sum_{X_1^n \in \mathcal{X}^{n}} p_{\theta}(X_1^n,Y_1^n),\quad \theta \in \Theta.
\end{align}
The true density of the observed data is $p_{\theta_0}(Y_1^n)$ for some $\theta_0 \in \Theta$. Given a prior distribution on $\Theta$ with density $\pi$, Markov chain Monte Carlo algorithms based on data augmentation-type and Metropolis-Hastings-type algorithms are widely used for posterior inference on $\theta$ \citep{Sco02,Ryd08}. When $n$ is large, both types of algorithms are computationally expensive because the former requires repeated passing through the full data and the latter evaluates the marginal likelihood $p_{\theta}(Y_1^n)$ in every iteration.

\subsection{Partitioning Scheme}
\label{div-step}

The first step in divide-and-conquer Bayesian inference partitions the full data into smaller computationally tractable subsets. A widespread practice is to randomly divide the samples into $K$ subsets, but this scheme fails to preserve the dependence structure of HMM on any subset. This leads to incorrect inference on $\theta$. We instead partition the full data into $K$ non-overlapping blocks of consecutive observations starting from $Y_1$, where each block has size $n / K$. For convenience, we assume that $n$ is divisible by $K$ and define the subset sample size as $m=n/K$. Denote the $j$th block of observation with sample size $m$ as $Y_{(j-1)m+1}^{jm}$ and the latent $X$-variables as $X_{(j-1)m+1}^{jm}$ ($j=1, \ldots, K$). The $K$ observation blocks constitute the $K$ partitions of $Y_1^n$ in that $Y_1^n = (Y_{1}^{m}, \ldots, Y_{(K-1)m+1}^{n})$. The choice of $K$ (or $m$) depending on $n$ is crucial for asymptotically valid inference on any subset relative to the true posterior distribution. Theorem \ref{thm1} in Section \ref{sec:main-results} discusses the choice of $K$ in greater detail, showing that it depends on the mixing properties of the hidden Markov chain.

After partitioning the data, the likelihoods of $\theta$ on the first and $j$th subsets ($j=2, \ldots, K$) conditional on the preceding subsets are 
{\small \begin{align}
  \label{eq:5}
  & p_{\theta}(Y_1^m) = \sum_{X_1 \in \mathcal{X}} p_{\theta}(Y_1^m \mid X_1) \,
    p_{\theta}(X_{1}), 
  \\
  & p_{\theta}(Y_{(j-1)m+1}^{jm} \mid Y_{1}^{(j-1)m}) = \sum_{X_{(j-1)m+1}} p_{\theta}(Y_{(j-1)m+1}^{jm} \mid X_{(j-1)m+1}) 
    p_{\theta}(X_{(j-1)m+1} \mid Y_{1}^{(j-1)m}),\nonumber
\end{align}}%
respectively, where $p_{\theta}(X_{(j-1)m+1} \mid Y_{1}^{(j-1)m})$ is the \textit{prediction filter} of $X_{(j-1)m+1}$ given the previous $(j-1)$ subset blocks and $p_{\theta}(X_{1})$ equals the stationary distribution $r_{\theta}(X_1)$. For every $j$, the prediction filter is computed in $O(n)$ operations using forward filtering if we have access to the full data. Unfortunately, no subset has access to the full data in the divide-and-conquer setting. Ignoring the full data, if we use $Y_{(j-1)m+1}^{jm}$ only for inference on $\theta$, then the subset $j$ likelihood uses $(1/K)$-fraction of the full data and ignores the dependence on $Y_{1}^{(j-1)m}$. We next address both problems by modifying the subset $j$ likelihood in \eqref{eq:5} using {the $K$th power of one-block conditional likelihood $\{p_{\theta}(Y_{(j-1)m+1}^{jm} \mid Y_{(j-2)m+1}^{(j-1)m})\}^{K}$} that requires accessing only $(2/K)$-fraction of the full data.

\subsection{Subset Sampling Scheme}
\label{samp-step}

In divide-and-conquer setup, any chosen sampler requires modification of the subset likelihood because every subset posterior only conditions on an $(1/K)$-fraction of the full data. A popular solution to this problem is raising the subset likelihood to the power of $K$. This is equivalent to replicating the data on a subset $K$-times, and  it compensates for the missing fraction of data, resulting in uncertainty estimates of parameters that are asymptotically equivalent to those obtained using the true posterior distribution 
\citep{Minetal14}. The likelihood modification is called \textit{stochastic approximation} and has been widely used in parametric models for independent data \citep{Srietal15,li2017simple,xue2019double}. The dependence induced by the hidden Markov chain implies that stochastic approximation developed for independent data is inapplicable for divide-and-conquer Bayesian inference in HMMs.

One of our main contributions is to extend stochastic approximation so that it accounts for dependence between subsets. The prediction filter $ p_{\theta}(X_{(j-1)m+1} \mid Y_{1}^{(j-1)m})$ in \eqref{eq:5} is geometrically ergodic and rapidly converges to its martingale limit as $m$ tends to infinity \citep{bickel1998asymptotic}. Motivated by this result and fixed-lag smoothing methods, we define the $j$th approximate ``prediction filter'' by conditioning only on the $(j-1)$th subset; that is, $p_{\theta}(X_{(j-1)m+1} \mid Y_{(j-2)m+1}^{(j-1)m})$ replaces $p_{\theta}(X_{(j-1)m+1} \mid Y_{1}^{(j-1)m})$ in \eqref{eq:5}. The likelihoods in \eqref{eq:5} are accordingly modified using this approximate prediction filter to yield the one-block conditional likelihood of $\theta$ on subset $j$ ($j = 2,\ldots,K$) as
\begin{align}
\label{eq:5.5}
  p_{\theta}(Y_{(j-1)m+1}^{jm} \mid Y_{(j-2)m+1}^{(j-1)m}) = \sum_{X_{(j-1)m+1} \in \mathcal{X}} & p_{\theta}(Y_{(j-1)m+1}^{jm} \mid X_{(j-1)m+1}) \times
\nonumber
\\
& \quad \quad 
    p_{\theta}(X_{(j-1)m+1} \mid Y_{(j-2)m+1}^{(j-1)m}).
\end{align}

The one-block conditional likelihood in \eqref{eq:5.5} accurately approximates the true conditional likehood in \eqref{eq:5}. Proposition \ref{filt-approx} in Section \ref{theory} shows that $p_{\theta}(X_{(j-1)m+1} \mid Y_{(j-1)m+1}^{(j-2)m})$ and $p_{\theta}(X_{(j-1)m+1} \mid Y_{1}^{(j-1)m})$ are close in total variation distance as $m$ tends to infinity; therefore, the one-block and true conditional likelihoods in \eqref{eq:5} and \eqref{eq:5.5}, respectively, are approximately equal for every $\theta$; see Section \ref{theory} for details. The one-block conditional likelihood in \eqref{eq:5.5} is also efficiently computed in $O(m)$ operations using forward-backward recursions and requires access to only $2m$ observations instead of $n$ required for computing the true likelihood in \eqref{eq:5}. This leads to computational gains when $m \ll n$.

The  $j$th subset posterior  distribution is defined using the $j$th one-block conditional  likelihood. {Specifically, given a prior distribution $\Pi$ on $\Theta$ with density $\pi$ defined with respect to Lebesgue measure, the density of the posterior distribution of $\theta$ given the $j$th subset is defined as}
\begin{align}
  \label{eq:s6}  
 & \tilde{\pi}_{m}(\theta \mid Y_{(j-1)m + 1}^{jm})  = \frac{e^{K w_j(\theta)} \pi(\theta)} {\int_{\Theta} e^{Kw_j(\theta)} \pi(\theta)d\theta}, 
\quad w_{j}(\theta)  = \log p_{\theta}(Y_{(j-1)m+1}^{jm} \mid Y_{(j-2)m+1}^{(j-1)m}) , 
\end{align}
for $j = 2,\ldots,K$,  where $w_1(\theta) = \log  p_{\theta}(Y_1^m)$. The quasi-likelihoods used in subset posterior densities \eqref{eq:s6} are the $K$th power of the one-block conditional likelihoods. The one-block conditional likelihoods account for the dependence of a subset on their immediately {preceding time block} and are raised to a power of $K$ to compensate for the missing $(1-1/K)$-fraction of the data. The latter idea is also used in the stochastic approximation for independent data. Indeed, if the $Y_1^n$ are independent, then \eqref{eq:s6} recovers the definition of subset posterior  distributions in divide-and-conquer Bayesian inference for independent data; therefore, our proposal in \eqref{eq:s6} generalizes stochastic approximation to dependent data.  
Following \citet{Varetal11}, we interpret the $j$th one-block conditional likelihood as a composite likelihood, which provides another composite ``marginal'' likelihood for inference in HMMs \citep{Ryd94, ryden1997recursive,andrieu2005line}. 

Existing results do not imply that $\tilde{\pi}_{m}(\theta \mid Y_{(j-1)m + 1}^{jm})$ in \eqref{eq:s6} is asymptotically normal. The one-block conditional likelihood approximates the true conditional likelihood in \eqref{eq:5.5} and is raised by a power $K$ in \eqref{eq:s6}. In our theoretical setup, $m$ and $K$ tend to infinity simultaneously. This is unlike the existing setup that uses the true conditional likelihood to define the posterior density of $\theta$, sets $K=1$, and increases $m$ to infinity \citep{de2008asymptotic,vernet2015posterior,vernet2015non}. Establishing asymptotic normality of $\tilde{\pi}_{m}(\theta \mid Y_{(j-1)m + 1}^{jm})$  is a crucial step in studying theoretical properties of the block filtered posterior distribution. Our other main contribution is to establish that if the growth rate of $K$ is chosen appropriately depending on the mixing properties of the prediction filter and $K,m$ tend to infinity, then the subset posterior distributions are asymptotically normal with covariance matrices having the same asymptotic order as that of true posterior distribution; see Theorem \ref{thm1} in Section \ref{theory} for details.

The two groups of Markov chain Monte Carlo algorithms can now be used for drawing $\theta$ using \eqref{eq:s6} on all the $K$ subsets in parallel. Let $\theta^{(t)}_{(j)}$ be the $t$th draw of $\theta$ on subset $j$ and $T$ be the number of post-burn-in draws collected on any subset.  Motivated from the asymptotic normality of subset posterior distributions, the next section develops a combination scheme that approximates the true posterior distribution using a mixture with $K$ components, where the draws from the $j$th mixture component are obtained by appropriately transforming, centering, and scaling the subsets posterior draws $\{ \theta_{(j)}^{(t)} \}_{t=1}^T$ for $j = 1,\ldots, K$. 

\subsection{Combination Scheme}

The final step of a divide-and-conquer Bayesian approach combines parameter draws from all the subsets.  Our combination scheme is an extension of the Double-Parallel Monte Carlo algorithm \citep{xue2019double}, which is computationally simple. It uses a global centering vector $\tilde \theta \in \RR^d$ and a $d \times d$ symmetric positive definite scaling matrix $\tilde \Sigma$. Let $\hat{\theta}_{j}$ be the maximum likelihood estimator and $\Sigma_{j}$ be the posterior covariance matrix of $\theta$ on subset $j$. Then, we approximate the density of the true posterior distribution using
\begin{align}
\label{eq:bfp}
  \tilde{\pi}_{n}(\theta \mid Y_{1}^{n}) =  \frac{1}{K}\sum_{j=1}^K   \frac{\det( \Sigma_j)^{1/2}}{\det(\tilde \Sigma)^{1/2}} \tilde \pi_m \left\{ \Sigma_j^{1/2} \tilde \Sigma^{-1/2} (\theta - \tilde \theta) + \hat \theta_j \mid Y_{(j-1)m + 1}^{jm} \right\},\quad \theta \in \Theta,
\end{align}
where $\tilde \pi_{m}(\theta | Y_{(j-1)m + 1}^{jm})$ is the $j$th subset posterior  density defined in \eqref{eq:s6} and $\tilde{\pi}_{n}(\theta \mid Y_{1}^{n})$ is a mixture of the $K$ subset posterior densities. The distribution that has density $\tilde{\pi}_{n}(\theta \mid Y_{1}^{n})$ in \eqref{eq:bfp} is called the \textit{block filtered posterior} distribution. Our combination scheme in \eqref{eq:bfp} depends on the choice of $(\tilde \theta,\tilde \Sigma)$, so we denote it as COMB$(\tilde \theta, \tilde \Sigma)$.

Let $\{ \theta_{(j)}^{(t)} \}_{t=1}^T$ be the parameter draws on subset $j$ for $j = 1,\ldots, K$. The application of COMB$(\tilde \theta, \tilde \Sigma)$ has two steps and yields $\theta$ draws from block filtered posterior distribution as
\begin{align}
  \label{eq:9}
  \tilde \theta_{(j)}^{(t)} = \tilde \theta + \tilde \Sigma^{1/2} \left\{ \Sigma_{j}^{-1/2}(\theta_{(j)}^{(t)}  - \hat \theta_j) \right\}, \quad j = 1, \ldots, K; \; t = 1, \ldots, T,
\end{align}
where the $\theta$ draws on subset $j$ are centered and scaled by $(\hat{\theta}_{j}, \Sigma_{j})$ for $j = 1,\ldots,K$ and then re-centered and re-scaled by $(\tilde \theta, \tilde \Sigma).$ The parameter draws $\tilde \theta_{(j)}^{(t)}$ in \eqref{eq:9} are used as the computationally efficient alternative to draws from the true posterior distribution for inference on $\theta$.
If the centering and scaling parameters $(\hat{\theta}_{j}, \Sigma_{j})_{j=1}^ K$ is unavailable, then we use their Monte Carlo estimates $(\mu_{\theta_{(j)}},\Sigma_{\theta_{(j)}})_{j=1}^ K$ for obtaining the $\theta$ draws as
\begin{align}
  \label{eq:10}
  \mu_{\theta_{(j)}} = \frac{1}{T} \sum_{t=1}^T \theta_{(j)}^{(t)}, \quad \Sigma_{\theta_{(j)}} = \frac{1}{T} \sum_{t=1}^T \left(\theta^{(t)}_{(j)} - \mu_{\theta_{(j)}} \right) \left(\theta^{(t)}_{(j)} - \mu_{\theta_{(j)}} \right)^\T,\quad j=1,\ldots,K,
\end{align}
which reduces to a version of the Double-Parallel Monte Carlo algorithm.

Many existing combination algorithms are special cases of \eqref{eq:9}
 with different choices of $(\tilde \theta, \tilde \Sigma)$. Let  $I_{d }$ be a $d \times d$ identity matrix and $A^{1/2}$ be the square root of a symmetric positive definite matrix $A$. Then, three candidates for $(\tilde \theta, \tilde \Sigma)$ are 
\begin{align}
  \label{eq:8}
  (\overline \mu, I_{d \times d}), \quad (\overline \mu, \check \Sigma), \quad (\overline \mu, \overline \Sigma), \quad \check \Sigma^{1/2} = \frac{1}{K} \sum_{j=1}^K  \Sigma^{1/2}_{j}, \quad \overline \Sigma = \frac{1}{K} \sum_{j=1}^K  \Sigma_{j}.
\end{align}
They lead to three combination algorithms COMB$(\overline \mu, I_{d\times d})$, COMB$(\overline \mu, \check \Sigma)$, and COMB$(\overline \mu, \overline \Sigma)$, where the first two have been proposed in \citet{xue2019double} and \citet{XuSri19}, respectively. In Section \ref{sec:4} of experiments, we draw $\theta$ from the block filtered posterior distribution using COMB$(\hat{\theta},\overline{\Sigma})$
in an HMM with Gaussian emission densities, where $\hat \theta$ is the maximum likelihood estimator using the full data.

\section{Theoretical Properties}
\label{theory}

\subsection{Setup}
\label{sec:setup}

We introduce some notation required for stating the assumptions of our theoretical setup. {Given $n = mK$ observations $Y_{1}^{n}\in \mathcal{Y}^{n}$, the first subset log-likelihood is $L_m(\theta) = \log p_{\theta}(Y_1^m)$ and the $j$th subset log-likelihood is $L_{jm}(\theta) = \log p_{\theta}(Y_{(j-1)m +1}^{jm})$ for $\theta \in \Theta$, $j = 2,\ldots, K$}, where $p_{\theta}$ is defined in \eqref{eq:3} in terms of $r_{\theta}, Q_{\theta}$, and $g_{\theta}$. Let $\theta_0 \in \Theta$ be the true parameter value, $P^{(n)}_{\theta_0}$ be the law of $Y_1^n$, $\{P^{(n)}_{\theta}:\theta \in \Theta\}$ be the family of probability measures on $\mathcal{Y}^n$ such that $P^{(n)}_{\theta}$ has density $p_{\theta}$ with respect to $\mu^n$, and $\Pi$ be the prior distribution on $\Theta$ with density $\pi(\theta)$ with respect to the {Lebesgue measure on $\RR^{d}$}. Denote the Euclidean norm as $\| \cdot \|_2$, $P^{(n)}_{\theta_0}$ as $\PP_0$, $P^{(n)}_{\theta}$ as $\PP_\theta$, expectation with respect to $\PP_{0}$ and $\PP_{\theta}$ as $\EE_0$ and $\EE_{\theta}$, respectively, and the first, second, and third derivative of $L_{jm}(\theta)$ with respect to $\theta$ as $L'_{jm}(\theta)$, $L''_{jm}(\theta)$, and $L'''_{jm}(\theta)$, respectively, for $j=1, \ldots, K$. 

The theoretical guarantees are based on the following assumptions on a family of HMM models $\{P^{(n)}_{\theta},\theta \in \Theta\}$ and the prior distribution $\Pi$. 
\begin{enumerate}[label=(\subscript{A}{{\arabic*}})]
\item \label{a1} (Existence of an envelope function.) There exists a constant $\delta_0 >0$, such that for any $m \geq 1$, $L_{m}(\theta)$ is three times differentiable with respect to $\theta$  in a neighborhood $B_{\delta_0}(\theta_0) = \{\theta \in \Theta : \| \theta - \theta_0 \|_2 \leq \delta_0 \}$.
For any $(z_1, \ldots, z_m) \in \mathcal{Y}^m$, there exists an envelope function $M(z_1, \ldots, z_m)$  that satisfies
  \begin{align*}
    \underset{\theta \in B_{\delta_0}(\theta_0)}{\sup} \left| \frac{d \log p_{\theta}(z_1, \ldots, z_m)}{d \theta_{l_1}} \right| &\leq M(z_1, \ldots, z_m), 
\\
    \underset{\theta \in B_{\delta_0}(\theta_0)}{\sup} \left| \frac{d^2 \log p_{\theta}(z_1, \ldots, z_m)}{d \theta_{l_1} d \theta_{l_2} } \right| &\leq M(z_1, \ldots, z_m), \\    
    \underset{\theta \in B_{\delta_0}(\theta_0)}{\sup} \left| \frac{d^3 \log p_{\theta}(z_1, \ldots, z_m)}{d \theta_{l_1} d \theta_{l_2} d \theta_{l_3}} \right| &\leq M(z_1, \ldots, z_m), \quad l_1, l_2, l_3 = 1, \ldots, d,
  \end{align*}
and $\frac{1}{m}M(Y_{1},\ldots,Y_{m}) \rightarrow 
\mathcal{C}_{M}$
$\PP_0$-almost surely as $m\rightarrow \infty$ for a constant $\mathcal{C}_{M} >0.$

\item \label{a2} (Law of large numbers for the log-likelihood function.) For any $\theta \in \Theta$ and $j \in \{1, \ldots, K\}$, the normalized $j$th subset log-likelihood $m^{-1}L_{jm}(\theta) \rightarrow L(\theta)$ $\PP_{0}$-almost surely as $m \rightarrow \infty$, where $L(\theta)$ is a continuous deterministic function with a unique global maximum at $\theta_{0}$. The convergence is uniform over any compact subsets of $\Theta$.

\item \label{a3} (Central limit theorem for the score function at the true parameter.) For any $j \in \{1, \ldots, K\}$, the normalized score function evaluated at $\theta_0$ on $j$th subset, $ m^{-1/2} L_{jm}'(\theta_{0})$,  $\PP_{0}$-weakly converges to $N_{d}(0, I_{0})$ as $m \rightarrow \infty$, where the $I_{0}$ is non-singular and defined as the Fisher information matrix of the model evaluated at $\theta_0$. 
\item \label{a4} (Law of large numbers for the observed information matrix.) If $\theta_{m}^{\ast}$ is any stochastic sequence in $\Theta$ such that $\theta_{m}^{\ast} \rightarrow \theta_{0}$ $\PP_0$-almost surely as $m \rightarrow \infty$, then $-m^{-1} L_{jm}''(\theta^{\ast}_{m}) \rightarrow I_{0}$ in $\PP_0$-probability as $m \rightarrow \infty$. Furthermore, assume that
  $-m^{-1} L_{jm}''(\theta) $ is positive definite with eigenvalues uniformly bounded from below and above by constants for all $\theta\in B_{\delta_0}(\theta_0)$, every $j=1, \ldots, K$, all values of $Y_{(j-1)m + 1}^{jm} \in \Ycal^m$, and sufficiently large $m$.
\item \label{a5} The parameter space $\Theta \subset \RR^{d}$ is compact with $\text{diam}(\Theta) := \underset{\theta_1,\theta_2 \in \Theta}{\sup} \|\theta_1 - \theta_2\|_{2} < \infty$ and $\theta_{0}$ as an interior point.
\item \label{a6} The prior density $\pi(\theta)$ is positive and continuous for all $\theta \in \Theta.$
\item \label{a7} For $\theta \in \Theta$, the transition probability matrix $\{Q_{\theta}(a,b)\}$ $(a, b \in \{1, \ldots, S\})$ is ergodic.
\end{enumerate}

Our regularity assumptions are commonly used for proving Bernstein-von Mises theorem in parametric models. Assumptions \ref{a2}--\ref{a6} are identical to those used in Theorem 2.1 of \citet{de2008asymptotic} if we set $m=n$ with $K=1$. Assumption \ref{a1} is slightly stronger than the requirement of continuous second derivatives of the log-likelihood function in \citet{de2008asymptotic}. We require the existence of an envelope function to control the growth of the remained term in the third order Taylor expansion of $L_{jm}(\theta)$  for every $j$ after stochastic approximation. For the same reasons, \citet{li2017simple} employ an assumption that is equivalent to Assumption \ref{a1}. Assumptions \ref{a5}--\ref{a6} ensure that the prior density $\pi(\theta)$ is well-behaved on compact set $\Theta$.  
Assumptions \ref{a1}--\ref{a6} are the extensions of assumptions required for proving the Bernstein-von Mises theorem for independent and identically distributed data \citep[Theorem 8.2, Chapter 6]{LehCas98}. Assumption \ref{a7} implies that for any $\theta \in \Theta$, there exists a positive integer $v\geq 1$ such that every element of the $v$-step transition matrix $Q_{\theta}^{v}$ is positive. We assume that $v = 1$ and $\epsilon \vcentcolon = \inf_{a,b\in \mathcal{X}, \theta \in \Theta} Q_{\theta}(a,b) > 0$. The ergodicity of $\{X_t\}$ in Assumption \ref{a7} implies that $\{Y_t\}$ is stationary and ergodic under $\PP_\theta$ \citep[Lemma 1]{leroux1992maximum}. Assumptions that are equivalent to \ref{a1}--\ref{a7} are also used for proving the asymptotic normality of the maximum likelihood estimator of $\theta$ in HMMs \citep{bickel1998asymptotic}. 

The definition of $j$th subset $(j = 2,\ldots,K)$ posterior distribution in \eqref{eq:s6} uses the $j$th quasi-likelihood with approximate prediction filter conditioning only on the $(j-1)$th subset. Due to this approximation, we need the following assumptions to gaurantee the consistency of the $j$th subset posterior  distribution defined using the $j$th quasi-likelihood. 
\begin{enumerate}[label=(\subscript{A}{{\arabic*}}),resume]
\item
\label{a8} 
\begin{enumerate}[label = (\roman*)]
\item 
\label{b}
For $\theta \in  \Theta$
\begin{align*}
  \max_{a \in \mathcal{X}} \int_{\mathcal{Y}} \left\{\max_{b \in \mathcal{X}}|\log g_{\theta}(y|b)| \right\} g_{\theta_0}(y|a)\mu(dy) <\infty.
\end{align*}
\item  
\label{c}
Let $h_{\theta}(y) =  \max_{a,b \in \mathcal{X}} \frac{g_{\theta}(y|a)}{g_{\theta}(y|b)} \geq 1.$ 
There exists a constant $M\geq 1$ such that $\PP_0$-almost surely,
$$
\sup_{\theta \in \Theta}  h_{\theta}(Y_{1}) <  M.$$
\item 
\label{a}(Convergence of prediction filter.)
For $\theta \in  \Theta$,
define
$
\mu_{\theta}(y) = \{ 1 + (S-1)\epsilon^{-2} h_{\theta}(y)\}^{-1}
$
such that $\PP_0$-almost surely,
$$\|p_{\theta}(X_{m+1}|Y^{m}_{1}) - r_{\theta}(X_{m+1})\|_{\text{TV}} \leq 
S \prod_{i= 2}^{m-1}\exp\{-2\mu_{\theta}(Y_{i})\}\leq 
S \rho^{m-2},$$
where the mixing coefficient
\begin{align}
\rho  \vcentcolon = e^{-\tfrac{2}{1 + (S-1)\epsilon^{-2}M }}
\label{eq:mixc}
\end{align}
and $\|d\PP - d\QQ \|_{\text{TV}}$ is the total variation distance between measures $\PP$ and $\QQ$.

\end{enumerate}
\end{enumerate}
Assumptions \ref{a8}\ref{b}--\ref{a} are based on \citet[Theorem 2.1]{le2000exponential} and are global over $\Theta$. Among them, \ref{a8}\ref{c} is slighly stronger as we require $h_{\theta}(Y_{1})$ to be bounded $\PP_0$-almost surely. Together they imply that the prediction filters, hence the subset log-likelihood functions, are forgetting the condition of initial distribution exponentially fast. Assumption \ref{a8} \ref{a} implies that the prediction filter $ p_{\theta}(X_{m+1} \mid Y^{m}_{1})$ converges to $r_{\theta}(X_{m+1})$ geometrically with mixing coefficient $\rho$, where the martingale convergence of $p_{\theta}(X_{m+1} \mid Y^{m}_{1})$ and functions $\mu_{\theta}(Y_{i})$ are discussed in \citet[Lemma 5]{bickel1998asymptotic}. Our assumption is stronger than \citet[Lemma 5]{bickel1998asymptotic} in that we assume the stationary distribution $r_{\theta}$ is the martingale limit of the prediction filter $p_{\theta}(X_{m+1}\mid Y_{1}^{m})$ as $m$ tends to infinity. Assumption \ref{a8} is global over $\Theta$, which is necessary to show that the subset posterior distribution using quasi-likelihood provides accurate estimation. Under this assumption, the following proposition shows that the $j$th approximate prediction filter converges to the true prediction filter ($j=2,\ldots,K$), which implies that the quasi-likelihood in \eqref{eq:5.5} provides an accurate approximation of the true likelihood in \eqref{eq:5} as $m$ tends to infinity.
\begin{proposition}
  \label{filt-approx}
  If Assumption \ref{a8} holds, then for $j=2, \ldots, K$ and $\theta \in \Theta$, as $m \rightarrow \infty$,
  $$ \| p_{\theta}(X_{(j-1)m+1} \mid Y_{(j-2)m+1}^{(j-1)m}) - p_{\theta}(X_{(j-1)m+1} \mid Y_{1}^{(j-1)m})\|_{\text{TV}} \rightarrow 0$$   
  $\PP_0$-almost surely. 
\end{proposition}

\subsection{Main Results}
\label{sec:main-results}

We now present two theorems stating theoretical properties of the subset and block filtered posterior distributions. First, the following theorem specifies the asymptotic order of $K$ depending on the mixing coefficient $\rho$ and other regularity assumptions such that the limiting distributions of the $K$ subset posterior  distributions are normal as $m$ and $K$ tend to infinity.
\begin{theorem}\label{thm1}
Let the number of subsets $K$ be a sequence $K_{m} \rightarrow \infty$ as $m \rightarrow \infty$ and $n = mK_{m}$. For $j = 1, \ldots, K_{m}$, denote the maximum likelihood estimator of $\theta$ as  $\hat{\theta}_{j}$ that solves $L'_{jm}(\theta) = 0$, define the local parameter $h_j =  \sqrt{n}(\theta - \hat{\theta}_{j})$, and denote the distribution of $h_j$ implied by the $j$th subset  posterior density in \eqref{eq:s6} as $\Pi_{m}^{h_j}(\cdot \mid Y_{(j-1)m+1}^{jm})$. If Assumptions \ref{a1}--\ref{a7} hold, then 
\begin{enumerate}[label=$(\roman*)$]
\item
\label{thm1:1}
as $m \rightarrow \infty$, 
  \begin{align}
    \| \Pi_{m}^{h_1}(\cdot \mid Y_{1}^{m}) - N_d(0, I_0^{-1}) \|_{\text{TV}} \rightarrow 0 \text{ in } \PP_0\text{-probability };
 \label{eq:sub1}
  \end{align}
\item
\label{thm1:2}
  If Assumption \ref{a8} also holds and $K_{m} = o(\rho^{-m})$ for $\rho $ defined in \eqref{eq:mixc}, then for $j=2,\ldots,K_{m}$, as $m \rightarrow \infty$,
  \begin{align}
    \| \Pi_{m}^{h_j}(\cdot \mid Y_{(j-1)m+1}^{jm}) - N_d(0, I_0^{-1}) \|_{\text{TV}} \rightarrow 0 \text{ in } \PP_0\text{-probability},
\label{eq:sub2-subK}
  \end{align}
\end{enumerate}
  where $N_d(0, I_{0}^{-1})$ is a $d$-variate normal distribution with zero mean and the inverse of Fisher information matrix $I_0^{-1}$ as the covariance matrix.
\end{theorem}

Theorem \ref{thm1} shows that the credible intervals and confidence intervals based on the $j$th subset posterior density \eqref{eq:s6} and maximum likelihood estimator on $j$th subset, respectively, are asymptotically equivalent ($j=1, \ldots, K$). We also establish that $K = o(\rho^{-m})$ is rate-optimal in that the limiting distribution of subset posterior distributions are normal. This result is required to prove that the combined posterior distribution provides an accurate approximation to the intractable true posterior distribution asymptotically. The upper bound on the growth of $K$ immediately yields a lower bound on the growth of the subset size because $m = n/K$. For weakly dependent HMMs with $\rho \ll 1$, the number of $K$ can be chosen to be large and $m$ will be relatively small. On the other hand, for strongly dependent HMMs with $\rho \approx 1$, $m$ has to be large to guarantee accurate estimation of the subset posterior  distributions, which results in a relatively small $K$. 

 Let $\tilde{\Pi}_{m}^{\theta}(\cdot\mid Y_{(j-1)m+1}^{jm})$ denote the $j$th subset posterior  distribution of $\theta $ with density defined in \eqref{eq:s6}. Then, the invariance of the total variation distance implies that as $m \rightarrow \infty$,
\begin{align}
  \label{eq:sub:bvm}
  \|\tilde{\Pi}_{m}^{\theta}(\cdot\mid Y_{(j-1)m+1}^{jm}) - N_{d}(\hat{\theta}_{j},(nI_{0})^{-1})\|_{TV} \rightarrow 0\text{ in } \PP_0\text{-probability for } j\geq 1.
\end{align}
If $K = 1$, Theorem \ref{thm1} recovers the usual Bernstein-von Mises theorem for the true posterior distribution \citep{de2008asymptotic}. Comparing the limiting distributions of subset posterior  distribution and true posterior distribution, we observe that they are both asymptotically $d$-variate normal distributions with the same covariance matrix $(nI_{0})^{-1}$, but with different maximum likelihood estimator  as mean parameters. Based on this observation, the proposed combined posterior distribution, \textit{block filtered posterior} distribution, is a mixture model of  subset posterior distribution after appropriate re-scaling and re-centering.

Our next result establishes conditions under which the block filtered posterior distribution converges to the true posterior distribution at an optimal parametric rate as both $m$ and $K$ tend to infinity. The first step in this process is to show the asymptotic normality of the block filtered posterior distribution. To this end, we need the following assumption on the re-scaling matrix $\tilde \Sigma$ and subset posterior covariance matrices. 
\begin{enumerate}[label=(\subscript{A}{{\arabic*}}),resume]
\item  Let $\Sigma_{j}$ be the covariance matrix of $j$th subset posterior distribution of $\theta$ ($j=1, \ldots, K$) and the rescaling matrix $\tilde \Sigma$ be a deterministic positive symmetric matrix. As $m \rightarrow \infty$,
  \begin{align}
    \label{eq:a9}
 \| (n\tilde{\Sigma})^{-1/2} - I_{0}^{1/2}\|_{\text{F}}^{2} \rightarrow 0,\quad 
\EE_0 \| (n\Sigma_{j})^{1/2} - I_0^{-1/2}\|_{\text{F}}^{2} \rightarrow 0,
  \end{align}
uniformly for $j=1,\ldots,K$,   where $\|\cdot\|_{\text{F}}$ is the Frobenius norm.
  \label{a9}
\end{enumerate}
Assumption \ref{a9} implies that $(n\Sigma_{j})^{1/2}$ is close to $I_0^{-1/2}$ and that the square Frobenius norm of their difference converges to zero in $\EE_0$-expectation. Furthermore, $(n\tilde \Sigma)^{-1/2}$ is close to $I_0^{1/2}$ and the square Frobenius norm of their difference converges to zero. This ensures that the limiting covariance matrices of block filtered posterior and true posterior distributions coincide. 
For two integrable measures $\nu_1,\nu_2$ on $\Theta$, the 1-Wasserstein distance is defined as $W_{1}(\nu_1,\nu_2) = \sup_{\|f\|_{\text{Lip}}\leq 1}\left\{ \int_{\Theta}f(\theta)d(\nu_1-\nu_2)(\theta),f:\Theta \rightarrow \mathbb{R} \text{ is continuous}\right\}.$
If we choose $\tilde{\Sigma}$ satisfying \eqref{eq:a9}, the following theorem states that the block filtered posterior distribution is asymptotically normal and
the 1-Wasserstein distance between the block filtered and true posterior distributions converges to $0$ in $\PP_0$-probability, with rate determined by the closeness between the re-centering vector and the maximum likelihood estimator of $\theta$.
\begin{theorem}
  \label{thm2}
  Let $m, K_{m}$ satisfy the conditions in Theorem \ref{thm1} and $n = m K_{m}$, $\tilde{\Pi}_{n}^{\theta}(\cdot \mid Y_{1}^{n})$ be the block filtered posterior distribution of $\theta$ given $Y_{1}^{n}$ with density defined in \eqref{eq:9} based on the combination scheme COMB$(\tilde{\theta},\tilde{\Sigma})$, and $\tilde{\Pi}_n^{h}(\cdot \mid Y_{1}^{n})$ be the distribution of $h = \sqrt{n}(\theta - \tilde{\theta})$ implied by $\tilde{\Pi}_{n}^\theta(\cdot \mid Y_{1}^{n})$. If Assumptions \ref{a1}--\ref{a9} hold, then 
\begin{enumerate}[label=$(\roman*)$]
\item
\label{thm2:1}
as $m \rightarrow \infty$,
  \begin{align}
    \| \tilde{\Pi}_{n}^{h}(\cdot \mid Y_{1}^{n}) - N_{d}(0, I_0^{-1})\|_{\text{TV}} \rightarrow 0 \text{ in } \PP_0\text{-probability}.
 \label{eq:comb1}
  \end{align}

\item
\label{thm2:2}
  Moreover, let $\hat{\theta}$ be the maximum likelihood estimator of $\theta$ that solves $L_{n}'(\theta) = 0$ and $\Pi_{n}^{\theta}(\cdot \mid Y_{1}^{n})$ be the true posterior distribution of $\theta$ given $Y_{1}^{n}$. Then, up to a negligible term in $\PP_0$-probability, 
  \begin{align}
    \text{diam}(\Theta)^{1/2} \sqrt{\Delta}
 \leq   \sqrt{n}W_{1}\{\tilde{\Pi}_{n}^{\theta}(\cdot \mid Y_{1}^{n}) , \Pi_{n}^{\theta}(\cdot \mid Y_{1}^{n}) \} \leq \Delta  ,
    \label{eq:comb2}
  \end{align}
  where $\Delta = \sqrt{n}\|\tilde{\theta} - \hat{\theta}\|_2$.
\end{enumerate}
\end{theorem}

Theorem \ref{thm2}\ref{thm2:1} shows that the block filtered posterior distribution defined by COMB$(\tilde \theta, \tilde{\Sigma})$ is asymptotically normal with mean parameter $\tilde \theta$ and covariance matrix $(nI_0)^{-1}$; Theorem \ref{thm2}\ref{thm2:2} specifies the approximation accuracy of block filtered posterior to the true posterior distributions, which is determined by the asymptotic order of $\Delta= \sqrt{n}\|\tilde{\theta} - \hat{\theta}\|_2$. If $\Delta = o_{\PP_0}(1)$, the credible regions obtained using the block filtered and true posterior distributions match up to $o(n^{-1/2}) $ in $\PP_0$-probability as $m, K$ tend to infinity, and the $W_1$ distance between them converges to zero in $\PP_0$-probability with optimal parametric rate $n^{1/2}$. The maximum likelihood estimator $\hat \theta$ can be efficiently obtained using online EM \citep{cappe2011online} or forward-backward recursions \citep{Capetal06}. The block filtered posterior distribution defined by COMB$(\hat \theta, \tilde{\Sigma})$ approximates true posterior distribution at an optimal rate. Similar optimal consistency results of block filtered posterior distribution hold if we replace $\hat \theta$ by any root-$n$ consistent estimator of $\theta_0$, which includes the maximum split data likelihood estimator of \citet{Ryd94}. If estimating a root-$n$ consistent estimator of $\theta_0$ is time consuming, then by setting $\tilde{\theta} = \hat{\theta}_{j}$ we have $\Delta =  o_{\PP_0}(\sqrt{\frac{n}{m}})$. In this case, Theorem \ref{thm2}\ref{thm2:2} shows that the $W_1$ distance between the block filtered posterior distribution and true posterior distribution converges to zero in $\PP_0$-probability with sub-optimal rate $m^{1/2}$.

Consider the degenerate case in which elements of the hidden chain $\{X_{t}\}_{t=1}^{n}$ are independent. The complete data $\{X_{t},Y_{t}\}_{t=1}^n$ are samples from a mixture model with $S$ components, where $X_{t}$ is the mixture component membership at time $t$. The total variance distance between prediction filters and initial distributions is exactly zero, which corresponds to a degenerate case of $\rho = 0$ in Assumption \ref{a8}. In this case, the proof of Theorem \ref{thm1} implies that the subset posterior inference using quasi-likelihood with prediction filter is equivalent to that using quasi-likelihood with initial distribution, and this equivalence removes the rate-optimal upper bound $K = o(\rho^{-m})$ in the dependent case, while having the same theoretical gaurantess of block filtered posterior distribution. Furthermore, Theorem \ref{thm2} recovers the existing results for divide-and-conquer Bayesian inference of independent data, where the number of subsets $K$ and the size of subsets $m$ tend to infinity. For the degenerate HMM, our theoretical setup of Theorem \ref{thm2} imposes no restriction on the growth of $K$ and $m$.  This includes the case where $K = O(n^{c})$ and $m = O(n^{1-c})$ for any $c \in (0,1)$ as specified in Theorem 2 of \citet{li2017simple}.

\section{Experiments}
\label{sec:4}

\subsection{Setup}

We compare the performance of block filtered posterior distribution with existing divide-and-conquer approaches to Bayesian inference. The true posterior distribution of $\theta$ serves as the benchmark in every comparison. It is estimated using the data augmentation algorithm based on forward-backward recursions. We select Double Parallel Monte Carlo \citep{xue2019double}, Posterior Interval Estimation \citep{li2017simple}, and Wasserstein Posterior \citep{srivastava2018scalable} as our divide-and-conquer competitors. There is no software support for using Consensus Monte Carlo, so we excluded it in our comparisons \citep{Scoetal16}. Following Theorem \ref{thm2}, {the combination steps are specified by $COMB(\hat{\theta},\overline{\Sigma})$ where the centering parameter $\hat \theta$ is the maximum likelihood estimator and is estimated using the Baum-Welch EM algorithm and the scaling matrix $\overline{\Sigma}$ is matrix barycenter defined in \eqref{eq:8}.} All the sampling algorithms run for 10,000 iterations. The first 5000 samples are discarded as burn-in and the remaining chain is thinned by collecting every fifth sample. We have implemented all the algorithms in R and used an Oracle Grid Engine cluster with 2.6GHz 16 core compute nodes for all our experiments. 

The accuracy metric for comparing the {true posterior distribution and its approximation} is based on the total variation distance. Let $\xi$ be a one-dimensional parameter, $\pi(\xi \mid Y_1^n)$ be the true posterior density of $\xi$, and $\hat \pi(\xi \mid Y_1^n)$ be the density of its approximation. Following \citet{li2017simple}, the approximation accuracy of $\hat \pi(\xi \mid Y_1^n)$ is 
\begin{align}
  \text{Acc} \left\{ \hat \pi (\xi \mid  Y^n_1)\right\} = 1 - \frac{1}{2} \int_{\RR} \left|\hat \pi(\xi \mid Y_1^n) - \pi(\xi \mid  Y^n_1)\right| d \xi. \label{acc}
\end{align}
The integral in \eqref{acc} is the total variation distance between $\hat \pi(\xi \mid Y_1^n)$ and $\pi(\xi \mid  Y^n_1)$. It is approximated numerically using the $\xi$ draws from $\hat \pi (\xi\mid  Y^n_1)$ and $\pi (\xi \mid  Y^n_1)$, respectively, {and kernel density estimation.} For a $\theta = (\xi_1, \ldots, \xi_d)$, we extend \eqref{acc} as the median of $\text{Acc} \left\{ \hat \pi (\xi_i \mid  Y^n_1)\right\}$ ($i=1, \ldots d$). If this metric is high, then $\hat \pi(\theta \mid Y_1^n)$ is an accurate estimator of $\pi(\theta \mid Y_1^n)$.

\subsection{Simulated Data Analysis}
\label{sim-dat}

Our simulation is based on an HMM with Gaussian emission densities, which has been used in \citet{Ryd08}. Using the notation from Section \ref{setup-back}, let $\Ycal = \RR$, $g(\cdot \mid X = a)=\phi(\cdot;\mu_{a},\sigma_{a}^{2})$ be the density of Normal distribution with mean $\mu_a$ and variance $\sigma^2_a$ $(a=1, \ldots, S)$, $\theta = \{r_1,\ldots,r_{S}, Q_{11},\ldots,Q_{SS}, \mu_1, \ldots, \mu_S, \sigma^2_1, \ldots, \sigma^2_S\}$, and $d = S^2 + 2S - 1$. 

The prior distribution of $\theta$ is conjugate to the likelihood $p_{\theta} (X_{1}^{n},Y_{1}^{n})$ in \eqref{eq:3}. The initial distribution $r = (r_1,\ldots,r_{S})$ and rows of the Markov transition matrix $\{(q_{a1},\ldots,q_{aS})\}_{a=1}^{S}$ are given an independent Dirichlet prior Dir$(1,\ldots,1)$; the mean parameters $\mu_{1},\ldots,\mu_{S}$ are given an independent Normal prior $N(\xi,\kappa^{-1})$ with $\xi = (\min y_{i} + \max y_{i})/2 $ and $\kappa = 1/(\max y_{i} - \min y_{i})^{2}$; the variance parameters $\sigma^{-2}_{1},\ldots,\sigma^{-2}_{S}$ are given an independent gamma prior with parameters $(1,1)$. The division and subset posterior computation steps for  
the block filtered posterior distribution  and its divide-and-conquer  competitors are based on Sections \ref{div-step} and \ref{samp-step}. This example slightly violates our Assumption \ref{a8} \ref{c} but is particularly attractive because posterior computations on a subset using  \eqref{eq:s6} are analytically tractable; see Appendix \ref{illus-ex} for greater details.

\emph{Simulation Study of  Accuracy and Efficiency.} This simulation study is aimed to analyze the performance of block filtered posterior and compare it with divide-and-conquer competitors. We set $S=3$, $\mu_1 = -2$, $\mu_2 = 0$, $\mu_3 = 2$, and $\sigma_1 = \sigma_2 = \sigma_3 = 0.5$. Following \citet{Ryd08}, $Q$ is defined as $Q_{11} = 0.6$, $Q_{12} = 0.3$, $Q_{13} = 0.1$, $Q_{21} = 0.1$, $Q_{22} = 0.8$, $Q_{23} = 0.1$, $Q_{31} = 0.1$, $Q_{32} = 0.3$, and $Q_{33} = 0.6$. This choice of $Q$ also determines {the stationary distribution $r = (0.2,0.6,0.2)^\T$} because $r^\T Q = r^\T$. We replicate this simulation ten times.

We use three choices of $K$ and $n$, {respectively}. For a given $n$, $K$ is varied as $\log n$, $n^{1/4}$, and $n^{1/3}$, where the first and last two choices of $K$ have been used to justify the asymptotic properties of Wasserstein Posterior and Posterior Interval Estimation, respectively. We vary $n$ as $10^4, 10^5$, and $10^6$. For a given $K$ and $n$, the subset sample size is defined as {$m=\lceil n/K \rceil$, where $\lceil x \rceil$ denotes the smallest integer $z$ such that $z \geq x$.} If $K= \log n$, then $m$ is large, which implies that for $n = 10^{6}$, all the divide-and-conquer methods are slow due to a large $m$; on the other hand, they are most efficient when $K=n^{1/3}$. We also benchmark the performance of divide-and-conquer approaches relative to Hamiltonian Monte Carlo, implementing it in RStan \citep{Stan17}. 

The accuracy based on \eqref{acc} of all the five methods are compared for posterior inference on $(\mu_1, \mu_2, \mu_3, \sigma_1, \sigma_2, \sigma_3)$ (Tables \ref{tab:par-sim}). The divide-and-conquer competitors of the block filtered posterior have low accuracy for $n=10^5,10^6$ and any $K$. In comparison, the block filtered posterior distribution is the top performer for every choice of $K$ and $n$ and its accuracy is close to that of the Hamiltonian Monte Carlo. When $n$ is large the block filtered posterior is almost ten times faster than Hamiltonian Monte Carlo and has better accuracy compared to divide-and-conquer competitors for all parameters. This shows that the block filtered posterior is an accurate and efficient algorithm for inference in HMMs under massive data settings; see Appendix \ref{app:sim} for timing comparisons and accuracy for posterior inference on $Q$. 

\begin{table}[ht]
\let\TPToverlap=\TPTrlap
\centering
\caption{\it Accuracy of the approximate posterior distributions of $(\mu_1, \mu_2, \mu_3, \sigma_1, \sigma_2, \sigma_3)$. The accuracy values are averaged over ten simulation replications and across all the parameter dimensions. The maximum Monte Carlo error for block filtered posterior distribution is 0.02}
\label{tab:par-sim}
\begin{threeparttable}
  \begin{tabular}{|l|l|c|c|c|c|c|}
    \hline
    $n$ & $K$ & BFP & WASP & DPMC & PIE & HMC \\ 
    \hline

    & $\log n$ & 0.93 & 0.48 & 0.48 & 0.49 &  \\ 
        $10^4$ & $n^{1/4}$ & 0.93 & 0.48 & 0.48 & 0.48 & 0.96 \\ 
    & $n^{1/3}$ & 0.92 & 0.37 & 0.36 & 0.37 & \\ 
    \hline

    & $\log n$ & 0.93 & 0.18 & 0.18 & 0.18 &  \\ 
    $10^5$ & $n^{1/4}$ & 0.93 & 0.18 & 0.18 & 0.18 & 0.95\\ 
    & $n^{1/3}$ & 0.92 & 0.18 & 0.18 & 0.18 & \\ 
    \hline

        & $\log n$ & 0.93 & 0.14 & 0.14 & 0.15 & \\ 
    $10^6$ & $n^{1/4}$ & 0.93 & 0.12 & 0.12 & 0.13 & 0.96 \\ 
    & $n^{1/3}$ & 0.93 & 0.14 & 0.14 & 0.14 & \\ 
    \hline
  \end{tabular}
\begin{tablenotes}[normal,flushleft]
\small
\item[1]
 BFP, block filtered posterior distribution; WASP, Wasserstein posterior; DPMC, double parallel Monte Carlo; PIE, posterior interval estimation; HMC, Hamiltonian Monte Carlo.
\end{tablenotes}
\end{threeparttable}
\end{table}

  
\emph{Simulation Study of Mixing Property.} This simulation analyzes the impact of mixing coefficient $\rho$ on the choice of $K$ and accuracy of the block filtered posterior distribution. In the previous simulation study, we have $(S=3,Q_{\epsilon} = 0.1,  \mu_1 = -2, \mu_2 = 0, \mu_3 = 2)$. We consider three more HMMs with increasing value of mixing coefficients $\rho$ in this simulation: $(S=2,Q_{\epsilon} = 0.3,\mu_1 = -2, \mu_2 = 2)$, $(S=5,Q_{\epsilon} = 0.05,\mu_1 = -4, \mu_2 = -2, \mu_3 = 0, \mu_4 = 2, \mu_5 = 4)$, and $(S=7,Q_{\epsilon} = 0.05,\mu_1 = -8, \mu_2 = -4, \mu_3 = -2, \mu_4 = 0, \mu_5 = 2, \mu_6 = 4, \mu_7 = 8)$. For every HMM, we set $\sigma_{a} = 0.5$ for $1\leq a \leq S$ and vary $n$ as $10^4$ and $10^5$. For a given $n$, $K$ is varied as $\log n$, $n^{1/4}$, and $n^{1/3}$. We replicate this simulation ten times and benchmark the performance of divide-and-conquer approaches relative to Hamiltonian Monte Carlo.

The accuracy of block filtered posterior distribution for inference on the emission distribution parameters follows from the results in Theorem \ref{thm1} (Tables \ref{tab:par-sim} and \ref{tab:rho-sim}). The dependence in HMMs increases with $\rho$. The HMMs with $S=5, 7$ have $\rho \approx 1$, resulting in a relatively higher dependence than the other two HMMs with $S=2,3$. Our theoretical results imply that the accuracy of the block filtered posterior is fairly insensitive to the choice of $K$ in the latter two HMMs than the former two. Indeed, Hamiltonian Monte Carlo and block filtered posterior have similar accuracy for every $(n,K)$ combinations when $S=2,3$. On the other hand, the accuracy of block filtered posterior is smaller than that of Hamiltonian Monte Carlo and decreases with increasing $K$ when $S=5,7$ due to the increased dependence. Our theoretical results also imply that the accuracy depends only on $\rho$ and not on $S$. This is further confirmed by our empirical results in that the accuracy values for the HMMs with $S=5,7$ are very similar due to almost equal $\rho$ values.

The block filtered posterior distribution outperforms its divide-and-conquer competitors for every $\rho$ and $K$. The competing divide-and-conquer methods are developed for independent data, so their accuracy values are relatively large when $S=2$, $\rho$ is small, and the dependence is weak; however, their accuracy values quickly drop as dependence increases with $\rho$, reducing to 0 when $S = 5,7$. On the other hand, the block filtered posterior distribution is designed for dependent data and maintains its superior performance over the divide-and-conquer competitors for every $\rho$. Relative to its divide-and-conquer competitors,  the block filtered posterior distribution is also very robust to choice of $K$ and its accuracy is impacted only when the subset size is very small (for example, $K = O(n^{1/3})$ and $n=10^4$); see Appendix \ref{app:sim} for accuracy for posterior inference on $Q$. We conclude based on this simulation that the block filtered posterior distribution provides accurate posterior inference using the divide-and-conquer technique for a broad range of $\rho$ values.

\begin{table}[ht]
\let\TPToverlap=\TPTrlap
\centering
\caption{\it Accuracy of the approximate posterior distributions of $(\mu_1, \ldots, \mu_S, \sigma_1, \ldots, \sigma_S)$. The accuracy values are averaged over ten simulation replications and across all the parameter dimensions. The maximum Monte Carlo error for block filtered posterior distribution is 0.02.}
\label{tab:rho-sim}
\begin{threeparttable}
  \begin{tabular}{|l|l|c|c|c|c|c|}
    \hline
    $n$ & $K$ & BFP & WASP & DPMC & PIE & HMC \\ 
    \hline
    \multicolumn{7}{|c|}{$S = 2$}\\
    \hline
    & $\log n$ & 0.97 & 0.64 & 0.64 & 0.65 &  \\
    $10^4$ & $n^{1/4}$ & 0.97 & 0.64 & 0.64 & 0.64 & 0.95 \\
    & $n^{1/3}$ & 0.97 & 0.61 & 0.61 & 0.61 &  \\
    \hline
    & $\log n$ & 0.97 & 0.46 & 0.46 & 0.46 &  \\
    $10^5$& $n^{1/4}$ & 0.97 & 0.46 & 0.46 & 0.46 & 0.95 \\
    & $n^{1/3}$ & 0.97 & 0.45 & 0.45 & 0.45 &  \\
    \hline
    \multicolumn{7}{|c|}{$S = 5$}\\
    \hline
    & $\log n$ & 0.83 & 0.00 & 0.00 & 0.00 &  \\
    $10^4$ & $n^{1/4}$ & 0.82 & 0.00 & 0.00 & 0.00 &0.95  \\
    & $n^{1/3}$ & 0.82 & 0.00 & 0.00 & 0.00 &  \\
    \hline
    & $\log n$ & 0.82 & 0.00 & 0.00 & 0.00 &  \\
    $10^5$& $n^{1/4}$ & 0.79 & 0.00 & 0.00 & 0.00 &0.95  \\
    & $n^{1/3}$ & 0.80 & 0.00 & 0.00 & 0.00 &  \\
    \hline
    \multicolumn{7}{|c|}{$S = 7$}\\
    \hline
    & $\log n$ & 0.80 & 0.00 & 0.00 & 0.00 &  \\
    $10^4$ & $n^{1/4}$ & 0.77 & 0.00 & 0.00 & 0.00 & 0.93 \\
    & $n^{1/3}$ & 0.69 & 0.00 & 0.00 & 0.00 &  \\
    \hline
    & $\log n$ & 0.82 & 0.00 & 0.00 & 0.00 &  \\
    $10^5$& $n^{1/4}$ & 0.81 & 0.00 & 0.00 & 0.00 & 0.94 \\
    & $n^{1/3}$ & 0.81 & 0.00 & 0.00 & 0.00 &  \\
    \hline        
  \end{tabular}
\begin{tablenotes}[normal,flushleft]
\small
\item[]
 BFP, block filtered posterior distribution; WASP, Wasserstein posterior; DPMC, double parallel Monte Carlo; PIE, posterior interval estimation; HMC, Hamiltonian Monte Carlo.
\end{tablenotes}
\end{threeparttable}
\end{table}

\subsection{Real Data Analysis}
\label{real-dat}

We illustrate the application of block filtered posterior distribution in a realistic setting using the three-month US Treasury bill rates from Jan 4, 1954 to April 8, 2021. We have downloaded the daily observations from  \url{https://www.federalreserve.gov/datadownload/Build.aspx?rel=H15} and have removed the missing observations, resulting in $n=16,808$ time points. A subset of this data has been analyzed using a stochastic volatility model in \citet{golightly2006bayesian}. We de-trend the time series using Friedman's super smoother as implemented in \emph{supsmu} function in R and model it using an HMM with $S = 3$ that has been used earlier in the first simulation study. Replicating the simulation setup, for $K = \log n, n^{1/4}$, and $n^{1/3}$, we use data augmentation, Hamiltonian Monte Carlo, and the block filtered posterior distribution and its divide-and-conquer competitors for posterior inference in this model.

The results of real and simulated data analyses in this section and Section \ref{sim-dat} agree closely (Table \ref{tab:par-real}).  For all choices of $K$, the block filtered posterior has better accuracy than its divide-and-conquer competitors for inference on the emission distribution parameters for every $K$, and its accuracy is comparable to that of Hamiltonian Monte Carlo. The accuracy is better when $K = O(\log n)$ and $n^{1/4}$ because smaller choices of $K$ result in more accurate estimation of subset posterior distribution. The block filtered posterior distribution is more efficient than data augmentation for every $K$, and its efficiency increases with $K$; see Appendix \ref{app:sim} for comparisons of timing and accuracy of posterior inference on $Q$. These observations are very similar to that of our simulation study, so we conclude that the block filtered posterior distribution provides an accurate and efficient alternative to the data augmentation algorithm for principled Bayesian inference in the modeling three-month US Treasury bill rate data. 

\begin{table}[ht]
  \caption{\it Accuracy of the approximate posterior distributions of $(\mu_1, \mu_2, \mu_3, \sigma_1, \sigma_2, \sigma_3)$. The entries in the table are the median of the accuracy values computed across the six dimensions. The maximum Monte Carlo error for the block filtered posterior distribution is 0.01}
\label{tab:par-real}  
\centering
\begin{threeparttable}
{\centering
  \begin{tabular}{|l|c|c|c|c|c|}
    \hline
    $K$ & BFP & WASP & DPMC & PIE & HMC \\
    \hline

    $\log n$ & 0.86 & 0.18 & 0.17 & 0.17 & 0.91 \\
    $n^{1/4}$ & 0.85 & 0.14 & 0.13 & 0.13 &  \\
    $n^{1/3}$ & 0.80 & 0.15 & 0.13 & 0.12 &  \\
    \hline
  \end{tabular}
}
\begin{tablenotes}
\small
\item[] BFP, block filtered posterior distribution; WASP, Wasserstein posterior; DPMC, double parallel Monte Carlo; PIE, posterior interval estimation; HMC, Hamiltonian Monte Carlo.
\end{tablenotes}
\end{threeparttable}
\end{table}
       
\section{Discussion}
\label{ds}

Our focus has been on developing a divide-and-conquer Monte Carlo algorithm for Bayesian analysis of parametric HMMs where the state space is known and discrete; however, our Monte Carlo algorithm and its theoretical properties can be extended to general state space models, where the state space is continuous and known. This requires minimal changes in the partitioning, subset sampling, and combining steps. Our theoretical results can also be extended to general state space models using the ideas of \citet{jensen1999asymptotic}, who show that the theoretical results of \citet{bickel1998asymptotic} for general HMMs extend naturally to general state space models. 
\appendix

\section{Proofs of Proposition \ref{filt-approx} and Technical lemmas}
We begin with several important lemmas and the proof of Proposition \ref{filt-approx}.
\label{lem}
\subsection{Preliminaries}
Let $\Pcal(\Xcal)$ denote the space of all probability measure on $\Xcal$ equipped with total variation distance. For any $\theta \in \Theta$, the prediction filter $p_{t}^{\theta} \in \Pcal(\Xcal)$ is defined as
\begin{align}
\label{eq:26}
  p_{t}^{\theta} = \left\{ \PP_{\theta}(X_{t}= 1 \mid Y_{1}^{t-1}), \ldots, \PP_{\theta}(X_{t}= S \mid Y_{1}^{t-1}) \right\}^\T , \quad t = 1, \ldots, n.
\end{align}
The Baum's forward equation implies that for $t = 1, \ldots, n-1$,
\begin{align}
  p_{t+1}^{\theta}&= \frac{\PP_{\theta}(X_{t+1} = \cdot,Y_{t}\mid Y_{1}^{t-1})}{\PP_{\theta}(Y_{t}\mid Y_{1}^{t-1})} = \frac{Q_{\theta}^{\T}G_{\theta}(Y_{t})p_{t}^{\theta}}{g_{\theta}(Y_{t})^{\T}p_{t}^{\theta}} \equiv f^{\theta}(Y_{t},p_{t}^{\theta}), 
\nonumber
\\
  {g_{\theta}(Y_{t})} &= \left\{ g_{\theta}(Y_{t} \mid X_{t}=1),\ldots,g_{\theta}(Y_{t} \mid X_{t} = S) \right\}^{\T},
\end{align}
where $G_{\theta}(Y_{t}) = \text{diag}\{{g_{\theta}(Y_{t})}\}$ and $Q_{\theta}$ is the transition matrix of $\{X_{t}\}$ with $(Q_{\theta})_{ab} = q_{\theta}(a, b)$. For $\theta \in \Theta$, $f^{\theta}$ is a measurable function on $\mathcal{Y} \times \mathcal{P}(\mathcal{X})$. Extending the definition of $f^{\theta}(Y_{t},p_{t}^{\theta})$ for lags $s = 0, 1, \ldots, t - 1$, we recursively define $p_{t+1}^{\theta}$ for $s = 0,
\ldots, t-1$ as
\begin{align}
  \label{eq:27}
  p_{t+1}^{\theta} = f_0^\theta(Y_t^t, p^{\theta}_{t}) = f_{1}^{\theta}(Y_{t-1}^{t},p^{\theta}_{t-1}) = \cdots = f_{s}^{\theta}(Y_{t-s}^{t},p^{\theta}_{t-s}) = \cdots = f_{t-1}^{\theta}(Y_{1}^{t},p^{\theta}_{1}),
\end{align} 
where $f_{0}^{\theta} =f^{\theta}$. 
\subsection{Technical lemmas}

We state five technical lemmas that are used to prove Theorem \ref{thm1} and Theorem \ref{thm2} in the main manuscript. The first lemma is a restatement of Lemma 3.1 in \citet{de2008asymptotic}, except the full data likelihood and $n$ are replaced by the $j$th subset likelihood and $m$. 
\begin{lemma}[\cite{de2008asymptotic}, Lemma 3.1]\label{lem1}
  For any $\delta > 0$ and $j\in \{1, \ldots, K\}$, there exists $\epsilon > 0$ such that
  \begin{align*}
    \PP_0\left[ \underset{\| \theta - \theta_0 \|_2 > \delta}{\sup} \, \frac{1}{m} \left\{ L_{jm}(\theta) - L_{jm}(\theta_0) \right\} \leq - \epsilon \right] \rightarrow 1 \text{ as } m \rightarrow \infty .
  \end{align*}
\end{lemma}
\begin{lemma}[\cite{le2000exponential} Theorem 2.1]
\label{lemma:bvm3:1}
Under Assumption \ref{a8}, for any $\theta \in B_{\delta_0}(\theta_0)$, any $p^{\theta}_{l}, \tilde{p}^{\theta}_{l} \in \mathcal{P}(\mathcal{X})$, any $l,s\geq 1$ and sequence $y_{l},\ldots,y_{l+s} \in \mathcal{Y}^{s}$,  
$$\|f_{s}^{\theta}\{y_{l+s},\ldots,y_{l},p^{\theta}_{s}\} - f_{s}^{\theta}\{y_{l+s},\ldots,y_{l},\tilde{p}^{\theta}_{s}\} \|_{TV} \leq \epsilon_{\theta}^{-1}h_{\theta}(y_{i})\prod_{i=l+1}^{l+s}\{1 - \epsilon_{\theta} h_{\theta}(y_{i})^{-1}\}\|p^{\theta}_{l} - \tilde{p}^{\theta}_{l}\|_{TV},$$
where $\epsilon_{\theta} =\underset{a,b\in \mathcal{X}}{\min} q_{\theta}(a,b)$, $ \rho_\theta(y) = \underset{a,b \in \mathcal{X}}{\max} \frac{g_{\theta}(y|a)}{g_{\theta}(y|b)}$ and $\{1 - \epsilon_{\theta} h_{\theta}(y_{i})^{-1}\} \in (0,1).$
\end{lemma} 

\begin{lemma}[\cite{le2000exponential} Example 3.3]
\label{lemma:bvm3:2}
Under Assumption \ref{a8}, for any $\theta \in B_{\delta_0}(\theta_0)$, any $y \in \mathcal{Y}$ and any $p, \tilde{p} \in \mathcal{P}(\mathcal{X})$,
$$|\log\{\sum_{a \in \mathcal{X}}p(a)g_{\theta}(y|a)\} - \log\{\sum_{a \in \mathcal{X}}\tilde{p}(a)g_{\theta}(y|a)\}| \leq [h_{\theta}(y) - 1] \|p - \tilde{p}\|_{TV}.$$
\end{lemma}

\begin{lemma}
\label{lemma:addi:lik} The log-likelihood function of $\theta$ given $Y_{m+1}^{2m}$ and the conditional log-likelihood of $\theta$ obtained from the conditional density of $Y_{m+1}^{2m}$ given $Y_{1}^{m}$ can be expressed as 
\begin{align}
&\exp\{w_{1}(\theta)\} = p_{\theta}(Y_{m+1}^{2m})   \nonumber \\
&= \sum_{X_{m+1}^{2m}}\prod_{t=m+1}^{2m-1}r_{\theta}(X_{m+1}) q_{\theta}(X_{t},X_{t+1}) \prod_{t=m+1}^{2m} g_{\theta}(Y_{t}\mid X_{t})
= \prod_{t = m+1}^{2m} \sum_{X_{t}}g_{\theta}(Y_{t}\mid X_{t})p^{\theta}_{t}(X_{t}),
\label{lik1}
\\
&\exp\{w_{2}(\theta)\} = p_{\theta}(Y_{m+1}^{2m}\mid Y_{1}^{m} )   \nonumber \\
&= \sum_{X_{m+1}^{2m}} \prod_{t=m+1}^{2m-1}p_{\theta}(X_{m+1}\mid Y_{1}^{m})q_{\theta}(X_{t},X_{t+1}) \prod_{t=m+1}^{2m} g_{\theta}(Y_{t}\mid X_{t})
= \prod_{t = m+1}^{2m} \sum_{X_{t}}g_{\theta}(Y_{t}\mid X_{t})\tilde{p}^{\theta}_{t}(X_{t}),
\label{lik2}
\end{align}
where \begin{align*}
p_{m+1}^{\theta}(X_{m+1}) = r_{\theta}(X_{m+1}),\quad \tilde{p}_{m+1}^{\theta}(X_{m+1}) = p_{\theta}(X_{m+1} \mid Y_1^m),
\end{align*}
and for $t = m+2,\ldots,2m$,
\begin{align*} 
 p_{t}^{\theta}(X_{t}) = f^{\theta}_{t - m - 2} \{Y^{t-1}_{m+1}, r_{\theta}(X_{m+1}  )\}, \quad
 \tilde{p}_{t}^{\theta}(X_{t}) = f^{\theta}_{t - m - 2}\{Y^{t-1}_{m+1}, p_{\theta}( X_{m+1}  \mid Y_1^m)\},
\end{align*}
where $f^{\theta}_{t-m-2}$ is defined in \eqref{eq:27}.
\end{lemma}
\begin{proof}[Proof of Lemma \ref{lemma:addi:lik}]
We first prove the following equality by induction:
\begin{align}
\sum_{X_{1}^{m}} \prod_{t=1}^{m-1}r_{\theta}(X_{1}) q_{\theta}(X_{t},X_{t+1}) \prod_{t=1}^{m} g_{\theta}(Y_{t}\mid X_{t})
= \prod_{t = 1}^{m} \sum_{X_{t}}g_{\theta}(Y_{t}|X_{t})p^{\theta}_{t}(X_{t}).
\label{likp}
\end{align}
For $m = 2$, the left hand side of \eqref{likp} is 
$$\sum_{X_{1},X_{2}}r_{\theta}(X_{1}) q_{\theta}(X_{1},X_{2})g_{\theta}(Y_{1}|X_{1})g_{\theta}(Y_{2}|X_{2}) = 
\sum_{X_{2}}g_{\theta}(Y_{2}|X_{2})\sum_{X_{1}}r_{\theta}(X_{1}) q_{\theta}(X_{1},X_{2})g_{\theta}(Y_{1}|X_{1}).$$
And the right hand side of \eqref{likp} is 
\begin{align*}
\left\{ \sum_{X_{1}}r_{\theta}(X_{1}) g_{\theta}(Y_{1}\mid X_{1})  \right\}
\left\{ \sum_{X_{2}}g_{\theta}(Y_{2}\mid X_{2})p_{\theta}(X_{2}\mid Y_1)  \right\} &= 
\sum_{X_{2}}g_{\theta}(Y_{2}\mid X_{2}) 
\nonumber 
\\
&\times \sum_{X_{1}}p_{\theta}(X_{2}\mid Y_1)  r_{\theta}(X_{1})g_{\theta}(Y_{1}\mid X_{1}).
\end{align*}
For fixed $X_{2} = x_{2} \in \mathcal{X}$, we have 
\begin{align*}
\sum_{X_{1}}&r_{\theta}(X_{1}) q_{\theta}(X_{1},x_{2})g_{\theta}(Y_{1}\mid X_{1}) 
- \sum_{X_{1}}p_{\theta}(x_{2}\mid Y_1)  r_{\theta}(X_{1})g_{\theta}(Y_{1}\mid X_{1})
\\
& = 
\sum_{X_{1}}r_{\theta}(X_{1}) q_{\theta}(X_{1},x_{2})g_{\theta}(Y_{1}\mid X_{1}) 
- \sum_{X_{1}} r_{\theta}(X_{1})g_{\theta}(Y_{1}\mid X_{1}) 
\frac{\sum_{a_1}q_{\theta}(a_1,x_{2})g_{\theta}(Y_{1}\mid a_1)r_{\theta}(a_1)}{\sum_{a_1}g_{\theta}(Y_{1}\mid a_1)r_{\theta}(a_1)}
\\
&=
\frac{1}{\sum_{a_1}g_{\theta}(Y_{1}\mid a_1)r_{\theta}(a_1)} \times 
\left\{\sum_{X_{1}}r_{\theta}(X_{1}) q_{\theta}(X_{1},x_{2})g_{\theta}(Y_{1}\mid X_{1})\sum_{a_1}g_{\theta}(Y_{1}\mid a_1)r_{\theta}(a_1)  \right. \\
&\quad \quad - \left. \sum_{X_{1}} r_{\theta}(X_{1})g_{\theta}(Y_{1}\mid X_{1}) 
\sum_{a_1}q_{\theta}(a_1,x_{2})g_{\theta}(Y_{1}\mid a_1)r_{\theta}(a_1)\right\}
\\&
= 0.
\end{align*}
Hence, we have 
$$
\sum_{X_{2}}g_{\theta}(Y_{2}\mid X_{2}) 
\left\{\sum_{X_{1}}r_{\theta}(X_{1}) q_{\theta}(X_{1},X_{2})g_{\theta}(Y_{1}\mid X_{1}) 
- \sum_{X_{1}}p_{\theta}(X_{2}\mid Y_1)  r_{\theta}(X_{1})g_{\theta}(Y_{1}\mid X_{1}) \right\} = 0,
$$
which proves \eqref{likp} at $m = 2$. Suppose \eqref{likp} holds for $m = k-1 \geq 2$, we have 
$$
\sum_{X_{1}^{k-1}} \prod_{t=1}^{k-2}r_{\theta}(X_{1}) q_{\theta}(X_{t},X_{t+1}) \prod_{t=1}^{k-1} g_{\theta}(Y_{t}|X_{t})
= \prod_{t = 1}^{k-1} \sum_{X_{t}}g_{\theta}(Y_{t}\mid X_{t})p^{\theta}_{t}(X_{t}).
$$
For $m = k$, the left hand side of \eqref{likp} can be expressed as 
\begin{align*}
\sum_{X_{1}^{k}}& \prod_{t=1}^{k-1}r_{\theta}(X_{1}) q_{\theta}(X_{t},X_{t+1}) \prod_{t=1}^{k} g_{\theta}(Y_{t}\mid X_{t})
\\&
= \sum_{X_{k}} g_{\theta}(Y_{k}\mid X_{k}) \sum_{X_{1}^{k-1}}q(X_{k-1},X_{k})r_{\theta}(X_{1})  \prod_{t=1}^{k-2}q_{\theta}(X_{t},X_{t+1}) \prod_{t=1}^{k-1} g_{\theta}(Y_{t}\mid X_{t}).
\end{align*}
\end{proof}
And the right hand side of \eqref{likp} can be expressed as
\begin{align*}
\prod_{t = 1}^{k} &\sum_{X_{t}}g_{\theta}(Y_{t}\mid X_{t})p^{\theta}_{t}(X_{t})
\\&
=  \left\{\sum_{X_{k}} g_{\theta}(Y_{k}\mid X_{k}) p_{\theta}(X_{k}\mid Y_{1}^{k-1}) \right\} \left\{\prod_{t = 1}^{k-1} \sum_{X_{t}}g_{\theta}(Y_{t}\mid X_{t})p^{\theta}_{t}(X_{t})\right\}
\\&
 = 
\left\{ \sum_{X_{k}} g_{\theta}(Y_{k}\mid X_{k}) p_{\theta}(X_{k}\mid Y_{1}^{k-1})\right\}
\left\{\sum_{X_{1}^{k-1}} \prod_{t=1}^{k-2}r_{\theta}(X_{1}) q_{\theta}(X_{t},X_{t+1}) \prod_{t=1}^{k-1} g_{\theta}(Y_{t}\mid X_{t}) \right\},
\end{align*}
where the last equality is due to the induction assumption at $m = k-1$.
And we have 
$$p_{\theta}(X_{k}\mid Y_{1}^{k-1}) 
= \frac{\sum_{a_{1}^{k-1}} q_{\theta}(a_{k-1},X_{k})r_{\theta}(a_1) 
\prod_{t=1}^{k-2}q_{\theta}(a_{t},a_{t+1})\prod_{t=1}^{k-1}g_{\theta}(Y_{t}\mid a_t)
 }{ \sum_{a_{1}^{k-1}}r_{\theta}(a_1) 
\prod_{t=1}^{k-2}q_{\theta}(a_{t},a_{t+1})\prod_{t=1}^{k-1}g_{\theta}(Y_{t}\mid a_t) }.
$$
For fixed $X_{k} = x_{k} \in \mathcal{X}$, we have 
\begin{align*}
&\left\{ \sum_{X_{1}^{k-1}}q(X_{k-1},x_{k})r_{\theta}(X_{1})  \prod_{t=1}^{k-2}q_{\theta}(X_{t},X_{t+1}) \prod_{t=1}^{k-1} g_{\theta}(Y_{t}\mid X_{t})\right\}
\left\{\sum_{a_{1}^{k-1}}r_{\theta}(a_1) 
\prod_{t=1}^{k-2}q_{\theta}(a_{t},a_{t+1})\prod_{t=1}^{k-1}g_{\theta}(Y_{t}\mid a_t) \right \}
\\
=&\left\{\sum_{X_{1}^{k-1}}r_{\theta}(X_{1}) \prod_{t=1}^{k-2} q_{\theta}(X_{t},X_{t+1}) \prod_{t=1}^{k-1} g_{\theta}(Y_{t}\mid X_{t}) \right\}\left\{
\sum_{a_{1}^{k-1}} q_{\theta}(a_{k-1},x_{k})r_{\theta}(a_1) 
\prod_{t=1}^{k-2}q_{\theta}(a_{t},a_{t+1})\prod_{t=1}^{k-1}g_{\theta}(Y_{t}\mid a_t)\right\}.
\end{align*}
Hence, we have 
\begin{align*}
&\sum_{X_{1}^{k-1}}q(X_{k-1},X_{k})r_{\theta}(X_{1})  \prod_{t=1}^{k-2}q_{\theta}(X_{t},X_{t+1}) \prod_{t=1}^{k-1} g_{\theta}(Y_{t}\mid X_{t})
\\
=& p_{\theta}(X_{k}\mid Y_{1}^{k-1}) \left\{\sum_{X_{1}^{k-1}} \prod_{t=1}^{k-2}r_{\theta}(X_{1}) q_{\theta}(X_{t},X_{t+1}) \prod_{t=1}^{k-1} g_{\theta}(Y_{t}\mid X_{t}) \right\},
\end{align*}
which implies that \eqref{likp} holds for $m = k$. Changing indices from $(1,\ldots,m)$ to $(m+1,\ldots,2m)$, we prove equation \eqref{lik1}. 

By the recursive definition of prediction filter, we have for $t = m+2,\ldots,2m,$ $$ 
p_{t}^{\theta}(X_{t}) = f^{\theta}_{t - m - 2} \{Y^{t-1}_{m+1}, r_{\theta}(X_{m+1}  )\}, \quad
 \tilde{p}_{t}^{\theta}(X_{t}) = f^{\theta}_{t - m - 2}\{Y^{t-1}_{m+1}, p_{\theta}( X_{m+1}  \mid Y_1^m)\}.$$
Given $Y_{1}^{m}$, replcing the stationary distribution $r_{\theta}(X_{m+1})$ 
 by the prediction filter distribution $p_{t}(X_{m+1} \mid Y_{1}^{m})$, we prove equation \eqref{lik2}.

 \hfill $\square$

\begin{lemma}[\cite{devroye2018total} Theorem 1.2]
\label{lemma:thm2}
Suppose $d >1$, let $\mu_1 \neq \mu_2 \in \RR^{d}$ and let $V_{1},V_{2}$ be two positive definite $d \times d$ matrices. Let $v = \mu_1 - \mu_2$ and $N$ be an arbitrary $d \times d-1$ matrix with columns forming a basis for the subspace orthogonal to $v$. Define the function 
$$tv(\mu_1,V_1,\mu_2,V_2) = \max \{  \frac{|v^{\T}(V_1 - V_2) v|}{v^{\T}V_1 v}, \frac{v^\T v}{\sqrt{v^\T V_1 v}}, \|(N^{\T}V_1 N)^{-1}N^{\T}V_2 N - I_{d-1}\|_{F} \}.$$
Then, we have 
$$\frac{1}{200} \leq \frac{\|N_{d}(\mu_1,V_1)-N_{d}(\mu_2,V_2)\|_{\text{TV}}}{tv(\mu_1,V_1,\mu_2,V_2)} \leq \frac{9}{2}.$$
\end{lemma}

\subsection{Proof of Proposition \ref{filt-approx}}
\label{pf:prop1}

\emph{Proof of Proposition \ref{filt-approx}}. For $j=2$, the total variation distance equals to zero. Based on \eqref{eq:27}, for $j = 3,\ldots, K$, we have the following expressions for the approximate prediction filter and the true prediction filter:
\begin{align}
p_{\theta}&(X_{(j-1)m+1}  \mid Y_{(j-2)m+1}^{(j-1)m}) = f^{\theta}_{m-1}\{Y_{(j-2)m+1}^{(j-1)m}, r_{\theta}(X_{(j-2)m +1})\}
\nonumber\\
p_{\theta}&(X_{(j-1)m+1} \mid Y_{1}^{(j-1)m}) = f^{\theta}_{m-1}\{Y_{(j-2)m+1}^{(j-1)m}, p_{\theta}(X_{(j-2)m +1} \mid Y_{1}^{(j-2)m})\},
\label{eq:proof:p0}
\end{align}
where $r_{\theta}(\cdot) \in \mathcal{P}(\mathcal{X})$ is the stationary distribution of $\{X_{t}, t\geq 1\}$.
Then for $\theta \in \Theta$, we have that with $\PP_0$-probability almost 1,
\begin{align}
&  \|p_{\theta}(X_{(j-1)m+1}  \mid Y_{(j-2)m+1}^{(j-1)m})  - p_{\theta}(X_{(j-1)m+1} \mid Y_{1}^{(j-1)m}) \|_{\text{TV}} 
\nonumber \\
&\overset{(i)}{\leq}
\epsilon_{\theta}^{-1} h_{\theta}(Y_{(j-2)m+1})\prod_{i=(j-2)m+2}^{(j-1)m}[1 - \epsilon_{\theta} h_{\theta}(Y_{i})^{-1}]
\|r_{\theta}(X_{(j-2)m +1}) 
\nonumber
\\&
\quad \quad \quad
-  p_{\theta}(X_{(j-2)m +1} \mid Y_{1}^{(j-2)m})\|_{\text{TV}}
\nonumber \\
&\overset{(ii)}{\leq}
\epsilon^{-1}M(1 - \epsilon M^{-1})^{m-1} 
\sup_{\theta \in B_{\delta_0}(\theta_0)}\|r_{\theta}(X_{(j-2)m +1}) -  p_{\theta}(X_{(j-2)m +1} \mid Y_{1}^{(j-2)m})\|_{\text{TV}}
\nonumber \\
&\overset{(iii)}{\leq}
\epsilon^{-1}M\{(1 - \epsilon M^{-1})\}^{m-1}S \rho^{(j-2)m-2}, 
\label{eq:proof:p1}
\end{align}
where $(i)$ follows from Lemma \ref{lemma:bvm3:1}, $(ii)$ follows from that $\epsilon = \underset{\theta \in \Theta}{\inf}\epsilon_{\theta}$, and from \ref{a8}\ref{c}, $h_{\theta}(Y_{i}) < M$ for $\theta \in \Theta$ and $i = 1,\ldots,n$ $\PP_0$-almost surely. And $(iii)$ follows from assumption \ref{a8}\ref{a}. Since $(1 -\epsilon M^{-1}) , \rho \in(0,1)$ $\PP_0$-almost surely, as $m \rightarrow \infty$, the right side of \eqref{eq:proof:p1} converges to zero $\PP_0$-almost surely.

\section{Proof of Theorem \ref{thm1}}
\label{pf:thm1}
We first prove Theorem \ref{thm1}\ref{thm1:1} following the arguments used for proving Theorem 1 in \citet{Lietal16} and Theorem 2.1 in \citet{de2008asymptotic}. The proof requires some essential technical changes to account for the definition of subset posterior distributions using (quasi) conditional distribution likelihood in \eqref{eq:s6} in the main manuscript. Then, we prove Theorem \ref{thm1}\ref{thm1:2} using results from Theorem 2.1 in \cite{le2000exponential}. 

\subsection{Proof of Theorem \ref{thm1}\ref{thm1:1}}
\label{apsub1}

\noindent \emph{Step 1: Show that $\hat \theta_1$ is a weakly consistent estimator of $\theta_0$.} Assumption \ref{a5} implies that $\theta_0$ is in the interior of $\Theta$, which is compact. Assumption \ref{a1} implies that $L'_{jm}(\theta)$ exists and is continuously differentiable in the compact neighborhood $B_{\delta_0}(\theta_0)$ of $\theta_0$ for every $(Y_{(j-1)m + 1}, \ldots, Y_{jm})$ such that the components of $L''_{jm}(\theta)$ are uniformly bounded by an envelope function $M(Y_{(j-1)m + 1}, \ldots, Y_{jm})$. Furthermore, $\EE_0 \left\{ L'_{jm}(\theta_0) \right\} = 0$, so there exists a root of the equation $L'_{jm}(\theta)=0$ with large $\PP_0$-probability in $B_{\delta_0}(\theta_0)$ neighborhood. Denote this root as $\hat \theta_j$. Then, the weak consistency of $\hat \theta_j$ is a consequence of Lemma \ref{lem1}. 

\vspace{5mm}

\noindent \emph{Step 2: Simplify the statement of Theorem \ref{thm1} \ref{thm1:1}.} For any $t \in \RR^d$, define
\begin{align}\label{eq:defs}
\theta_t &= \hat \theta_1 + \frac{t}{(Km)^{1/2}} ,
\\
w(t) &= L_{1m}\left\{ \hat{\theta}_1 + \frac{t}{(Km)^{1/2}} \right\} -L_{1m}(\hat{\theta}_1 ),\nonumber
\\
C_{m} &= \int_{\mathbb{R}^{d}} e^{Kw(z)} \pi \left\{ \hat{\theta}_1+ \frac{z}{(Km)^{1/2}} \right\}dz,\nonumber
\\
g_{m}(t) &= e^{Kw(t)} \pi \left\{ \hat{\theta}_1+ \frac{t}{(Km)^{1/2}} \right\} - \exp \left( -\frac{t^{\T} I_{0}t}{2} \right) \pi(\theta_0).\nonumber
\end{align}
For proving Theorem \ref{thm1} \ref{thm1:1}, we need to show that the sequence of random variable
\begin{align}
  \label{pf:1:1}
  &\zeta_{m} 
    = \int_{\mathbb{R}^{d}} \left| \frac{e^{Kw(t)} \pi \left\{ \hat{\theta}_1+ \frac{t}{(Km)^{1/2}} \right\}} {C_{m}} -  \frac{1}{(2\pi)^{d/2}\det(I_{0})^{-1/2}} \exp \left( -\frac{1}{2}t^{\T} I_{0}t \right)  \right| dt 
\end{align}
converges to 0 in $\PP_0$-probability as $m \rightarrow \infty$. Using the definition of $g_{m}(t)$ in \eqref{eq:defs},
\begin{align*}
  &\zeta_{m} 
    = \int_{\mathbb{R}^{d}} \left| \frac{g_m(t) + \exp \left( -\frac{t^{\T} I_{0}t}{2} \right) \pi(\theta_0)} {C_{m}} -  \frac{1}{(2\pi)^{d/2}\det(I_{0})^{-1/2}} \exp \left( -\frac{1}{2}t^{\T} I_{0}t \right)  \right| dt ,
\end{align*}
and the triangle inequality implies that 
\begin{align}
\label{pf:1:2}
\zeta_{m} \leq 
\frac{1}{C_{m}}\int_{\mathbb{R}^{d}} &|g_{m}(t)|dt + \left| \frac{(2\pi)^{d/2}\det(I_{0})^{-1/2}\pi(\theta_0)}{C_{m}}-1 \right|\times 
\nonumber
\\&
\quad 
\int _{\mathbb{R}^{d}} \frac{1}{(2\pi)^{d/2}\det(I_{0})^{-1/2}} \exp \left( -\frac{1}{2}t^{\T} I_{0}t \right) dt. 
\end{align}
The definition of $C_m$ in \eqref{eq:defs} also implies that
\begin{align}
  \label{eq:6}
  C_m = \int_{\mathbb{R}^{d}} \left\{ g_m(z) +  \exp \left( -\frac{z^{\T} I_{0}z}{2} \right) \pi(\theta_0) \right\} dz = \int_{\mathbb{R}^{d}} g_m(z) dz +  (2\pi)^{d/2}\det(I_{0})^{-1/2}\pi(\theta_0).
\end{align}
If $\int _{\mathbb{R}^{d}} |g_{m}(z)|dz \rightarrow 0$ in $\PP_0$-probability as $m \rightarrow \infty$, then \eqref{eq:6} implies that
\begin{align*}
  \left| C_m - (2\pi)^{d/2}\det(I_{0})^{-1/2}\pi(\theta_0) \right| \rightarrow 0 \text{ in } \PP_0\text{-probability as }  m \rightarrow \infty,
\end{align*}
which after an application of Slutsky's theorem implies that 
\begin{align}
  \label{eq:13}
  \left| \frac{(2\pi)^{d/2}\det(I_{0})^{-1/2}\pi(\theta_0)}{C_m} - 1 \right|  \rightarrow 0 \text{ in } \PP_0\text{-probability as }  m \rightarrow \infty;
\end{align}
therefore, the second term on the right hand side in \eqref{pf:1:2} converges to 0 in $\PP_0$-probability as $m \rightarrow \infty$. This also shows that proving $\zeta_{m}  \rightarrow 0$ in $\PP_0$-probability as $m \rightarrow \infty$ is equivalent to showing that $\int _{\mathbb{R}^{d}} |g_{m}(t)|dt \rightarrow 0$ in $\PP_0$-probability  as $m \rightarrow \infty$.

\vspace{5mm}

\noindent \emph{Step 3: Show that $\int _{\mathbb{R}^{d}} |g_{m}(t)|dt \rightarrow 0$ in $\PP_0$-probability  as $m \rightarrow \infty$.}
Consider a $\theta_{m}^{\ast}$ in the interior of $\Theta$ that is close to $\hat{\theta}_1$. The Taylor expansion of $L_{1m}(\theta_{m}^{\ast} )$ at $\hat{\theta}_1$ is 
\begin{align}
\label{pf:1:3}
L_{1m}(\theta_{m}^{\ast}) = 
L_{1m}(\hat{\theta}_1)+ (\theta_{m}^{\ast} - \hat{\theta}_1)^{\T} L_{1m}'(\hat{\theta}_1) + 
\frac{1}{2}  (\theta_{m}^{\ast} - \hat \theta_{1})^{\T} L''_{1m}(\overline{\theta}_m) (\theta_{m}^{\ast} - \hat \theta_{1}), 
\end{align}
where $\overline{\theta}_m$ is between $\theta_{m}^{\ast}$ and $\hat{\theta}_1$ and $L_{1m}'(\hat{\theta}_1) = 0$ because $\hat{\theta}_1$ is the MLE of $\theta$ for the first subset. Define
\begin{align*}
  \epsilon_{m}(\overline{\theta}_m) = \frac{I_{0}}{2}  + \frac{L_{1m}''(\overline{\theta}_m)}{2m} ,
\end{align*}
and rewrite \ref{pf:1:3} as 
\begin{align}
\label{pf:1:4}
L_{1m}(\theta_{m}^{\ast}) = 
L_{1m}(\hat{\theta}_1) -
\frac{m}{2}  (\theta_{m}^{\ast} - \hat \theta_1)^{\T}I_{0} (\theta_{m}^{\ast} - \hat \theta_1)
+ m (\theta_{m}^{\ast} - \hat \theta_{1})^{\T} \epsilon_{m}(\overline{\theta}_m) (\theta_{m}^{\ast} - \hat \theta_{1}).
\end{align}
If $\theta_m^{\ast}$ is a (possibly stochastic) sequence converging to $\theta_0$, then $\epsilon_{m}(\overline{\theta}_m) \rightarrow 0$ in $\PP_0$-probability 
as $m\rightarrow \infty$ due to Assumption \ref{a4}.

We use this Taylor expansion in \eqref{pf:1:4} to show that $\int _{\mathbb{R}^{d}} |g_{m}(t)|dt \rightarrow 0$ in $\PP_0$-probability  as $m \rightarrow \infty$. We start by partitioning the domain of the integral into three sets: $A_{1} = \{t \in \mathbb{R}^{d}: \|t\|_{2} \geq \delta_{1}(Km)^{1/2}\}$, $A_{2} = \{t \in \mathbb{R}^{d}: \delta_{2} \leq \|t\|_{2} < \delta_{1}(Km)^{1/2} \}$ and $A_{3} = \{t \in \mathbb{R}^{d}: \|t\|_{2} < \delta_{2} \}$, where $\delta_{1}, \delta_{2}$ are positive quantities to be chosen later. Using the triangular inequality,
\begin{align*}
  \int _{\mathbb{R}^{d}} |g_{m}(t)|dt  \leq \int_{A_{1}}|g_{m}(t)|dt +\int_{A_{2}}|g_{m}(t)|dt  + \int_{A_{3}}|g_{m}(t)|dt.
\end{align*}
The proof is complete if $\int_{A_{i}}|g_{m}(t)|dt \rightarrow 0$ in $\PP_0$-probability  as $m \rightarrow \infty$ for $i=1, 2, 3$. 

\vspace{5mm}

\noindent \emph{Step 3.1: Show that $\int _{A_1} |g_{m}(t)|dt \rightarrow 0$ in $\PP_0$-probability  as $m \rightarrow \infty$.} Let $\theta_m^{\ast} = \hat \theta_1 + t / (Km)^{1/2}$ in \eqref{pf:1:4} for some $t \in A_1$. Using Lemma \ref{lem1}, there exists an $\epsilon_1$ depending on $\delta_1$ such that
\begin{align}
  \label{eq:15}
  L_{1m}\{\hat{\theta}_1 + t / (Km)^{1/2}\} - L_{1m}(\theta_0 )  \leq -m \epsilon_{1}
\end{align}
with $\PP_0$-probability approaching 1 as $m \rightarrow \infty$. Because $\hat{\theta}_1$ is a weakly consistent estimator $\theta_0$ and $L_{1m}(\theta)$ is a continuous function of $\theta$, $L_{1m}(\hat \theta_1) \rightarrow L_{1m}(\theta_0 )$ with $\PP_0$-probability approaching 1 as $m \rightarrow \infty$. Using this in \eqref{eq:15}, we have that for any $t \in A_1$ and $\delta_1>0$ there exists an $\epsilon_1$ depending on $\delta_1$ such that 
\begin{align}
  \label{eq:16}
  w(t) = L_{1m}\{\hat{\theta}_1 + t / (Km)^{1/2}\} - L_{1m}(\hat \theta_1 )  \leq -m \epsilon_{1}  
\end{align}
with $\PP_0$-probability approaching 1 as $m \rightarrow \infty$, where $w(t)$ is defined in \eqref{eq:defs}; therefore, with $\PP_0$-probability approaching 1 as $m \rightarrow \infty$,
\begin{align*}
  \int_{A_{1}}&|g_{m}(t)|dt \leq \int_{A_{1}} e^{Kw(t)} \pi \left\{ \hat{\theta}_1+ \frac{t}{(Km)^{1/2}} \right\} dt + \int_{A_{1}} \exp \left( -\frac{t^{\T} I_{0}t}{2} \right) \pi(\theta_0) dt\\  
  &\overset{(i)}{\leq} \exp(-Km\epsilon_{1}) \int_{A_{1}} \pi \left\{ \hat{\theta}_1+ \frac{t}{(Km)^{1/2}} \right\} dt
  +  \pi(\theta_0) \int_{A_{1}} \exp \left( -\frac{t^{\T} I_{0}t}{2} \right) dt
  \\
  &\overset{(ii)}{\leq} \exp(-Km\epsilon_{1}) (Km)^{d/2} \int_{\|\theta - \hat{\theta}_{1}\|_{2}\geq \delta_{1}}\pi(\theta)d\theta  +   \pi(\theta_0) \int_{\|t\|_{2}\geq \delta_{1}(Km)^{1/2}} \exp \left( -\frac{t^{\T} I_{0}t}{2} \right) dt
  \\
  &\overset{(iii)}{\rightarrow} 0,
\end{align*}
where $(i)$ follows from \eqref{eq:16}, $(ii)$  follows by a change of variable from $t \in A_1$ to $\theta = \hat{\theta}_1+ {t} / {(Km)^{1/2}}$ in the first term of the summation, and $(iii)$ follows from two facts: $\int_{\|\theta - \hat{\theta}_1\|_{2}\geq \delta_{1}}\pi(\theta)d\theta$ is bounded because $\pi(\theta)$ is density function and  $\int_{\|t\|_{2}\geq \delta_{1}(Km)^{1/2}} \exp (-{t^{\T} I_{0}t} / {2}) dt$ converges to 0 because the integrand equals an un-normalized multivariate Gaussian density, which implies that the un-normalized tail probability converges to $0$ as $m \rightarrow \infty$.

\noindent \emph{Step 3.2: Show that $\int _{A_2} |g_{m}(t)|dt \rightarrow 0$ in $\PP_0$-probability  as $m \rightarrow \infty$.}
We first show that $|g_{m}(t)|$ is bounded above by an integrable function with large $\PP_0$-probability  as $m \rightarrow \infty$. The desired result follows from an appropriate choice of $\delta_1$ and $\delta_2$.
Consider the Taylor expansion of $L_{1m}(\theta)$ at $\hat \theta_1$ and set $\theta = \hat \theta_1 + t / (Km)^{1/2}$ any $t \in A_2$. Because $L_{1m}(\hat \theta_1)=0$, 
\begin{align}
  \label{eq:17}
  w(t) &= L_{1m} \{\hat \theta_1 + t / (Km)^{1/2}\}  - L_{1m}(\hat \theta_1) = \frac{1}{2K} t^\T \frac{L_{1m}''(\hat \theta_1)}{m} t + R_m(t), \nonumber\\
  R_m(t) &= \frac{1}{6 (Km)^{3/2}}t^\T \left\{ L_{1m}'''(\tilde \theta) t \right\} t ,
\end{align}
for any $t \in A_2$, where $L_{1m}'''(\theta)$ is a $d\times d \times d$ array and $\tilde \theta$ satisfies $\| \tilde \theta - \hat \theta_1 \|_2 \leq \| t \|_2 / (Km)^{1/2} \leq \delta_1$, where the last inequality follows because $\delta_2 \leq \| t \|_2 \leq \delta_1 (Km)^{1/2}$. Furthermore, $\| \tilde \theta - \theta_0 \|_2  \leq \| \tilde \theta - \hat \theta_1 \|_2 + \| \hat \theta_1 - \theta_0 \|_2 \leq \delta_1 + \| \hat \theta_1 - \theta_0 \|_2 $. If we choose $\delta_1 \leq \delta_0 / 3$   and $m$ to be large enough so that $\| \hat \theta_1 - \theta_0 \|_2 \leq \delta_0 /3$ with $\PP_0$-probability approaching 1, then $\| \tilde \theta - \theta_0 \|_2  < \delta_0$ for every $t \in A_2$ with $\PP_0$-probability approaching 1 as $m \rightarrow \infty$. For  a given $t \in A_2$,
\begin{align}
  \label{eq:18}
  |R_m(t)| &\leq \frac{\| t \|_2^3}{6 K^{3/2} m^{1/2}} \sum_{l_3=1}^d\sum_{l_2=1}^d\sum_{l_1=1}^d \left| \frac{1}{m} \{L'''_{1m}(\tilde \theta)\}_{l_1 l_2 l_3} \right| \nonumber\\
             &\overset{(i)}{\leq}
               \frac{\| t \|_2^3}{6 K^{3/2} m^{1/2}} \sum_{l_3=1}^d\sum_{l_2=1}^d\sum_{l_1=1}^d  \frac{1}{m} \underset{\theta \in B_{\delta_0}(\theta_0)}{\sup} \left| \{L'''_{1m}(\theta)\}_{l_1 l_2 l_3} \right| \nonumber\\
           &\overset{(ii)}{\leq} \frac{d^3 \| t \|_2^3 }{6 K^{3/2} m^{1/2}} \, \frac{1}{m} \, M(Y_{1}, \ldots, Y_m) \overset{(iii)}{\leq} \frac{d^3 \delta_1}{6 K} \| t \|_2^2  \frac{1}{m} M(Y_{1}, \ldots, Y_m)
\nonumber \\ &
               \overset{(iv)}{\rightarrow} \frac{d^3 \delta_1}{6 K} \| t \|_2^2  \mathcal{C}_{M} \quad \PP_0\text{-almost surely;}
\end{align}
where $(i)$ follows because $\tilde \theta \in B_{\delta_0}(\theta_0)$ with a large $\PP_0$-probability as $m \rightarrow \infty$, $(ii)$ follows from Assumption \ref{a1}, $(iii)$ follows because $\| t \|_2 \leq \delta_1 (Km)^{1/2}$ and $(iv)$ also follows from Assumption \ref{a1}. By Assumption \ref{a4}, eigenvalues of $-\frac{1}{m}L_{1m}''(\theta)$ are bounded from blow and above by constants for all $\theta \in B_{\delta_0}(\theta_0)$. Hence we choose
  \begin{align}
    \label{eq:del1}
    \delta_1 = \min\left[ \frac{\delta_0}{3},
    \frac{6 \min_{\theta \in B_{\delta_0}(\theta_0)}\lambda_1\{-\frac{1}{m}L_{1m}''(\theta)\}} {4d^{3}\mathcal{C}_{M}}
    \right],
  \end{align}
where $\lambda_{1}(B)$ denotes the smallest eigenvalue of matrix $B$.  With this choice of $\delta_1$, \eqref{eq:17} and \eqref{eq:18} imply that for every  $t \in A_{2}$ and with $\PP_0$-probability approaching 1 as $m \rightarrow \infty$, 
\begin{align}
|KR_m(t)| &\leq  \min_{\theta \in B_{\delta_0}(\theta_0)} \lambda_1 \{  - L_{1m}''(\theta)\} \; t^\T t / (4m)  \leq - t^\T L_{1m}''(\hat \theta_1) t / (4m), \nonumber\\
  \label{eq:20}
  K w(t) &\leq  \frac{1}{2} t^\T \frac{L_{1m}''(\hat \theta_1)}{m} t + K |R_m(t)| \leq \frac{1}{2} t^\T \frac{L_{1m}''(\hat \theta_1)}{m} t - \frac{1}{4} t^\T \frac{L_{1m}''(\hat \theta_1)}{m} t =
   \frac{1}{4} t^\T \frac{L_{1m}''(\hat \theta_1)}{m} t .
\end{align}
We use \eqref{eq:20} to find an upper bound for $|g(t)|$ that is integrable on $A_2$. By assumption \ref{a4}, $-m^{-1} L_{1m}''(\hat{\theta}_{1}) \rightarrow I_{0}$ in $\PP_0$-probability as $m \rightarrow \infty$.  This implies that $\exp(K w(t)) \leq \exp (- t^\T I_0t / 4 )$ with a large $\PP_0$-probability as $m \rightarrow \infty$; therefore, with a large $\PP_0$-probability as $m \rightarrow \infty$,
\begin{align}
  \label{eq:21}
  \int_{A_2} | g_{m}(t) | dt &\overset{(i)}{\leq} \int_{A_2}  \exp \left( -\frac{t^{\T} I_{0}t}{4} \right) \pi \left\{ \hat{\theta}_1+ \frac{t}{(Km)^{1/2}} \right\} dt + \int_{A_2} \exp \left( -\frac{t^{\T} I_{0}t}{2} \right)  \pi(\theta_0) dt\nonumber\\
                             &\leq \int_{A_2}  \exp \left( -\frac{t^{\T} I_{0}t}{4} \right) \pi \left\{ \hat{\theta}_1+ \frac{t}{(Km)^{1/2}} \right\} dt + \int_{A_2} \exp \left( -\frac{t^{\T} I_{0}t}{4} \right)  \pi(\theta_0) dt\nonumber\\
                             &\overset{(ii)}{\leq} \left\{ \underset{\theta \in \Theta}{\sup} \, \pi(\theta) \right\} \int_{A_2}  2 \exp \left( -\frac{t^{\T} I_{0}t}{4} \right) dt \overset{(iii)}{<} \infty,
\end{align} 
where $(i)$ follows from triangular inequality and \eqref{eq:20}, $(ii)$ is because $\theta_0, \hat{\theta}_1+ {t} / {(Km)^{1/2}} \in \Theta$, $\Theta$ is compact from Assumption \ref{a5}, and $\pi(\cdot)$ is continuous on $\Theta$ from Assumption \ref{a6}, and $(iii)$ follows because the continuity of $\pi(\cdot)$ and the compactness of $\Theta$ imply that $\pi(\theta)$ is bounded for every $\theta \in \Theta$ and the integrand equals an un-normalized Gaussian density, which has a finite integral over $\RR^d$.  We now can choose $\delta_2$ sufficiently large such that $\int_{A_2} | g_{m}(t) | dt$ is arbitrarily small in $\PP_0$-probability.

\vspace{5mm}

\noindent \emph{Step 3.3: Show that $\int _{A_3} |g_{m}(t)|dt \rightarrow 0$ in $\PP_0$-probability  as $m \rightarrow \infty$.} 
Using \eqref{eq:18} and Assumption \ref{a1}, we have that for any $\|t\|_{2} < \delta_{2}$ and a sufficiently large $m$,
\begin{align}
\label{eq:25}
\sup_{\|t\|_{2} < \delta_{2}} |KR_{m}(t)| 
& \leq 
\sup_{\|t\|_{2} < \delta_{2}} \frac{d^3 \| t \|_2^3 }{6 K^{1/2} m^{1/2}} \, \frac{1}{m} \, M(Y_{1}, \ldots, Y_m) 
\nonumber
\\
&
\leq
\frac{d^3\delta_2^3 }{6 K^{1/2} m^{1/2}} \, \frac{1}{m} \, M(Y_{1}, \ldots, Y_m) \rightarrow 0,
\end{align}
where the convergence is almost surely in $\PP_0$-probability. The Taylor expansion in \eqref{eq:17} yields that as $m \rightarrow \infty$,
\begin{align}
\label{eq:261}
\sup_{t \in A_{3}}\exp\left\{ \left|Kw(t) + \frac{1}{2}t^{\top}I_0 t \right| \right\}
&
= \sup_{\|t\|_{2} < \delta_{2}}\exp\left\{ \left| \frac{1}{2} t^\T \frac{L_{1m}''(\hat \theta_1)}{m} t + KR_m(t) + \frac{1}{2}t^{\top}I_0 t \right|\right\}
\nonumber\\
&\leq\exp\left\{ \frac{\delta_{2}^{2}}{2} \left \|\frac{L_{1m}''(\hat \theta_1)}{m} + I_{0} \right\|_{\text{op}} \right\}\exp \left\{ \sup_{\|t\|_{2} < \delta_{2}}|KR_{m}(t)|\right\}
\nonumber\\
&\overset{(i)}{\rightarrow} 1, \text{ in } \PP_0\text{-probability},
\end{align}
where $\| \cdot \|_{\text{op}}$ is the operator norm and $(i)$ follows from ${L_{1m}''(\hat \theta_1)}/{m} + I_{0} \rightarrow 0$ in $\PP_0$-probability as $m \rightarrow \infty$ using Assumption \ref{a4} and weak consistency of $\hat \theta_1$, and $\sup_{\|t\|_{2} < \delta_{2}}|KR_{m}(t)|\rightarrow 0$ in $\PP_0$-probability as $m \rightarrow \infty$ using \eqref{eq:25}. The weak consistency of $\hat \theta_1$ and continuity of $\pi(\cdot)$ from Assumption \ref{a6} imply that for any fixed $\delta_2>0$
\begin{align}
  \label{eq:23}
  \underset{\|t\|_2 \leq \delta_2}{\sup} \left| \pi \left\{ \hat{\theta}_1+ \frac{t}{(Km)^{1/2}} \right\} - \pi(\theta_0) \right| \rightarrow 0 \text{ as } m \rightarrow \infty
\end{align}
in $\PP_0$-probability. Combining \eqref{eq:261} and \eqref{eq:23}, we get that
\begin{align}
  \label{eq:24}
  \int _{A_3} |g_{m}(t)|dt \leq C \int_{A_3} &\underset{\|t\|_2 \leq \delta_2}{\sup} \left[ \exp \left\{ Kw(t) + \frac{1}{2}  t^{\T} I_{0} t \right\} -1 \right] 
\times 
\nonumber
\\
&
\underset{\|t\|_2 \leq \delta_2}{\sup} \left| \pi \left\{ \hat{\theta}_1+ \frac{t}{(Km)^{1/2}} \right\} - \pi(\theta_0) \right| dt,
\end{align}
where $C$ is a universal constant; therefore, $\int _{A_3} |g_{m}(t)|dt \rightarrow 0$ in $\PP_0$-probability as $m \rightarrow \infty$ because the first and second terms in the integrand converge to 0 uniformly from \eqref{eq:261} and \eqref{eq:23}.

\vspace{5mm}

\noindent \emph{Step 4: Show that $\int _{\mathbb{R}^{d}} |g_{m}(t)|dt \rightarrow 0$ in $\PP_0$-probability  as $m \rightarrow \infty$.}
Using steps 3.1 to 3.3, choose  $\delta_1$ as defined in \eqref{eq:del1} and  $\delta_{2}$ large enough such that $\int _{A_2} |g_{m}(t)|dt \rightarrow 0$ in $\PP_0$-probability  as $m \rightarrow \infty$. This implies using   the triangle inequality that
\begin{align*}
  \int _{\mathbb{R}^{d}} |g_{m}(t)|dt  \leq \int_{A_{1}}|g_{m}(t)|dt +\int_{A_{2}}|g_{m}(t)|dt  + \int_{A_{3}}|g_{m}(t)|dt \rightarrow 0
\end{align*}
in $\PP_{0}$-probability as $m \rightarrow \infty$. The proof is complete. \hfill $\square$

\subsection{Proof of Theorem \ref{thm1}\ref{thm1:2}}
\label{apsub2}

We first prove Theorem \ref{thm1}\ref{thm1:2} for the second subset posterior distribution and then prove for $j$th subset posterior distribution defined in \eqref{eq:s6} in the main manuscript  $(j=3,\ldots,K)$.

\vspace{5mm}

\noindent
\emph{Step 1: Define two log-likelihood functions.} We define the two log-likelihood functions of $\theta$ that correspond to two different marginal distributions of $X_{m+1}$. Let $Y_{1}^{m}$ and $Y_{m+1}^{2m}$ denote the first and second data subsets, respectively, $w_{1}(\theta)$ be the log-likelihood of $\theta$ given $Y_{m+1}^{2m}$, and $w_2(\theta)$ be the conditional  log-likelihood of $\theta$ obtained from the conditional density of $Y_{m+1}^{2m}$ given $Y_1^m$. Then,
\begin{align}
  \exp\{w_{1}(\theta)\} &= p_{\theta} (Y_{m+1}^{2m}) = \sum_{X_{m+1}^{2m} \in \Xcal^m} r_{\theta}(X_{m+1}) \prod_{t=m+1}^{2m-1} q_{\theta}(X_{t}, X_{t+1}) \prod_{t=m+1}^{2m} g_{\theta}(Y_{t} \mid X_{t}), \label{eq:28}\\
  \exp\{w_{2}(\theta)\} &= p_{\theta}(Y_{m+1}^{2m} \mid Y_1^m)  = \sum_{X_{m+1}^{2m} \in \Xcal^m}  p_{\theta}(X_{m+1} \mid Y_1^m)\prod_{t=m+1}^{2m-1}q_{\theta}(X_{t},X_{t+1}) \prod_{t=m+1}^{2m} g_{\theta}(Y_{t} \mid X_{t}). \label{eq:29}
\end{align}
If the marginal distribution of $X_{m+1}$ is $r_{\theta}(\cdot)$, then $\exp\{w_{1}(\theta)\}$ is obtained by marginalizing over  $X_{m+1}^{2m}$ in the joint distribution of $(Y_{m+1}^{2m}, X_{m+1}^{2m})$. Similarly, we recover $\exp\{w_{2}(\theta)\}$ after marginalizing over $X_{m+1}^{2m}$ from the conditional distribution of $(Y_{m+1}^{2m}, X_{m+1}^{2m})$ given $Y_1^m$. This is also equivalent to the marginal density of $Y_{m+1}^{2m}$ obtained from the joint distribution of $(Y_{m+1}^{2m}, X_{m+1}^{2m})$, where the marginal density of $X_{m+1}$ is set to be $p_{\theta}(X_{m+1} \in \cdot \mid Y_1^m)$.
We re-express $\exp\{w_{1}(\theta)\}$ and $\exp\{w_{2}(\theta)\}$ in terms of the prediction filter in \eqref{eq:27} using Lemma \ref{lemma:addi:lik} as
\begin{align}
 \exp\{w_{1}(\theta)\} &= \prod_{t = m+1}^{2m} \sum_{a=1}^S g_{\theta}(Y_{t} \mid X_{t} = a) \, p^{\theta}_{t}(a),
\nonumber
\\
\exp\{w_{2}(\theta)\}& = \prod_{t = m+1}^{2m} \sum_{a=1}^S g_{\theta}(Y_{t} \mid X_{t} = a) \, \tilde{p}^{\theta}_{t}(a), \label{eq:2loglik}
\end{align}
where \begin{align*}
p_{m+1}^{\theta}(X_{m+1}) = r_{\theta}(X_{m+1}),\quad \tilde{p}_{m+1}^{\theta}(X_{m+1}) = p_{\theta}(X_{m+1} \mid Y_1^m),
\end{align*}
and for $t = m+2,\ldots,2m$
\begin{align*} 
 p_{t}^{\theta}(X_{t}) = f^{\theta}_{t - m - 2} \{Y^{t-1}_{m+1}, r_{\theta}(X_{m+1}  )\}, \quad
 \tilde{p}_{t}^{\theta}(X_{t}) = f^{\theta}_{t - m - 2}\{Y^{t-1}_{m+1}, p_{\theta}( X_{m+1}  \mid Y_1^m)\},
\end{align*}
where $f^{\theta}_{t-m-2}$ is defined in \eqref{eq:27}.
\vspace{5mm}

\noindent
\emph{Step 2: Show that $\sup_{\theta \in \Theta}|w_{1}(\theta)- w_{2}(\theta)| \leq
 \tilde  C \rho^{m-2}$ $\PP_0$-almost surely for some constant $\tilde C$.} 

For any $\theta \in\Theta$ we have that
\begin{align}\label{eq:log-liks-filt}
  &|w_{1}(\theta) - w_{2}(\theta)|
\nonumber
\\
& 
= \left| \sum_{t=m+1}^{2m} \left[ \log \left\{ \sum_{a=1}^S p_{t}^{\theta}(a)g_{\theta}(Y_{t} \mid X_{t} = a)  \right\}  
- 
\log \left\{ \sum_{a=1}^S \tilde{p}_{t}^{\theta}(a) g_{\theta}(Y_{t} \mid X_{t} = a) \right\}  \right]  \right| 
  \nonumber\\
  &\overset{(i)}{\leq} \sum_{t=m+1}^{2m} \left| \log \left\{ \sum_{a=1}^S p_{t}^{\theta}(a)g_{\theta}(Y_{t} \mid X_{t} = a)  \right\}  - \log \left\{ \sum_{a=1}^S \tilde{p}_{t}^{\theta}(a) g_{\theta}(Y_{t} \mid X_{t} = a) \right\}  \right| 
  \nonumber\\
  &\overset{(ii)}{\leq}\sum_{t=m+1}^{2m}\left\{ h_{\theta}(Y_{t}) - 1 \right\} \|p_{t}^{\theta} - \tilde{p}_{t}^{\theta}\|_{\text{TV}} \nonumber\\
  &= \left\{ h_{\theta}(Y_{m+1}) - 1 \right\} \| r_{\theta}(X_{m+1}) - p_{\theta}(X_{m+1} \mid Y_1^m )\|_{\text{TV}} + 
    \left\{ h_{\theta}(Y_{m+2}) - 1 \right\} \|p_{m+2}^{\theta} - \tilde{p}_{m+2}^{\theta}\|_{\text{TV}} \nonumber\\
  & \quad \quad + \sum_{t=m+3}^{2m}\left\{ h_{\theta}(Y_{t}) - 1 \right\} \|p_{t}^{\theta} - \tilde{p}_{t}^{\theta}\|_{\text{TV}}
  \nonumber\\
  &\overset{(iii)}{\leq}    \left\{ h_{\theta}(Y_{m+1}) - 1 \right\} \| r_{\theta}(X_{m+1}) - p_{\theta}(X_{m+1} \mid Y_1^m )\|_{\text{TV}} 
\nonumber
\\
& \quad \quad+
    \left\{ h_{\theta}(Y_{m+2}) - 1 \right\} \epsilon_{\theta}^{-1} h_{\theta}(Y_{m+2}) \| r_{\theta}(X_{m+1}) 
p_{\theta}(X_{m+1} \mid Y_1^m )\|_{\text{TV}} 
\nonumber
\\ 
& \quad \quad +
\sum_{t=m+3}^{2m}\left\{ h_{\theta}(Y_{t}) - 1 \right\} \epsilon_{\theta}^{-1} h_{\theta}(Y_t) \prod_{i=m+2}^{t-1} \{1 - \epsilon_{\theta} h_{\theta}(Y_i)^{-1}\} 
\| r_{\theta}(X_{m+1}) - p_{\theta}(X_{m+1} \mid Y_1^m )\|_{\text{TV}}
  \nonumber\\
  &\overset{(iv)}{\leq}    \left\{ h_{\theta}(Y_{m+1}) - 1 \right\} \epsilon_{\theta}^{-1} h_{\theta}(Y_{m+1}) \| r_{\theta}(X_{m+1}) - p_{\theta}(X_{m+1} \mid Y_1^m )\|_{\text{TV}}  \nonumber\\
  &\quad \quad  +\left\{ h_{\theta}(Y_{m+2}) - 1 \right\} \epsilon_{\theta}^{-1} h_{\theta}(Y_{m+2}) \| r_{\theta}(X_{m+1}) - p_{\theta}(X_{m+1} \mid Y_1^m )\|_{\text{TV}} 
  \nonumber\\
  &\quad \quad +\sum_{t=m+3}^{2m}\left\{ h_{\theta}(Y_{t}) - 1 \right\} \epsilon_{\theta}^{-1} h_{\theta}(Y_t) \prod_{i=m+2}^{t-1} \{1 - \epsilon h_{\theta}(Y_i)^{-1}\} 
   \| r_{\theta}(X_{m+1}) - p_{\theta}(X_{m+1} \mid Y_1^m )\|_{\text{TV}}, 
\end{align}
where $(i)$ follows from the triangle inequality, $(ii)$ follows from Lemma  \ref{lemma:bvm3:2}, $(iii)$ follows from Lemma \ref{lemma:bvm3:1}, and $(iv)$ follows from from Lemma \ref{lemma:bvm3:1} because $\epsilon^{-1} h_{\theta}(Y_{m+1}) > 1$.  Assumption \ref{a8} implies that $\sup_{\theta \in \Theta} h_{\theta}(Y_{t}) \leq M$ for $1\leq t \leq n$ $\PP_0$-almost surely. The right side of \eqref{eq:log-liks-filt} is a decreasing function of $\epsilon_{\theta}$ and $\epsilon = \inf_{\theta \in \Theta}\epsilon_{\theta} > 0 $. Hence, with $\PP_0$-probability almost 1, by taking $\sup_{\theta \in \Theta}$ over  
the right hand side of \eqref{eq:log-liks-filt} we have that
\begin{align}
& \sup_{\theta \in \Theta} |w_{1}(\theta) - w_{2}(\theta)|
  \leq  \frac{1}{\epsilon} \sup_{\theta \in\Theta } \| r_{\theta}(X_{m+1}) - p_{\theta}(X_{m+1} \mid Y_1^m )\|_{\text{TV}} \times 
\nonumber \\
& \quad \quad
\sup_{\theta \in \Theta}
\left[
\left\{ h_{\theta}(Y_{m+1}) - 1 \right\} h_{\theta}(Y_{m+1}) + \left\{ h_{\theta}(Y_{m+2}) - 1 \right\} h_{\theta}(Y_{m+2}) \right.
\nonumber
\\
&
\left. \quad \quad
+ \sum_{t=m+3}^{2m}\left\{ h_{\theta}(Y_{t}) - 1 \right\} \epsilon^{-1} h_{\theta}(Y_t) \prod_{i=m+2}^{t-1} \{1 - \epsilon h_{\theta}(Y_i)^{-1}\} 
 \right] 
\nonumber \\
&\leq  \frac{1}{\epsilon}\sup_{\theta \in \Theta} \| r_{\theta}(X_{m+1}) - p_{\theta}(X_{m+1} \mid Y_1^m )\|_{\text{TV}} \times 
\nonumber
\\
& \quad \quad
\left\{ 2M(M-1) + 
 \sum_{t=m+3}^{2m} M(M-1) \prod_{i=m+2}^{t-1} (1 - \epsilon M^{-1}) \right\} \nonumber \\
  &\leq \frac{1}{\epsilon}  M(M-1)  \sup_{\theta \in \Theta} \| r_{\theta}(X_{m+1}) - p_{\theta}(X_{m+1} \mid Y_1^m )\|_{\text{TV}} \left\{ 2 + \sum_{t=1}^{\infty} (1 - \epsilon M^{-1})^{t}\right\} \nonumber \\
  &\leq \frac{1}{\epsilon^2} M(M^{2}-1)\sup_{\theta \in \Theta}  \| r_{\theta}(X_{m+1}) - p_{\theta}(X_{m+1} \mid Y_1^m )\|_{\text{TV}}
\nonumber \\
& \overset{(i)}{\leq} 
 \frac{1}{\epsilon^2} M(M^{2}-1)S \rho^{m-2} = \tilde C \rho^{m-2}
 ,\label{eq:32}
\end{align}
where $(i)$ follows from assumption \ref{a8}\ref{a}: $\sup_{\theta \in \Theta}\| r_{\theta}(X_{m+1}) - p_{\theta}(X_{m+1} \mid Y_1^m )\|_{\text{TV}} < S\rho^{m-2}$ $\PP_0$-almost surely.

\vspace{5mm}
\noindent
\emph{Step 3: Show that the total variation distance of subset quasi posterior distributions induced by likelihood $\exp\{w_1(\theta)\}$ and $\exp\{w_2(\theta)\}$ tends to 0 in $\PP_0$-probability as $m$ tends to infinity.} 

The definition of $w_1(\theta)$ implies that $w_1(\theta) = L_{2m}(\theta)$, where $L_{2m}(\theta)$ is the log-likelihood of $\theta$ on the second subset $Y_{m+1}^{2m}$. Let $\hat{\theta}_{2}$ denote the maximum likelihood estimator of $\theta$ that solves $L_{2m}'(\theta) = 0$ or $w_{1}'(\theta) = 0$. For large $m$, $\hat{\theta}_{2} \in B_{\delta_0}(\theta_0)$, hence the limiting set of local parameters $ \lim_{m \rightarrow \infty} \{t = (Km)^{1/2} (\theta - \hat \theta_2): \theta \in \Theta\} = \RR^{d}$. Then the densities of the second subset posterior distributions of $t = (Km)^{1/2} (\theta - \hat \theta_2)$ with likelihoods $\exp\{w_1(\hat \theta_2 + t / (Km)^{1/2})\}$ and $\exp\{w_2(\hat \theta_2 + t / (Km)^{1/2})\}$ are
\begin{align*}
  \tilde \pi_{w_1}(t \mid Y_{m+1}^{2m}) &= C_{1m}^{-1}\, {e^{K\left  [w_1\left\{\hat \theta_2 + \tfrac{t} {(Km)^{1/2}}\right\}- w_{1}\left\{ \hat{\theta}_2\right\}  \right]}\pi \left\{\hat \theta_2 + \tfrac{t}{(Km)^{1/2}} \right\}}, \nonumber\\ 
  \tilde \pi_{w_2}(t \mid Y_{m+1}^{2m}) &= C_{2m}^{-1}\, {e^{K \left [w_2\left\{\hat \theta_2 + \tfrac{t} {(Km)^{1/2}}\right\}- w_{2}\left\{ \hat{\theta}_2\right\}  \right]}\pi \left\{\hat \theta_2 + \tfrac{t}{(Km)^{1/2}} \right\}}, \nonumber\\   
  C_{1m} &= \int_{\RR^d} e^{K\left [w_1\left\{\hat \theta_2 + \tfrac{z} {(Km)^{1/2}}\right\} - w_{1}\left\{ \hat{\theta}_2\right\}  \right]}\pi \left\{\hat \theta_2 + \tfrac{z}{(Km)^{1/2}} \right\} dz, \\
  C_{2m} &= \int_{\RR^d} e^{K\left [w_2\left\{\hat \theta_2 + \tfrac{z} {(Km)^{1/2}}\right\}- w_{2}\left\{ \hat{\theta}_2\right\}  \right]}\pi \left\{\hat \theta_2 + \tfrac{z}{(Km)^{1/2}} \right\} dz
\end{align*}
The total variation distance between subset posterior distributions of $\theta$ with the likelihoods $\exp\{w_1(\theta)\}$ and $\exp\{w_2(\theta)\}$ equals $\int_{\RR^d} \left| \tilde \pi_{w_1}(t \mid Y_{m+1}^{2m})  - \tilde \pi_{w_2}(t \mid Y_{m+1}^{2m})  \right| dt$; therefore, we want to prove that 
\begin{align}
  \int_{\RR^{d}} &\left| C_{1m}^{-1}\, e^{K\left [w_1\left\{\hat \theta_2 + \tfrac{t} {(Km)^{1/2}}\right\}- w_{1}\left\{ \hat{\theta}_2\right\} \right ] } - C_{2m}^{-1}\, e^{K \left [w_2\left\{\hat \theta_2 + \tfrac{t} {(Km)^{1/2}}\right\}- w_{2}\left\{ \hat{\theta}_2\right\} \right ] }  \right|
 \times
\nonumber
\\
& \quad \quad \quad 
\pi \left\{\hat \theta_2 + \tfrac{t}{(Km)^{1/2}} \right\}  dt \rightarrow 0
  \label{eq:ax1}
\end{align}
in $\PP_0$-probability as $m \rightarrow \infty$. An application of triangular inequality implies that an upper bound for the integral in \eqref{eq:ax1} is 
\begin{align}
  \label{eq:ax2}
  &C_{1m}^{-1} \int_{\RR^d} \left|e^{K  \left [w_1\left\{\hat \theta_2 + \tfrac{t} {(Km)^{1/2}}\right\}- w_{1}\left\{ \hat{\theta}_2\right\} \right ] } -  e^{K  \left [w_2\left\{\hat \theta_2 + \tfrac{t} {(Km)^{1/2}}\right\}- w_{2}\left\{ \hat{\theta}_2\right\} \right ] }  \right| \pi \left\{\hat \theta_2 + \tfrac{t}{(Km)^{1/2}} \right\}   dt  \nonumber\\
  & + \left| C_{1m}^{-1} C_{2m} - 1 \right| \int_{\RR^d} C_{2m}^{-1} \, e^{K  \left[w_2\left\{\hat \theta_2 + \tfrac{t} {(Km)^{1/2}}\right\}- w_{2}\left\{ \hat{\theta}_2\right\} \right ]} \pi \left\{\hat \theta_2 + \tfrac{t}{(Km)^{1/2}} \right\} dt,
\end{align}
and we will show that the upper bound in \eqref{eq:ax2} tends to 0 in $\PP_0$-probability as $m \rightarrow \infty$.

First, we show that $C_{1m}$ is bounded away from 0 with a large $\PP_0$-probability as $m \rightarrow \infty$. Because $w_1(\theta) = L_{2m}(\theta)$, we have that 
$w_1 \left\{ \hat \theta_2 + \tfrac{t}{(Km)^{1/2}} \right\} = L_{2m} \left\{ \hat \theta_2 + \tfrac{t}{(Km)^{1/2}} \right\}$ and that 
\begin{align}
  \label{eq:ax3}
  C_{1m} \rightarrow (2 \pi)^{d/2} \det(I_0)^{-1/2} \pi(\theta_0)  
\end{align}
in $\PP_0$-probability as $m \rightarrow \infty$ based on the same arguments used to prove \eqref{eq:13}; therefore, $C_{1m}^{-1}$ is finite with a large $\PP_0$-probability as $m \rightarrow \infty$.

Second, we find a sufficient condition which guarantees that  $|C_{1m}^{-1} C_{2m} - 1 | \rightarrow 0$ in $\PP_0$-probability as $m \rightarrow \infty$.
We have that
\begin{align}
  \label{eq:ax5}
  &|C_{2m} - C_{1m}| 
\nonumber
\\
&=   \left| \int_{\RR^d} e^{K  \left [w_1\left\{\hat \theta_2 + \tfrac{t} {(Km)^{1/2}}\right\}- w_{1}\left\{ \hat{\theta}_2\right\} \right ] } -  e^{K  \left [w_2\left\{\hat \theta_2 + \tfrac{t} {(Km)^{1/2}}\right\}- w_{2}\left\{ \hat{\theta}_2\right\} \right ] }   \pi \left\{\hat \theta_2 + \tfrac{t}{(Km)^{1/2}} \right\}   dt \right|,\nonumber\\
  &\leq \int_{\RR^d} \left|e^{K   \left [w_1\left\{\hat \theta_2 + \tfrac{t} {(Km)^{1/2}}\right\}- w_{1}\left\{ \hat{\theta}_2\right\} \right ] } -  e^{K   \left [w_2\left\{\hat \theta_2 + \tfrac{t} {(Km)^{1/2}}\right\}- w_{2}\left\{ \hat{\theta}_2\right\} \right ] }  \right| \pi \left\{\hat \theta_2 + \tfrac{t}{(Km)^{1/2}} \right\}   dt
\end{align}
where the last term is the first integral in \eqref{eq:ax2}; therefore, it is sufficient to show that the first integral in \eqref{eq:ax2} tends to 0 in $\PP_0$-probability as $m \rightarrow \infty$ because $C_{1m}$ is bounded away from zero with a large $\PP_0$-probability as $m \rightarrow \infty$. 

Thirdly, we find a sufficient condition which gaurantees that the second term of the upper bound in \eqref{eq:ax2} tends to 0 in $\PP_0$-probability as $m \rightarrow \infty$. If the first integral in \eqref{eq:ax2} tends to 0 in $\PP_0$-probability as $m \rightarrow \infty$, then $|C_{1m}^{-1} C_{2m} - 1 | \rightarrow 0$ and $|C_{2m} - C_{1m}| \rightarrow 0$ in $\PP_0$-probability as $m \rightarrow \infty$, and by triangle inequality we have that
\begin{align}
  \label{eq:ax6}
&\int_{\RR^d} e^{K  \left [w_2\left\{\hat \theta_2 + \tfrac{t} {(Km)^{1/2}}\right\}- w_{2}\left\{ \hat{\theta}_2\right\} \right ] }\pi \left\{\hat \theta_2 + \tfrac{t}{(Km)^{1/2}} \right\}  dt \nonumber\\
&\leq C_{1m} +  \int_{\RR^d} \left|e^{K   \left [w_1\left\{\hat \theta_2 + \tfrac{t} {(Km)^{1/2}}\right\}- w_{1}\left\{ \hat{\theta}_2\right\} \right ] } -  e^{K   \left [w_2\left\{\hat \theta_2 + \tfrac{t} {(Km)^{1/2}}\right\}- w_{2}\left\{ \hat{\theta}_2\right\} \right ] }  \right| \times  
\nonumber
\\
& \quad \quad 
\pi \left\{\hat \theta_2 + \tfrac{t}{(Km)^{1/2}} \right\}   dt
\rightarrow  (2 \pi)^{d/2} \det(I_0)^{-1/2} \pi(\theta_0)
\end{align}
in $\PP_0$-probability as $m \rightarrow \infty$, which implies that  the second integral in \eqref{eq:ax2} is bounded in $\PP_0$-probability as $m \rightarrow \infty$. By Slutsky's theorem we have that 
\begin{align}
  \label{eq:ax7}
\left| C_{1m}^{-1} C_{2m} - 1 \right| \int_{\RR^d} C_{2m}^{-1} \, e^{K  \left[w_2\left\{\hat \theta_2 + \tfrac{t} {(Km)^{1/2}}\right\}- w_{2}\left\{ \hat{\theta}_2\right\} \right ]} \pi \left\{\hat \theta_2 + \tfrac{t}{(Km)^{1/2}} \right\} dt \rightarrow  0 
\end{align}
in $\PP_0$-probability as $m \rightarrow \infty$; therefore, it is sufficient to show that in $\PP_0$-probability as $m \rightarrow \infty$,
\begin{align}
\label{eq:ax7.5}
 \int_{\RR^d} \left|e^{K  \left [w_1\left\{\hat \theta_2 + \tfrac{t} {(Km)^{1/2}}\right\}- w_{1}\left\{ \hat{\theta}_2\right\} \right ] } -  e^{K  \left [w_2\left\{\hat \theta_2 + \tfrac{t} {(Km)^{1/2}}\right\}- w_{2}\left\{ \hat{\theta}_2\right\} \right ] }  \right| \pi \left\{\hat \theta_2 + \tfrac{t}{(Km)^{1/2}} \right\}   dt \rightarrow 0.
\end{align}

\vspace{5mm}

\noindent
\emph{Step 4: Show that the integral in \eqref{eq:ax7.5} tends to 0 in $\PP_0$-probability as $m \rightarrow \infty$.}  

For any  $t \in \RR^d$, let $\theta^{\ast} = \hat{\theta}_2 + \frac{t}{(Km)^{1/2}}$.  
Using \eqref{eq:32} in step 2, as $ m \rightarrow \infty$, with $\PP_0$ probability tending to 1,
\begin{align}
  \label{eq:ax14}
  &\left| K w_{1}(\theta^\ast) - K w_{2}(\theta^\ast) \right|
  \leq   K \underset{\theta \in \Theta}{\sup} \left|  w_{1}(\theta) -  w_{2}(\theta)\right| 
\nonumber \\
&\leq  
    \tilde C\rho^{-2}K \rho^m.
\end{align}
If $K = o( \rho^{-m})$, as $m \rightarrow \infty $, $\left| K w_{1}(\theta^\ast) - K w_{2}(\theta^\ast) \right| \rightarrow 0$ with $\PP_0$-probability tending to 1. Let $t = 0$, we have that  $\left| K w_{1}(\hat{\theta}_2) - K w_{2}(\hat{\theta}_2) \right| \rightarrow 0$ with $\PP_0$-probability tending to 1. By triangle inequality, as $m \rightarrow \infty$, $\left| K [w_{1}(\theta^\ast) - w_{1}(\hat{\theta}_2)]- K [w_{2}(\theta^\ast) - w_{2}(\hat{\theta}_2)]\right| \rightarrow 0$ with $\PP_0$-probability tending to 1. Then by continuous mapping theorem, we have that   
\begin{align}
  \label{eq:35}
  \left|1 - \exp\left(K [w_{2}(\theta^\ast) - w_{2}(\hat{\theta}_2)]- K [w_{1}(\theta^\ast) - w_{1}(\hat{\theta}_2)]\right) \right| 
\rightarrow 0
\end{align}
with $\PP_0$-probability tending to 1 as $m \rightarrow \infty$. The  integral in \eqref{eq:ax7.5} can be expressed as
\begin{align}
  \label{eq:ax15}
  &\int_{\RR^{d}} \left|e^{K  \left [w_1\left\{\hat \theta_2 + \tfrac{t} {(Km)^{1/2}}\right\}- w_{1}\left\{ \hat{\theta}_2\right\} \right ] } -  e^{K  \left [w_2\left\{\hat \theta_2 + \tfrac{t} {(Km)^{1/2}}\right\}- w_{2}\left\{ \hat{\theta}_2\right\} \right ] }  \right| \pi \left\{\hat \theta_2 + \tfrac{t}{(Km)^{1/2}} \right\}   dt
\nonumber \\
& =
  \int_{\RR^{d}} e^{K  \left [w_1\left\{\hat \theta_2 + \tfrac{t} {(Km)^{1/2}}\right\}- w_{1}\left\{ \hat{\theta}_2\right\} \right ] }  \left|1 - \exp\left(K [w_{2}(\theta^\ast) - w_{2}(\hat{\theta}_2)]- K [w_{1}(\theta^\ast) - w_{1}(\hat{\theta}_2)]\right) \right| \times
\nonumber
\\
&\quad \quad \quad
  \pi \left\{\hat \theta_2 + \tfrac{t}{(Km)^{1/2}} \right\}   dt ,
\end{align}
and we know that $  \int_{\RR^{d}} e^{K  \left [w_1\left\{\hat \theta_2 + \tfrac{t} {(Km)^{1/2}}\right\}- w_{1}\left\{ \hat{\theta}_2\right\} \right ] } \pi \left\{\hat \theta_2 + \tfrac{t}{(Km)^{1/2}} \right\}   dt= C_{1m}$, which is bounded in $\PP_0$-probability by \eqref{eq:ax3}. An application of the dominated convergence theorem implies that integral in \eqref{eq:ax15} tends to 0 in $\PP_0$-probability as $m \rightarrow \infty$.

\vspace{5mm}
\noindent
\emph{Step 5: Show that Theorem \ref{thm1}\ref{thm1:2} holds for $j = 2$.}
 
The proof of Theorem \ref{thm1}\ref{thm1:1} in Section \ref{apsub1} implies that as $m \rightarrow \infty$,
\begin{align}
\int_{\mathbb{R}^{d}} \left| C_{1m}^{-1} e^{K  \left [w_1\left\{\hat \theta_2 + \tfrac{t} {(Km)^{1/2}}\right\}- w_{1}\left\{ \hat{\theta}_2\right\} \right ] }\pi \left\{\hat \theta_2 + \tfrac{t}{(Km)^{1/2}} \right\} - (2\pi)^{-d/2}\det(I_{0})^{1/2} e^{-\tfrac{1}{2} t^{\T} I_{0}t}\right| dt \rightarrow 0 \label{eq:ax151}
\end{align}
in $\PP_0$-probability. For $j=2$, \eqref{eq:sub2-subK} reduces to 
\begin{align}
  \label{eq:ax16}
  &\int_{\mathbb{R}^{d}} \left| C_{2m}^{-1} e^{K \left [w_2\left\{\hat \theta_2 + \tfrac{t} {(Km)^{1/2}}\right\}- w_{2}\left\{ \hat{\theta}_2\right\} \right ]  }\pi \left\{\hat \theta_2 + \tfrac{t}{(Km)^{1/2}} \right\} - (2\pi)^{-d/2}\det(I_{0})^{1/2} e^{-\tfrac{1}{2} t^{\T} I_{0}t}\right| dt  \nonumber\\
\leq & 
\int_{\mathbb{R}^{d}} \left| C_{2m}^{-1} e^{K  \left [w_2\left\{\hat \theta_2 + \tfrac{t} {(Km)^{1/2}}\right\}- w_{2}\left\{ \hat{\theta}_2\right\} \right ] }\pi \left\{\hat \theta_2 + \tfrac{t}{(Km)^{1/2}} \right\} -
       C_{1m}^{-1} e^{K  \left [w_1\left\{\hat \theta_2 + \tfrac{t} {(Km)^{1/2}}\right\}- w_{1}\left\{ \hat{\theta}_2\right\} \right ] }\right|\times
\nonumber\\
& \quad \quad
\pi \left\{\hat \theta_2 + \tfrac{t}{(Km)^{1/2}} \right\}  dt  
\nonumber\\
& + \int_{\mathbb{R}^{d}} \left| C_{1m}^{-1} e^{K  \left [w_1\left\{\hat \theta_2 + \tfrac{t} {(Km)^{1/2}}\right\}- w_{1}\left\{ \hat{\theta}_2\right\} \right ] }\pi \left\{\hat \theta_2 + \tfrac{t}{(Km)^{1/2}} \right\} - (2\pi)^{-d/2}\det(I_{0})^{1/2} e^{-\tfrac{1}{2} t^{\T} I_{0}t}\right| dt  .
\end{align}
The first and second integrals on the right hand side of \eqref{eq:ax16} tend to 0 in $\PP_0$-probability as $m \rightarrow \infty$ using steps 3, 4 and \eqref{eq:ax151}, respectively. This proves \eqref{eq:sub2-subK} for $j=2$

We conclude the proof by showing that Theorem \ref{thm1}\ref{thm1:2} holds for $j = 3, \ldots, K$. 
\vspace{5mm}

\noindent \emph{Step 6: Prove Theorem \ref{thm1} \ref{thm1:2} for $j$th subset quasi posterior distributions $(j = 3, \ldots, K)$}.

For any $\theta \in \Theta$ and $j = 3,\ldots,K$, use \eqref{eq:28}, \eqref{eq:29} to 
define the $j$th subset likelihood functions that are equivalent to $\exp\{w_1(\theta)\}$ and $\exp\{w_2(\theta)\}$, respectively, as
\begin{align*}
\exp\{w_{1}^{j}(\theta)\} &= p_{\theta}(Y_{(j-1)m+1}^{jm})  = \prod_{t = (j-1)m+1}^{jm} \sum_{a=1}^S g_{\theta}(Y_{t} \mid X_{t} = a) p^{\theta}_{t}(a),
\\
\exp\{w_{2}^{j}(\theta)\}  &= p_{\theta}(Y_{(j-1)m+1}^{jm} \mid Y_{(j-2)m + 1}^{(j-1)m}) = \prod_{t = (j-1)m+1}^{jm} \sum_{a=1}^S g_{\theta}(Y_{t} \mid X_{t} = a) \tilde{p}^{\theta}_{t}(a),
\end{align*}
where \begin{align*}
p_{(j-1)m+1}^{\theta}(X_{(j-1)m+1}) &= r_{\theta}(X_{(j-1)m+1}),\nonumber \\  \tilde{p}_{(j-1)m+1}^{\theta}(X_{(j-1)m+1}) &= p_{\theta}( X_{(j-1)m+1}  \mid Y_{(j-2)m + 1}^{(j-1)m})
\end{align*}
and for $t=(j-1)m+2,\ldots,jm$,
\begin{align*}
p_{t}^{\theta}(X_{t}) &= f^{\theta}_{t - m - 2}\left\{ Y^{t-1}_{(j-1)m+1}, r_{\theta}(X_{(j-1)m+1}  ) \right\},
\\
\tilde{p}_{t}^{\theta}(X_{t}) &= f^{\theta}_{t - m - 2}\left\{ Y^{t-1}_{(j-1)m+1}, p_{\theta}( X_{(j-1)m+1}  \mid Y_{(j-2)m + 1}^{(j-1)m}) \right\},
\end{align*}
where $f^{\theta}_{t-m-2}$ is defined in \eqref{eq:27}. 
By Assumption \ref{a8}, we have that with $\PP_0$-probability almost 1, $$\sup_{\theta \in B_{\delta_0}(\theta_0)}\| r_{\theta}(X_{(j-1)m+1}) - p_{\theta}(X_{(j-1)m+1} | Y_{(j-2)m + 1}^{(j-1)m}) \|_{\text{TV}} < \rho^{m-2}.$$
For any  $t \in \RR^d$, let $\theta^{\ast} = \hat{\theta}_j + \frac{t}{(Km)^{1/2}}$.  
Using a similar argument of \eqref{eq:32} in step 2, we have that 
\begin{align}
  \label{eq:ax14:j}
  &\left| K w^{j}_{1}(\theta^\ast)  - K w^{j}_{2}(\theta^\ast) \right|
\nonumber
\\ 
&\leq  
    \frac{1}{\epsilon^2} M(M^{2}-1) K \|r_{\theta^\ast}(X_{(j-1)m+1}) - p_{\theta^\ast}(X_{(j-1)m+1}  \mid Y_{(j-2)m + 1}^{(j-1)m} )\|_{\text{TV}} \nonumber\\
  &\leq  
    \frac{1}{\epsilon^2} M(M^2-1) K \underset{\theta \in \Theta}{\sup}\|r_{\theta}(X_{(j-1)m+1}) - p_{\theta}(X_{(j-1)m+1}  \mid Y_{(j-2)m + 1}^{(j-1)m})\|_{\text{TV}} 
\nonumber\\
& \leq
    \epsilon^{-2}C \rho^{-2}M(M^2-1)K \rho^{m}.
\end{align}
As $m \rightarrow \infty $,  if $K = o( \rho^{-m})$, $\left| K w^{j}_{1}(\theta^\ast) - K w^{j}_{2}(\theta^\ast) \right| \rightarrow 0$ with $\PP_0$-probability tending to 1.
For $t \in  \lim_{m \rightarrow \infty} \{t = (Km)^{1/2} (\theta - \hat \theta_j): \theta \in \Theta\} = \mathbb{R}^{d}$, we define
\begin{align}
  \label{eq:ax15:j}
  \tilde \pi^{j}_{w_1}(t \mid Y_{m+1}^{2m}) &= C_{1m}^{-1}\, {e^{K\left  [w^{j}_1\left\{\hat \theta_j + \tfrac{t} {(Km)^{1/2}}\right\}- w^{j}_{1}\left\{ \hat{\theta}_j\right\}  \right]}\pi \left\{\hat \theta_j + \tfrac{t}{(Km)^{1/2}} \right\}}, \nonumber\\ 
  \tilde \pi_{w^{j}_2}(t \mid Y_{m+1}^{2m}) &= C_{2m}^{-1}\, {e^{K \left [w^{j}_2\left\{\hat \theta_j + \tfrac{t} {(Km)^{1/2}}\right\}- w^{j}_{2}\left\{ \hat{\theta}_j\right\}  \right]}\pi \left\{\hat \theta_j + \tfrac{t}{(Km)^{1/2}} \right\}}, \nonumber\\   
  C^{j}_{1m} &= \int_{\RR^d} e^{K\left [w^{j}_1\left\{\hat \theta_2 + \tfrac{z} {(Km)^{1/2}}\right\} - w^{j}_{1}\left\{ \hat{\theta}_j\right\}  \right]}\pi \left\{\hat \theta_j + \tfrac{z}{(Km)^{1/2}} \right\} dz, \\
  C^{j}_{2m} &= \int_{\RR^d} e^{K\left [w^{j}_2\left\{\hat \theta_j + \tfrac{z} {(Km)^{1/2}}\right\}- w^{j}_{2}\left\{ \hat{\theta}_j\right\}  \right]}\pi \left\{\hat \theta_j + \tfrac{z}{(Km)^{1/2}} \right\} dz
\end{align}
Following a similar argument of step 3 and step 4, we can prove that as $m \rightarrow \infty$, with large $\PP_0$-probability
\begin{align}
  \label{eq:ax1:j}
&C^{j}_{1m} \rightarrow (2 \pi)^{d/2} \det(I_0)^{-1/2} \pi(\theta_0)  ,
\nonumber \\
&C^{j}_{2m} \rightarrow (2 \pi)^{d/2} \det(I_0)^{-1/2} \pi(\theta_0)  ,
\nonumber \\
  &\int_{\mathbb{R}^{d}}  \left| (C^{j}_{1m})^{-1}\, e^{K\left [w^{j}_1\left\{\hat \theta_j + \tfrac{t} {(Km)^{1/2}}\right\} - w^{j}_{1}\left\{ \hat{\theta}_j\right\} \right ] } - (C^{j}_{2m})^{-1}\, e^{K \left [w^{j}_2\left\{\hat \theta_j + \tfrac{t} {(Km)^{1/2}}\right\}- w^{j}_{2}\left\{ \hat{\theta}_j\right\} \right ] }  \right| \times
\nonumber
\\
&
\quad \quad
 \pi \left\{\hat \theta_j + \tfrac{t}{(Km)^{1/2}} \right\}  dt \rightarrow 0.
\end{align}
By the proof of Theorem \ref{thm1}\ref{thm1:1} in Section \ref{apsub1} implies that as 
$m \rightarrow \infty$,

\begin{align}
\int_{\mathbb{R}^{d}} \left| (C^{j}_{1m})^{-1}\, e^{K\left [w^{j}_1\left\{\hat \theta_j + \tfrac{t} {(Km)^{1/2}}\right\} - w^{j}_{1}\left\{ \hat{\theta}_j\right\} \right ] } \pi \left\{\hat \theta_j + \tfrac{t}{(Km)^{1/2}} \right\} - (2\pi)^{-d/2}\det(I_{0})^{1/2} e^{-\tfrac{1}{2} t^{\T} I_{0}t}\right| dt \rightarrow 0 \label{eq:ax151:j}
\end{align}
in $\PP_0$-probability.

Theorem \ref{thm1} \ref{thm1:2} at $j$ can be expressed as
\begin{align}
  \label{eq:ax16:j}
  &\int_{\mathbb{R}^{d}} \left| (C^{j}_{2m})^{-1}\, e^{K\left [w^{j}_2\left\{\hat \theta_j + \tfrac{t} {(Km)^{1/2}}\right\} - w^{j}_{2}\left\{ \hat{\theta}_j\right\} \right ] } \pi \left\{\hat \theta_j + \tfrac{t}{(Km)^{1/2}} \right\} - (2\pi)^{-d/2}\det(I_{0})^{1/2} e^{-\tfrac{1}{2} t^{\T} I_{0}t}\right| dt \nonumber\\
\leq & 
\int_{\mathbb{R}^{d}}  \left| (C^{j}_{1m})^{-1}\, e^{K\left [w^{j}_1\left\{\hat \theta_j + \tfrac{t} {(Km)^{1/2}}\right\} - w^{j}_{1}\left\{ \hat{\theta}_j\right\} \right ] } - (C^{j}_{2m})^{-1}\, e^{K \left [w^{j}_2\left\{\hat \theta_j + \tfrac{t} {(Km)^{1/2}}\right\}- w^{j}_{2}\left\{ \hat{\theta}_j\right\} \right ] }  \right|\times  
\nonumber\\
& \quad \quad 
\pi \left\{\hat \theta_j + \tfrac{t}{(Km)^{1/2}} \right\}  dt  
\nonumber\\
&
+ \int_{\mathbb{R}^{d}} \left| (C^{j}_{1m})^{-1}\, e^{K\left [w^{j}_1\left\{\hat \theta_j + \tfrac{t} {(Km)^{1/2}}\right\} - w^{j}_{1}\left\{ \hat{\theta}_j\right\} \right ] } \pi \left\{\hat \theta_j + \tfrac{t}{(Km)^{1/2}} \right\} - (2\pi)^{-d/2}\det(I_{0})^{1/2} e^{-\tfrac{1}{2} t^{\T} I_{0}t}\right| dt   .
\end{align}
The first and second integrals on the right hand side of \eqref{eq:ax16:j} tend to 0 in $\PP_0$-probability as $m \rightarrow \infty$ using  \eqref{eq:ax1:j} and \eqref{eq:ax151:j}, respectively. This proves Theorem\ref{thm1} \ref{thm1:2} for $j$th subset quasi posterior distribution, which completes the proof.  \hfill $\square$
\section{Proof of Theorem \ref{thm2}}
\label{pf:thm2}

\subsection{Proof of Theorem \ref{thm2}\ref{thm2:1}}
\vspace{5mm}
\noindent
\emph{Step 1: Find an upper bound for the total variation distance between $\tilde{\Pi}_n(\cdot \mid Y_1^n)$ and  $d$-variate normal distribution with mean zero and $I_0^{-1}$ as the covariance matrix.}

For any $K \in \NN$ and $j=1,\ldots,K$, let $h_j =  \sqrt{n}(\theta - \hat{\theta}_{j})$ denote the $j$th local parameter, $\Pi_m^{h_j}(\cdot \mid Y_{(j-1)m + 1}^{jm})$ denote the distribution of $h_j$ implied by the $j$th subset posterior distribution $\tilde \Pi_m(\cdot \mid Y_{(j-1)m + 1}^{jm})$ with density $\tilde \pi_m(\theta \mid Y_{(j-1)m + 1}^{jm})$. Theorem \ref{thm1} in the main manuscript implies that 
\begin{align}
    \| \Pi_{m}^{h_j}(\cdot \mid Y_{(j-1)m + 1}^{jm}) - N_d(0, I_0^{-1}) \|_{\text{TV}} \rightarrow 0 \text{ in } \PP_0\text{-probability as } m \rightarrow \infty, 
\;\;j=1,\ldots,K.
\label{eq:thm2:0}
\end{align}

By combination scheme, the density of block filtered posterior distribution $\tilde{\Pi}_{n}(\theta \mid Y_{1}^{n})$ is defined as 
\begin{align}
\tilde{\pi}_{n}(\theta \mid Y_{1}^{n}) = \frac{1}{K}\sum_{j=1}^K  \frac{\det( \Sigma_j)^{1/2}}{\det(\tilde \Sigma)^{1/2}} \tilde \pi_m \left\{ \Sigma_j^{1/2} \tilde \Sigma^{-1/2} (\theta - \tilde \theta) + \hat \theta_j \mid Y_{(j-1)m + 1}^{jm} \right\},\quad \theta \in \Theta.
\label{eq:thm2:1}
\end{align}
Let $\tilde{\Pi}_n^{h}(\cdot \mid Y_{1}^{n})$ denote the distribution of $h = \sqrt{n}(\theta - \tilde{\theta})$ implied by block filtered posterior distribution $\tilde \Pi_n(\cdot \mid Y_{1}^{n})$. Then the density of $\tilde{\Pi}_n^{h}(\cdot \mid Y_{1}^{n})$ is given by 
\begin{align}
\tilde{\pi}_{n}^{h}(h \mid Y_{1}^{n}) &= \frac{1}{K} \sum_{j=1}^K  \frac{\det( \Sigma_j)^{1/2}}{\det(\tilde \Sigma)^{1/2}}\tilde \pi_m \left\{ \Sigma_j^{1/2} \tilde \Sigma^{-1/2} \frac{h}{\sqrt{n}} + \hat \theta_j \mid Y_{(j-1)m + 1}^{jm} \right\} n^{-d/2}
\nonumber \\
&=
 \frac{1}{K} \sum_{j=1}^K  \frac{\det( \Sigma_j)^{1/2}}{\det(\tilde \Sigma)^{1/2}} \pi_m^{h_{j}} \left\{ \Sigma_j^{1/2} \tilde \Sigma^{-1/2} h  \mid Y_{(j-1)m + 1}^{jm} \right\}
,\quad h \in \RR^{d}
,
\label{eq:thm2:2}
\end{align}
where $ \pi_m^{h_{j}} (h \mid Y_{(j-1)m}^{jm})$ is the density of $ \Pi_m^{h_{j}} (\cdot \mid  Y_{(j-1)m}^{jm}).$

Let $\phi(\theta;\mu_{\theta},\Sigma)$ denote the density function of $d$-variate normal distribution with mean $\mu_{\theta}$ and $\Sigma$ as covariance matrix. The total variation distance between $\tilde{\Pi}_{n}^{h}(\cdot \mid Y_{1}^{n})$ and $N_{d}(0, I_{0}^{-1})$  is upper bounded by 
\begin{align}
&\|\tilde{\Pi}_{n}^{h}(\cdot \mid Y_{1}^{n})  -N_{d}(0, I_{0}^{-1})\|_{\text{TV}} 
\nonumber\\& 
\overset{(i)}{=}\frac{1}{2} \int_{\RR^{d}} |\frac{1}{K}\sum_{j=1}^K 
\frac{\det( \Sigma_j)^{1/2}}{\det(\tilde \Sigma)^{1/2}} \pi_m^{h_{j}} \left\{ \Sigma_j^{1/2} \tilde \Sigma^{-1/2} h  \mid Y_{(j-1)m + 1}^{jm} \right\} - \phi(h;0, I_{0}^{-1})| dh
\nonumber\\&
\overset{(ii)}{\leq}  \frac{1}{2K} \sum_{j=1}^{K} \int_{\RR^{d}} | \frac{\det( \Sigma_j)^{1/2}}{\det(\tilde \Sigma)^{1/2}} \pi_m^{h_{j}} \left\{ \Sigma_j^{1/2} \tilde \Sigma^{-1/2} h  \mid Y_{(j-1)m + 1}^{jm} \right\} - \phi(h;0, I_{0}^{-1}) | dh
\nonumber\\
&
\overset{(iii)}{=}  \frac{1}{2K} \sum_{j=1}^{K} \int_{\RR^{d}} |  \pi_m^{h_{j}} \left\{ t  \mid Y_{(j-1)m + 1}^{jm} \right\} - \frac{\det(\tilde \Sigma)^{1/2}}{\det(\Sigma_j)^{1/2}}\phi(\tilde \Sigma^{1/2} \Sigma_j^{-1/2}  t ;0, I_{0}^{-1}) | dt
\nonumber\\
&
=
 \frac{1}{2K} \sum_{j=1}^{K} \int_{\RR^{d}} |\pi_m^{h_{j}} \left\{ t  \mid Y_{(j-1)m + 1}^{jm} \right\} - \phi(t;0,  \Sigma_j^{1/2}\tilde \Sigma^{-1/2}I_{0}^{-1}\tilde \Sigma^{-1/2}\Sigma_j^{1/2})| dt
\nonumber\\
&
= \frac{1}{K}\sum_{j=1}^{K}\left\{ \| \Pi_{m}^{h_j}(\cdot \mid Y_{(j-1)m + 1}^{jm} )  - N_{d}(0, \Sigma_j^{1/2}\tilde \Sigma^{-1/2}I_{0}^{-1}\tilde \Sigma^{-1/2}\Sigma_j^{1/2}) \|_{\text{TV}} \right\}
\nonumber\\&
\overset{(iv)}{\leq} \frac{1}{K}\sum_{j=1}^{K}\left\{ \| \Pi_{m}^{h_j}(\cdot \mid Y_{(j-1)m + 1}^{jm} ) - N_{d}(0, I_0^{-1}) \|_{\text{TV}}
 \right.
\nonumber 
\\
&\quad \quad \left. +  \| N_{d}(0, I_0^{-1}) -  N_{d}\left(0,  \Sigma_j^{1/2}\tilde \Sigma^{-1/2}I_{0}^{-1}\tilde \Sigma^{-1/2}\Sigma_j^{1/2}\right)) \|_{\text{TV}} \right\},
\label{eq:thm2:3}
\end{align}
where $(i)$ follows from the definition of total variation distance and \eqref{eq:thm2:2}, $(ii)$ follows from triangle inequality, $(iii)$ follows from the transformation $t = \Sigma_{j}^{1/2} \tilde \Sigma ^{-1/2} h$, and $(iv)$ again follows from triangle inequality. The first term inside the summation of \eqref{eq:thm2:3} converges to 0 in $\PP_{0}$ probability by \eqref{eq:thm2:0}. The second term inside the summation of \eqref{eq:thm2:3} is the total variation distance between two normal distribution with zero mean and different variance. 

\vspace{5mm}

\noindent
\emph{Step 2: Find an upper bound for the second term inside the summation of \eqref{eq:thm2:3}.}

Let $N_{d}(\mu_1,V_1)$ and $N_{d}(\mu_2,V_2)$ be two $d$-variate normal distributions with mean $\mu_1,\mu_2$ and covariance matrix $V_1, V_2$, respectively. Let $v=\mu_1-\mu_2$.  By Lemma \ref{lemma:thm2}, the total variation distance $\|N_{d}(\mu_1,V_1) - N_{d}(\mu_1,V_1)\|_{\text{TV}}$ is upper bounded by : 
\begin{align}
\|N_{d}(\mu_1,V_1) &-N_{d}(\mu_2,V_2)\|_{\text{TV}}
\nonumber
\\
&
\leq \frac{9}{2}
\max \{  \frac{|v^{\T}(V_1 - V_2) v|}{v^{\T}V_1 v}, \frac{v^\T v}{\sqrt{v^\T V_1 v}}, \|(N^{\T}V_1 N)^{-1}N^{\T}V_2 N - I_{d-1}\|_{F} \},
\label{eq:thm2:4}
\end{align}
where $N$ is a $d\times (d-1)$ matrix whose columns form a basis for the subspace orthogonal to $v$.

Let $A = (a_{ik}) $ and $B = (b_{kj})$ be two matrix such that $AB$ exists. In general, we have an upper bound for $\|AB\|_{F}$
\begin{align}
\|AB\|_{F}  &= \left(\sum_{i}\sum_{j} |\sum_{k}a_{ik}b_{kj}|^{2}\right)^{1/2}
\overset{(i)}{\leq}
\left(\sum_{i}\sum_{j}\sum_{k}a_{ik}^{2}\sum_{k}b_{kj}^{2}\right)^{1/2}
\nonumber \\
&
= \left( \sum_{i,k}a_{ik}^{2}\sum_{j,k}b_{kj}^{2} \right)^{1/2}
= \|A\|_{F}\|B\|_{F},
\label{eq:thm2:5}
\end{align} 
where $(i)$ follows from Cauchy-Schwarz inequality.

Denote $F_{j} = \Sigma_{j}^{1/2} - (nI_0)^{-1/2}$ and $\tilde{F} = \tilde{\Sigma}^{-1/2} - (nI_0)^{1/2}$. By assumption \ref{a9}, we have 
\begin{align}
\label{eq:fix6}
\EE_0 &\|\Sigma_{j}^{1/2}\tilde{\Sigma}^{-1/2} - I_{d}\|_{F}^{2}
=  \EE_0 \|\{(nI_0)^{-1/2} + F_{j}\}\{ (nI_0)^{1/2} + \tilde{F}\} - I_{d}\|_{F}^{2}
\nonumber
\\
&
=
\EE_0 \| n^{-1/2}I_{0}^{-1/2} \tilde{F} + n^{1/2}F_{j} I_{0}^{1/2} + F_{j}\tilde F\|_{F}^{2}
\nonumber
\\
&
\overset{(i)}{\leq} 3n^{-1} \|I_{0}^{-1/2}\|_{F}^{2} \|\tilde F\|_{F}^{2} + 3n\|I_{0}^{1/2}\|_{F}^{2}\EE_0 \| F_j\|_{F}^{2} + 3\EE_{0}n\|F_{j}\|_{F}^{2}n^{-1}\|\tilde F\|_{F}^{2}
\nonumber
\\
&
\overset{(ii)}{\rightarrow}0,
\end{align}
where $(i)$ follows from \eqref{eq:fix6}, and $(ii)$ follows \ref{a9} that $n^{-1}\|\tilde F\|_{F}^{2} ,\EE_0 n\|F_{j}\|_{F}^{2} \rightarrow 0$ and $\|I_0^{1/2}\|_{F},\|I_0^{-1/2}\|_{F}<\infty$.
Hence, for $j = 1,\ldots,K$, denote $M_{j} = \Sigma_{j}^{1/2}\tilde{\Sigma}^{-1/2} - I_{d}$. By Assumption \ref{a9} we have as $m \rightarrow \infty$,
\begin{align}
\label{eq:fix1}
\EE_0 \|M_{j}\|_{F}^{2} \rightarrow 0, \quad \text{uniformly for } j = 1,\ldots,K.
\end{align}

For $j = 1,$ we consider the difference $\| N_{d}(0, I_0^{-1}) -  N_{d}\left(0,  \Sigma_1^{1/2}\tilde{\Sigma}^{-1/2}I_{0}^{-1}\tilde \Sigma^{-1/2}\Sigma_1^{1/2}\right) \|_{\text{TV}}$. And the corresponding vector $v$ and matrix $V_{1},V_{2}$ in \eqref{eq:thm2:4} are 
\begin{align*}
&v =0, 
\\
& V_2-V_1  = \Sigma_{1}^{1/2}\tilde{\Sigma}^{-1/2}I_{0}^{-1}\tilde{\Sigma}^{-1/2}\Sigma_{1}^{1/2} -  I_{0}^{-1} =  M_{1}I_{0}^{-1} +I_{0}^{-1}M_{1}^{\T} + M_{1}I_{0}^{-1}M_{1}^{\T}.
\end{align*}
Since $v = 0$, the first two terms of the right hand side of \eqref{eq:thm2:4} is zero. For the third term we have that
\begin{align}
\|(N^{\T}V_1 N)^{-1}N^{\T}V_2 N - I_{d-1}\|_{F} &= \|(N^{\T}V_1 N)^{-1}N^{\T}(V_{2}-V_{1})N\|_{F} 
\nonumber\\
& = \|(N^{\T}I_{0}^{-1} N)^{-1}N^{\T}(M_{1}I_{0}^{-1} +I_{0}^{-1}M_{1}^{\T} + M_{1}I_{0}^{-1}M_{1}^{\T}) N\|_{F}.
\label{eq:thm2:5}
\end{align}
We have that
\begin{align}
\|(N^{\T}I_{0}^{-1} N)^{-1}&N^{\T}(M_{1}I_{0}^{-1} +I_{0}^{-1}M_{1}^{\T} + M_{1}I_{0}^{-1}M_{1}^{\T}) N\|_{F} 
\nonumber
\\
&
\overset{(i)}{\leq}\|I_{0}^{-1} \|_{F}\|M_{1}\|_{F}(2 + \|M\|_{F}) \|(N^{\T}I_{0}^{-1} N)^{-1}\|_{F} \|N\|_{F}^{2} 
\nonumber 
\\&
\overset{(ii)}{\leq}
\|I_{0}^{-1}\|_{F} \overline{\lambda}(I_{0}) \|M_{1}\|_{F}(2 + \|M\|_{F}) (d-1)^{3/2}  ,
\label{eq:thm2:6}
\end{align}
where $\overline{\lambda}(I_{0})$ denotes the largest singular value of $I_{0}$, $(i)$ follows from \eqref{eq:thm2:5}, and $(ii)$ follows from $\|(N^{\T}I_{0}^{-1} N)^{-1}\|_{F} \leq (d-1)^{1/2} \overline{\lambda}(I_{0})$ and $\|N\|_{F} = \sqrt{d-1}$ Hence by \eqref{eq:thm2:4} and \eqref{eq:thm2:6}, we have that
\begin{align}
\| N_{d}(0, I_0^{-1}) -  N_{d}\left(0,  \Sigma_1^{1/2}\tilde \Sigma^{-1/2}I_{0}^{-1}\tilde \Sigma^{-1/2}\Sigma_1^{1/2}\right)) \|_{\text{TV}}  \leq \tilde{C}\|M_{1}\|_{F}(2 + \|M_{1}\|_{F}),
\label{eq:thm2:7}
\end{align}
where $\tilde{C} = \frac{9 (d-1)^{3/2} \overline{\lambda}(I_{0})\|I_{0}^{-1} \|_{F}}{2} < \infty$.
By a simlar argument, we can show that for $j=2,\ldots, K$,
\begin{align}
\| N_{d}(0, I_0^{-1}) -  N_{d}\left(0,  \Sigma_j^{1/2}\tilde \Sigma^{-1/2}I_{0}^{-1}\tilde \Sigma^{-1/2}\Sigma_j^{1/2}\right)) \|_{\text{TV}}  \leq \tilde{C}\|M_{j}\|_{F}(2 + \|M_{j}\|_{F}).
\label{eq:thm2:8}
\end{align}

\vspace{5mm}
\noindent
\emph{Step 3: Show that Theorem \ref{thm2}\ref{thm2:1} holds for any sequence $K_{m} \rightarrow \infty$ as $m \rightarrow \infty$ and $K_{m} = o(\rho^{-m})$.}

For any $K = 1,2,\ldots$, we define a sequence of random variables $(D_{j,m}^{K})_{m=1}^{\infty}$ as
\begin{align}
D_{j,m}^{K} = 
\| \Pi_{m}^{h_j}(\cdot \mid Y_{(j-1)m + 1}^{jm}) -  N_{d}\left(0,  I_{0}^{-1}\right) \|_{\text{TV}} ,\quad j = 1,\ldots,K,\;\;
m = 1,2,\ldots.
\label{eq:thm2:10}
\end{align}   
Since $D_{j.m}^{K} \in [0,1]$ for any $ 1 \leq j \leq K$ and $K ,m \in \NN$, hence for any given sequence $K_{m}  \in \NN$ with  $K_{m} \overset{m \rightarrow \infty}{\rightarrow} \infty$, we obtain that $\{D_{j,m}^{K_{m}},j = 1,\ldots,K_{m},m=1,2,\ldots\}$ is uniformly integrable. Together with $D_{j,m}^{K_{m}} $ converging to 0 in $\PP_0$-probability we have that as $m \rightarrow \infty$,
\begin{align}
\EE_0 D_{j,m}^{K_{m}} \rightarrow 0,\quad j =1,\ldots,K_{m}.
\label{eq:mk1}
\end{align}
Based on definition of subset quasi posterior distribution $\tilde{\Pi}_{m}(\cdot \mid Y_{(j-1)m +1}^{jm})$, the density of the first subset posterior $\tilde{\Pi}_{m}(\cdot \mid Y_{1}^{m})$ is a function of $Y_{1}^{m}$, whereas the density of the $\tilde{\Pi}_{m}(\cdot \mid Y_{(j-1)m + 1}^{jm})$ is a function of $Y_{(j-2)m +1 }^{jm}$ for $j =2,\ldots, K_{m}$.
By assumption \ref{a7}, $\{Y_{t},t \geq 1\}$ is stationary. Hence we have that for any $K_{m}, m \geq 1$,
\begin{align}\EE_0 D_{1,m}^{K_{m}} \neq  \EE_0 D_{2,m}^{K_{m}} =\cdots =  \EE_0 D_{K_{m},m}^{K_{m}}.
\label{eq:mk2}
\end{align}
For any $\epsilon >0$, we have that as $m \rightarrow \infty$,
\begin{align}
&\PP_0  \left[ \|\tilde{\Pi}_{n}^{h}(\cdot \mid Y_{1}^{n})  -N_{d}(0, I_{0}^{-1})\|_{\text{TV}}  > \epsilon \right ]
\nonumber \\
&\overset{(i)}{\leq}
\PP_0 \left [ \frac{1}{K_m}\sum_{j=1}^{K_m}\left\{  D_{j,m}^{K_m} +  \| N_{d}(0, I_0^{-1}) -  N_{d}(0,  \Sigma_j^{1/2}\tilde \Sigma^{-1/2}I_{0}^{-1}\tilde \Sigma^{-1/2}\Sigma_j^{1/2}) \|_{\text{TV}} \right\}  > \epsilon \right ]
\nonumber \\
& \overset{(ii)}{\leq }
 \frac{1}{ \epsilon K_{m}} \sum_{j=1}^{K_{m}} 
\EE_0 \left \{D_{j,m}^{K_m} +  \| N_{d}(0, I_0^{-1}) -  N_{d}(0,  \Sigma_j^{1/2}\tilde \Sigma^{-1/2}I_{0}^{-1}\tilde \Sigma^{-1/2}\Sigma_j^{1/2}) \|_{\text{TV}} \right \} 
\nonumber \\
& 
\overset{(iii)}{ = } \frac{1}{ \epsilon K_{m} } \EE_0  D_{1,m}^{K_{m}} + \frac{K_{m} - 1 }{ \epsilon K_{m}  } \EE_0  D_{2,m}^{K_{m}}
+  \frac{1}{ \epsilon K_{m}} \sum_{j=1}^{K_{m}} 
\EE_0  \| N_{d}(0, I_0^{-1}) -  N_{d}(0,  \Sigma_j^{1/2}\tilde \Sigma^{-1/2}I_{0}^{-1}\tilde \Sigma^{-1/2}\Sigma_j^{1/2}) \|_{\text{TV}}
\nonumber \\
& \overset{(iv)}{\leq}
 \frac{1}{ \epsilon K_{m} } \EE_0  D_{1,m}^{K_{m}} + \frac{K_{m} - 1 }{ \epsilon K_{m}  } \EE_0  D_{2,m}^{K_{m}}
+  \frac{\tilde{C}}{\epsilon} \frac{1}{K_{m}}\sum_{j=1}^{K_{m}}\EE_0\|M_{j}\|_{F}(2 + \|M_{j}\|_{F})
\label{eq:mk3}
 \\
& \overset{(v)}{ \rightarrow  } 0,
\nonumber
\end{align}
where $(i)$ follows from \eqref{eq:thm2:3} in step 1, $(ii)$ follows from Markov inequality, $(iii)$ follows from \eqref{eq:mk2}, $(iv)$ follows from \eqref{eq:thm2:8} in step 2. The convergence in $(v)$ follows from that the first term of \eqref{eq:mk3} converges to $0$ as $\frac{1}{K_{m}} \rightarrow  0$ and $ \EE_0  D_{1,m}^{K_{m}} \rightarrow 0$ as $m \rightarrow \infty$; the second term of \eqref{eq:mk3} converges to 0 as $\frac{K_{m} - 1 }{  K_{m}  } \rightarrow 1$ and $ \EE_0  D_{2,m}^{K_{m}} \rightarrow 0$ as $K_{m}\overset{m \rightarrow \infty}{ \rightarrow}\infty$;  for the thrid term of \eqref{eq:mk3}, by \eqref{eq:fix1} and Jensen's inequality we have $\frac{1}{K_{m}}\sum_{j=1}^{K_{m}}\EE_0\|M_{j}\|_{F}(2 + \|M_{j}\|_{F}) \rightarrow 0$  as $m,K_{m}\rightarrow \infty$. And this completes the proof of Theorem \ref{thm2} \ref{thm2:1} for any given sequence $K_{m}\overset{m \rightarrow \infty}{ \rightarrow} \infty$ with $K_{m} = o(\rho^{-m})$. \hfill $\square$

\subsection{Proof of Theorem \ref{thm2}\ref{thm2:2}}

Let $\mu,\nu$ be two probablity measures on $(\Theta,\mathcal{F})$ where $\mathcal{F}$ is the Borel $\sigma$-field. The 1-Wasserstein distance is bounded by the total variation distance \citep[Theorem 6.15]{villani2008optimal}
$$W_{1}\{\mu,\nu\} \leq \text{diam}(\Theta) \|\mu-\nu\|_{\text{TV}},$$
where $ \text{diam}(\Theta)  = \underset{\theta_1 , \theta_2 \in \Theta}{\sup}\|\theta_1 - \theta_2\|_2 < \infty$ by the compactness of $\Theta \subset \mathbb{R}^{d}$. 
Hence, as $m\rightarrow \infty$ with $\PP_0$-probability tending to 1,
\begin{align}
\sqrt{n}W_{1}\left \{\tilde{\Pi}_{n}(\cdot\mid Y_{1}^{n}) , N_{d}(\tilde{\theta}, (nI_{0})^{-1})\right\} & \overset{(i)}{=}   W_{1}\left\{\tilde{\Pi}_{n}^{\tilde{h} = \sqrt{n}(\theta - \tilde{\theta})}(\cdot\mid Y_{1}^{n}) , N_{d}(0, I_{0}^{-1})\right\}  \nonumber\\
&\leq \text{diam}(\Theta) \|\tilde{\Pi}_{n}^{\tilde{h}}(\cdot\mid Y_{1}^{n})  -N_{d}(0, I_{0}^{-1})\|_{\text{TV}} \overset{(ii)}{\rightarrow} 0,
\nonumber\\
\sqrt{n}W_{1}\left \{\Pi_{n}(\cdot\mid Y_{1}^{n}) , N_{d}(\hat{\theta}, (nI_{0})^{-1})\right\} &\overset{(i)}{=}    W_{1}\left\{\Pi_{n}^{h = \sqrt{n}(\theta - \hat{\theta})}(\cdot\mid Y_{1}^{n}) , N_{d}(0, I_{0}^{-1})\right\} \nonumber\\
&\leq \text{diam}(\Theta) \|\Pi_{n}^{h}(\cdot\mid Y_{1}^{n})  -N_{d}(0, I_{0}^{-1})\|_{\text{TV}}\overset{(iii)}{\rightarrow} 0,
\label{eq:39}
\end{align}
where $(i)$ follows from the rescaling property of $W_{1}$-distance, $(ii)$ follows from Theorem \ref{thm2}\ref{thm2:2} and $(iii)$ follows from Theorem \ref{thm1} in the case $K= 1$.
By \citet{Vil03}, if $\Theta$ is bounded, we have the following relation between $W_1$ distance and $W_2$ distance. For two square integrable probability measures $\mu,\nu$ on $(\Theta, \mathcal{F})$,
\begin{align}
W_{2}\{\mu,\nu\}^{1/2} \text{diam}(\Theta)^{1/2}\leq W_{1}\{\mu,\nu\} \leq W_{2}\{\mu,\nu\},
\label{eq:w12}
\end{align}
Furthermore, let $\mu$ be the probability measure induced by $N_{d}(m_{1},\Sigma_1)$ and $\nu$ be the probability measure induced by $N_{d}(m_{2},\Sigma_2)$, then $W_2(\mu,\nu)$ has a closed form \citep{olkin1982distance}
\begin{align}
W_{2}\{N_{d}(m_{1},\Sigma_1) ,N_{d}(m_{1},\Sigma_1)\} = \|m_{1} - m_2\|_{2} + \|\Sigma_{1}^{1/2} - \Sigma_{2}^{1/2}\|_{F}.
\label{eq:40}
\end{align}

Now we consider the upper bound of Theorem \ref{thm2} \ref{thm2:2}:
\begin{align}
& \sqrt{n} W_{1}\left\{\tilde{\Pi}_{n}(\cdot\mid Y_{1}^{n}),\Pi_{n}(\cdot\mid Y_{1}^{n})\right\}
\nonumber\\&
\overset{(i)}{\leq} \sqrt{n}W_{1}\left\{\tilde{\Pi}_{n}(\cdot\mid Y_{1}^{n}) , N_{d}(\tilde{\theta}, (nI_{0})^{-1})\right\}+ \sqrt{n}W_{1}\left\{\tilde{\Pi}_{n}(\cdot\mid Y_{1}^{n}) , N_{d}(\hat{\theta}, (nI_{0})^{-1})\right\}
\nonumber
\\
&
\quad \quad
+ \sqrt{n}W_{1}\left\{N_{d}(\tilde{\theta}, (nI_{0})^{-1}) , N_{d}(\hat{\theta}, (nI_{0})^{-1})\right\}
\nonumber\\&
\overset{(ii)}{\leq}
\text{diam}(\Theta) \left\{\|\tilde{\Pi}_{n}^{\tilde{h}}(\cdot\mid Y_{1}^{n})  -N_{d}(0, I_{0}^{-1})\|_{\text{TV}} +  \|\Pi_{n}^{h}(\cdot\mid Y_{1}^{n})  ,N_{d}(0, I_{0}^{-1})\|_{\text{TV}} \right\} 
\nonumber
\\
&
\quad \quad
+ \sqrt{n}W_{1}\left\{N_{d}(\tilde{\theta}, (nI_{0})^{-1}) , N_{d}(\hat{\theta}, (nI_{0})^{-1})\right\}
\label{eq:41}
\end{align}
where $(i)$ follows from triangle inequality, $(ii)$ follows from \eqref{eq:39}. The first term of the right hand side of \eqref{eq:41} converges to $0$ in $\PP_0$-probability by \eqref{eq:39}. As for the second term, by \eqref{eq:w12} and \eqref{eq:40}, we have that
\begin{align}
\text{diam}(\Theta)^{1/2}\left( \sqrt{n} \|\tilde{\theta}- \hat{\theta}\|_{2}\right)^{1/2} \leq \sqrt{n}W_{1}\left\{N_{d}(\tilde{\theta}, (nI_{0})^{-1}) , N_{d}(\hat{\theta}, (nI_{0})^{-1})\right\} \leq \sqrt{n} \|\tilde{\theta}- \hat{\theta}\|_{2}.
\label{eq:42}
\end{align}
Hence up to a negligible term in $\PP_0$-probability, we have 
\begin{align}
\label{eq:42.5}
\sqrt{n} W_{1}\left\{\tilde{\Pi}_{n}(\cdot\mid Y_{1}^{n}),\Pi_{n}(\cdot\mid Y_{1}^{n})\right\}
\leq \sqrt{n} \|\tilde{\theta}- \hat{\theta}\|_{2}.
\end{align}
For the lower bound of Theorem \ref{thm2}\ref{thm2:2}, we have that 
\begin{align}
& \sqrt{n} W_{1}\left\{\tilde{\Pi}_{n}(\cdot\mid Y_{1}^{n}),\Pi_{n}(\cdot\mid Y_{1}^{n})\right\}
\nonumber\\&
\geq \sqrt{n}W_{1}\left\{N_{d}(\tilde{\theta}, (nI_{0})^{-1}), N_{d}(\hat{\theta}, (nI_{0})^{-1})\right\} - \sqrt{n}W_{1}\left\{\tilde{\Pi}_{n}(\cdot\mid Y_{1}^{n}) , N_{d}(\tilde{\theta}, (nI_{0})^{-1})\right) 
\nonumber
\\
&
\quad \quad
- \sqrt{n}W_{1}\left\{\tilde{\Pi}_{n}(\cdot\mid Y_{1}^{n}) , N_{d}(\hat{\theta}, (nI_{0})^{-1})\right\}
\nonumber \\&
\geq 
\sqrt{n}W_{1}\left\{N_{d}(\tilde{\theta}, (nI_{0})^{-1}) , N_{d}(\hat{\theta}, (nI_{0})^{-1})\right\} 
\nonumber
\\
&
\quad \quad - \text{diam}(\Theta) \left\{\|\tilde{\Pi}_{n}^{\tilde{h}}(\cdot\mid Y_{1}^{n}) , N_{d}(0, I_{0}^{-1})\|_{\text{TV}} +  \|\Pi_{n}^{h}(\cdot\mid Y_{1}^{n})  
-N_{d}(0, I_{0}^{-1})\|_{\text{TV}} \right\} 
\label{eq:43},
\end{align}
where the second term of \eqref{eq:43} converges to 0 in $\PP_0$-probability and the first term is lower bounded by $\text{diam}(\Theta)^{1/2}\left( \sqrt{n} \|\tilde{\theta}- \hat{\theta}\|_{2}\right)^{1/2} $ by \eqref{eq:42}. Hence up to a negligible term in $\PP_0$-probability, we have 
\begin{align}
\label{eq:44}
\sqrt{n} W_{1}\left\{\tilde{\Pi}_{n}(\cdot\mid Y_{1}^{n}),\Pi_{n}(\cdot\mid Y_{1}^{n})\right\}
\geq \text{diam}(\Theta)^{1/2}\left( \sqrt{n} \|\tilde{\theta}- \hat{\theta}\|_{2}\right)^{1/2} .
\end{align}

Combining \eqref{eq:42.5} and \eqref{eq:43} completes the proof of Theorem \ref{thm2}\ref{thm2:2}.
 \hfill $\square$

\section{Further Details of Experiments}
\label{app:exp}
\subsection{Illustrative Example of Simulations}
\label{illus-ex}

We consider an illustrative example of HMM in Section \ref{sim-dat} of the main manuscript. Using the notation from Section \ref{setup-back} in the main manuscript, let $\Ycal = \RR$, $g(\cdot \mid X = a)=\phi(\cdot;\mu_{a},\sigma_{a}^{2})$ be the density of normal distribution with mean $\mu_a$ and variance $\sigma^2_a$ $(a=1, \ldots, S)$, $\theta = \{r, Q, \mu_1, \ldots, \mu_S, \sigma^2_1, \ldots, \sigma^2_S\}\subset \mathbb{R}^{d}$, and $d = S^2 + 2S - 1$. Our example is motivated from \citet{Ryd08}, who uses a similar example with $S=3$, $\sigma_1 = \ldots = \sigma_3 = \sigma$, and a specific choice of $Q$, which automatically determines $r$ because $r^\T Q = r^\T$. Under this setup, posterior computations on subsets using the data augmentation based on \eqref{eq:s6} in the main manuscript are analytically tractable.

The prior distribution of $\theta$ is taken conjugate to the likelihood $p_{\theta} (X_{1}^{n},Y_{1}^{n})$ in \eqref{eq:3} of the main manuscript. The initial distribution $r = (r_1,\ldots,r_{S})$ and rows of the Markov transition matrix $\{(q_{a1},\ldots,q_{aS})\}_{a=1}^{S}$ are given an independent Dirichlet prior Dir$(1,\ldots,1)$; the mean parameters $\mu_{1},\ldots,\mu_{S}$ are given an independent normal prior $N(\xi,\kappa^{-1})$ with $\xi = (\min y_{i} + \max y_{i})/2 $ and $\kappa = 1/(\max y_{i} - \min y_{i})^{2}$; the variance parameters $\sigma^{-2}_{1},\ldots,\sigma^{-2}_{S}$ are given an independent gamma prior $\Gamma(1,1)$. 

For $j = 1,\ldots K$, the full conditional distributions of $\theta = \{r, Q, \mu_1, \ldots, \mu_S, \sigma^2_1, \ldots, \sigma^2_S\}$ after stochastic approximation given $(X_{(j-1)m+1}^{(jm)}, Y_{(j-1)m+1}^{(jm)})$  are as follows. The full conditional distribution of $r$ is
\begin{align}
\label{pos:init}
(r_{1},\ldots,r_{S})\mid \cdots \sim \text{Dir}\left(K\cdot I\{X_{(j-1)m+1}=1\}+1,\ldots, K\cdot I\{X_{(j-1)m+1}=S\}+1 \right)
\end{align}
where $\mid \cdots$ denotes the conditioning on the remaining random quantities and $I\{\cdot\}$ denotes the indicator function. The rows of transition matrix $Q$ are conditional independent over $a = 1,\ldots,S$ and drawn as
\begin{align}
\label{pos:Q}
(q_{a1},\ldots,q_{aS})\mid \cdots \sim \text{Dir}\left(Kn_{a1}+1,\ldots, Kn_{aS}+1 \right)
\end{align}
where $n_{ab} $ is the cardinality of the set $\{(j-1)m+1 < i \leq jm \;:\; X_{i-1}=a,X_{i}=b\}$. The mean parameters are conditionally independent over $a = 1,\ldots,S,$ and drawn as
\begin{align}
\label{pos:mean}
\mu_{a}\mid \cdots \sim \text{N} 
\left(\frac{KS_{a} + \kappa\xi\sigma^{2}}{Kn_{a} + \kappa \sigma^{2}},\frac{\sigma^{2}}{Kn_{a} + \kappa \sigma^{2}} \right)
\end{align}
where $S_{a} = \sum_{X_{i} = a,(j-1)m+1 \leq i \leq jm} y_{i}$ and $n_{a}$ is the cardinality of the set $\{(j-1)m+1 \leq i \leq jm\;:\;X_{i} = a\}$. The variance parameters are conditionally independent over $a = 1,\ldots,S$ and drawn as
\begin{align}
\label{pos:var}
\sigma^{-2}_{a}\mid \cdots \sim \Gamma 
\left(1 + \frac{K}{2}n_{a},1 + \frac{K}{2}\sum_{i \in S_{a}} (y_{i} - \mu_{a})^{2} \right).
\end{align}
Given the parameters and observed chain $\{Y_{t}\}$, the hidden Markov chain is updated as 
\begin{align}
\label{pos:X:init}
\PP(X_{1} = a \mid \ldots) &\propto \{r_{a}\phi(Y_{1};\mu_{a},\sigma_{a}^{2})\PP(Y_{2}^{m}\mid X_{1}=a,\theta)\}^{K},
\end{align}
and as 
\begin{align}
\label{pos:X:tran}
\PP(X_{i} = b\mid X_{i-1} = a)&\propto \{ q_{ab}\phi(Y_{i};\mu_{b},\sigma_{b}^{2})\PP(Y_{i+1}^{jm}\mid X_{i}=b,\theta)\}^{K},
\end{align}
for $1\leq j \leq K, (j-1)m+1 < i \leq jm$. The prediction filter is updated as 
\begin{align}
\PP(&X_{(j-1)m+1} = b \mid Y_{(j-2)m+1}^{(j-1)m}) 
\nonumber
\\
& \propto\sum_{a}\{ \PP(X_{(j-1)m+1} = b \mid X_{(j-1)m} = a)\PP_{\theta}(Y_{(j-2)m+1}^{(j-1)m},X_{(j-1)m}=a)\}^{K}
\nonumber
\\
& = \sum_{a}\{q_{ab}\PP_{\theta}(Y_{(j-2)m+1}^{(j-1)m},X_{(j-1)m}=a)\}^{K}
\label{pos:X:pf}
\end{align}
for $1<j\leq K$,
where $\PP_{\theta}(Y_{(j-2)m+1}^{(j-1)m},X_{(j-1)m}=a)$ is calculated using forward-backward recursion. 

The subset sampling scheme is carried out in parallel using Gibbs sampler that alternates between updating parameter $\theta$ and the hidden Markov chain $X$. On the first subset, we draw parameters by cycling through the sequence for the following six steps until the Markov chain convergent to the stationary distribution: 
\begin{enumerate}[label=(\subscript{a}{{\arabic*}})]
\item  Draw $X_{1},\ldots,X_{m}$ from $\{1,\ldots,S\}$ with equal weights as initial value for hidden Markov chain.
\label{a.1}

\item Update $(r_{1},\ldots,r_{S})$ according to \eqref{pos:init}.

\item Update $Q$ by drawing $(q_{a1},\ldots,q_{aS})$ according to \eqref{pos:Q} independently for $a = 1,\ldots,S$.
\label{a.3}

\item Update $\mu_{a}$ according to \eqref{pos:mean} independently for $a = 1,\ldots,K$.

\item Update $\sigma^{-2}_{a}$ according to \eqref{pos:var} independently for $a = 1,\ldots,K$.

\item Update $(X_{1},\ldots,X_{m})$ by drawing $X_{1}$ from \eqref{pos:X:init} and $X_{i-1}\rightarrow X_{i}$ from \eqref{pos:X:tran} for $i = 2,\ldots,m$.
\label{a.6}
\end{enumerate}
On $j$th subset ($j=2, \ldots, K$), we use the prediction filter on the immediately preceding subset to approximate the initial distribution of the hidden Markov chain; therefore, we require access to $2m$ observations ($j$th and $(j-1)$th subsets), where the $(j-1)$th subset is used to estimate the prediction filter and the $j$th subset is used to sample posterior draws $\theta$. For $j =2,\ldots,K$, with access to $Y_{(j-2)m+1}^{(j-1)m}$ and $Y_{(j-1)m+1}^{jm}$, run the above Gibbs sampler \ref{a.1} - \ref{a.6} on $Y_{(j-2)m+1}^{(j-1)m}$ to get an estimation of prediction filter according to \eqref{pos:X:pf}, and then run Gibbs sampler on $Y_{(j-1)m+1}^{jm}$ while setting $r = (r_{1},\ldots,r_{S})$ equal to the estimation of prediction filter obtained from $Y_{(j-2)m+1}^{(j-1)m}$.

\begin{enumerate}[label=(\subscript{b}{{\arabic*}})]
\item Run the Gibbs sampler \ref{a.1} - \ref{a.6} with observations $Y_{(j-2)m+1}^{(j-1)m}$.

\item Calculate the prediction filter weights $(w_{1},\ldots,w_{S})$ according to \eqref{pos:X:pf} based on the MCMC estimation of paramters.

\item Set $(r_{1},\ldots,r_{S}) =  (w_{1},\ldots,w_{S}) / \sum_{a}w_{a}$ as initial distribution for hidden Markov chain. Then draw $X_{(j-1)m+1}$ with weight $(r_{1},\ldots,r_{S})$ and $X_{(j-1)m+2},\ldots,X_{jm}$ with equal weight from $\{1,\ldots,S\}$.

\item Run the Gibbs sampler \ref{a.3} - \ref{a.6} with observations $Y_{(j-1)m+1}^{jm}$.
\end{enumerate}
If we set $K =1$, then we get the same sampling scheme as those in \citet{Ryd08}. Let $\{\theta^{(t)}_{(j)}\}_{t = 1,j=1}^{T,K}$ be the posterior draws collected from subset posterior sampling scheme. Then, the final step of divide-and-conquer method simply re-scales and re-centers the scaled and centered draws from subsets as follows:
\begin{enumerate}[label=(\subscript{c}{{\arabic*}})]
\item Calculate the maximum likelihood estimator $\hat{\theta}$ using the full data $Y_{1}^{n}$ by Baum-Welch EM algorithm and the Monte Carlo estimate of subset posterior variance matrices $(\Sigma_{\theta_{(j)}})_{j=1}^ K$ by
\begin{align}
  \label{eq:10}
  \mu_{\theta_{(j)}} = \frac{1}{T} \sum_{t=1}^T \theta_{(j)}^{(t)}, \quad \Sigma_{\theta_{(j)}} = \frac{1}{T} \sum_{t=1}^T \left(\theta^{(t)}_{(j)} - \mu_{\theta_{(j)}} \right) \left(\theta^{(t)}_{(j)} - \mu_{\theta_{(j)}} \right)^\T,\quad j=1,\ldots,K.
\end{align} 
\item Adopt $COMB(\tilde\theta = \hat{\theta}, \tilde{\Sigma} = \frac{1}{K}\sum_{j=1}^{K} \Sigma_{\theta_{(j)}})$ and then calculate the $\theta$ draws from block filtered posterior distribution according to the following
\begin{align} 
 \tilde \theta_{(j)}^{(t)} = \tilde \theta + \tilde \Sigma^{1/2} \left\{ \Sigma_{\theta_{(j)}}^{-1/2}(\theta_{(j)}^{(t)}  - \hat \theta_j) \right\}, \quad j = 1, \ldots, K; \; t = 1, \ldots, T.
\end{align}

\end{enumerate}

\subsection{Time Comparisons and Accuracy for posterior inference on Q}
\label{app:sim}

\emph{Simulation Study of Accuracy and Efficiency.}
Following the discussion of the simulation study in the main manuscript of HMM with $(S=3,Q_{\epsilon} = 0.1,  \mu_1 = -2, \mu_2 = 0, \mu_3 = 2)$, we compare the accuracy of block filtered posterior distribution and its divide-and-conquer competitors for posterior inference on $Q$ (Table \ref{tab:Q-sim}). The divide-and-conquer competitors of block filtered posterior distribution have low accuracy for $n = 10^{6}$ and all $K$. In comparison, the block filtered posterior distribution has the highest accuracy for all $(n,K)$ combinations, and its accuracy is close to the benchmark accuracy of the Hamiltonian Monte Carlo. These observations agree with that of the simulation study for posterior inference on the emission distribution parameters in the main manuscript.

The block filtered posterior distribution is about ten times faster than Hamiltonian Monte Carlo and about fifteen times faster than data augmentation when $n=10^6$ and $K= n^{1/3}$ (Table \ref{tab:t-sim}). For all $n$, the running time of block filtered posterior distribution decreases as $K$ increases from $\log n$ to $n^{1/3}$ because larger $K$ results in smaller subsets size. The run-time of subset posterior sampling algorithms dominates the time required to combine the parameter draws from the subsets; therefore, the block filtered posterior distribution is more efficient than data augmentation and Hamiltonian Monte Carlo for all values of $K$ when $n = 10^{5} $ and $n= 10^{6}$.

\label{tables}
\begin{table}[ht]
\let\TPToverlap=\TPTrlap
\centering
\caption{\it Accuracy of the approximate posterior distributions of $Q$. The accuracies are averaged over ten simulation replications and across all the dimensions of $Q$. The maximum Monte Carlo error for the block filtered posterior distribution is 0.02}
\label{tab:Q-sim}
\begin{threeparttable}
  \begin{tabular}{|l|l|c|c|c|c|c|}
    \hline
   $n$ & $K$ & BFP & WASP & DPMC & PIE & HMC\\ 
    \hline

    & $\log n$ & 0.96 & 0.87 & 0.87 & 0.88 & \\ 
   $10^4$  &  $n^{1/4}$ & 0.96 & 0.87 & 0.87 & 0.87 & 0.96\\ 
    & $n^{1/3}$ & 0.93 & 0.58 & 0.57 & 0.58&  \\ 
    \hline

    & $\log n$ & 0.97 & 0.62 & 0.62 & 0.62 & \\ 
   $10^5$ & $n^{1/4}$ & 0.97 & 0.63 & 0.63 & 0.63 & 0.95\\ 
    & $n^{1/3}$ & 0.97 & 0.66 & 0.66 & 0.66 & \\ 
    \hline

    & $\log n$ & 0.96 & 0.22 & 0.22 & 0.23 & \\ 
   $10^6$ & $n^{1/4}$ & 0.97 & 0.22 & 0.22 & 0.23 & 0.96\\ 
    & $n^{1/3}$ & 0.96 & 0.25 & 0.25 & 0.25 &  \\ 
    \hline
  \end{tabular}
\begin{tablenotes}[normal,flushleft]
\small
\item[]
 BFP, block filtered posterior distribution; WASP, Wasserstein posterior; DPMC, double parallel Monte Carlo; PIE, posterior interval estimation; HMC, Hamiltonian Monte Carlo.
\end{tablenotes}
\end{threeparttable}
\end{table}


\begin{table}[ht]
\let\TPToverlap=\TPTrlap
\caption{\it Time (in hours) for computing the posterior distributions of $\theta$. The times are averaged over ten simulation replications}
\label{tab:t-sim}
\centering
\begin{threeparttable}
  \begin{tabular}{|l|c|c|c|c|c|}
    \hline
    & DA & \multicolumn{3}{c|}{BFP} & HMC\\
    \hline
      &  & $\log n$ & $n^{1/4}$ & $n^{1/3}$ & \\ 
    \hline
    $n = 10^4$ &1.36  & 1.88 & 0.73 & 0.33 & 0.68 \\ 
    $n = 10^5$ &14.80  & 5.58 & 4.97 & 1.22  & 8.60 \\ 
    $n = 10^6$ &131.13   & 59.76 & 21.16 & 8.83 & 91.89 \\ 
    \hline
  \end{tabular}
\begin{tablenotes}
\small
\item[]
DA, data augmentation; BFP, block filtered posterior distribution; HMC, Hamiltonian Monte Carlo.
\end{tablenotes}
\end{threeparttable}
\end{table}

\emph{Simulation Study of Mixing Property.} Following the discussion of the simulation study in the main manuscript of three more HMMs with $S = 2,5,7$, the accuracy of block filtered posterior distribution and its divide-and-conquer competitors for posterior inference on $Q$ is compared in Tables \ref{tab:Q-sim} and \ref{tab:rho:Q-sim}. For HMMs with $S = 2,3$,  $\rho$ is relatively small and the block filtered posterior distribution has similar accuracy compared to Hamiltonian Monte Carlo over all $(n,K)$ combinations. On the other hand, HMMs with $S = 5,7$ have  $\rho \approx 1$ and the block filtered posterior distribution has a smaller accuracy compared to Hamiltonian Monte Carlo. Furthermore, as $K$ increases, the accuracy of block filtered posterior distribution decreases because of smaller subset size and increased dependence.

For HMMs with $S=2,3,5,7$ and every $(n,K)$ combinations, the block filtered posterior distribution is the top performer compared to its divide-and-conquer competitors. For HMMs with $S=2,3$ and relatively weak dependence, the divide-and-conquer competitors have acceptable accuracy, but their accuracy values decrease quickly as the dependence $\rho$ increases and drop to $0$ for $S=5,7.$ In contrast, the block filtered posterior distribution keeps its great performance on the accuracy for $S=5,7$ and all $(n,K)$ combinations. These observations agree with that of the simulation study for posterior inference on the emission distribution parameters in the main manuscript.
\begin{table}[ht]
\let\TPToverlap=\TPTrlap
\centering
\caption{\it Accuracy of the approximate posterior distributions of $Q$. The accuracies are averaged over ten simulation replications and across all the dimensions of $Q$. The maximum Monte Carlo error for the block filtered posterior distribution is 0.02}
\label{tab:rho:Q-sim}
\begin{threeparttable}
  \begin{tabular}{|l|l|c|c|c|c|c|}
    \hline
    $n$ & $K$ & BFP & WASP & DPMC & PIE & HMC\\
    \hline
    \multicolumn{7}{|c|}{$S = 2$}\\
    \hline
    & $\log n$ & 0.97 & 0.98 & 0.98 & 0.98 & \\
   $10^4$  &  $n^{1/4}$ & 0.97 & 0.98 & 0.98 & 0.99 & 0.96 \\
    & $n^{1/3}$ & 0.97 & 0.96 & 0.96 & 0.96 & \\
    \hline
    & $\log n$ & 0.97 & 0.98 & 0.98 & 0.99 & \\
   $10^5$ & $n^{1/4}$ & 0.97 & 0.98 & 0.98 & 0.99 & 0.96 \\
    & $n^{1/3}$ & 0.97 & 0.98 & 0.98 & 0.98 & \\
    \hline
    \multicolumn{7}{|c|}{$S = 5$}\\
    \hline
    & $\log n$ & 0.65 & 0.00 & 0.00 & 0.00 &  \\
   $10^4$  &  $n^{1/4}$ & 0.66 & 0.00 & 0.00 & 0.00 & 0.95  \\
    & $n^{1/3}$ & 0.62 & 0.00 & 0.00 & 0.00 &  \\
    \hline
    & $\log n$ & 0.68 & 0.00 & 0.00 & 0.00 &  \\
   $10^5$ & $n^{1/4}$ & 0.65 & 0.00 & 0.00 & 0.00 & 0.95\\
    & $n^{1/3}$ & 0.43 & 0.00 & 0.00 & 0.00 &  \\
    \hline
    \multicolumn{7}{|c|}{$S = 7$}\\
    \hline
    & $\log n$ & 0.56 & 0.01 & 0.01 & 0.00 &  \\
   $10^4$  &  $n^{1/4}$ & 0.57 & 0.00 & 0.01 & 0.00 & 0.94  \\
    & $n^{1/3}$ & 0.56 & 0.00 & 0.00 & 0.00 &  \\
    \hline
    & $\log n$ & 0.51 & 0.00 & 0.00 & 0.00 &  \\
   $10^5$ & $n^{1/4}$ & 0.54 & 0.00 & 0.00 & 0.00 &  0.94\\
    & $n^{1/3}$ & 0.55 & 0.00 & 0.00 & 0.00 &  \\
    \hline        
  \end{tabular}
\begin{tablenotes}[normal,flushleft]
\small
\item[]
 BFP, block filtered posterior distribution; WASP, Wasserstein posterior; DPMC, double parallel Monte Carlo; PIE, posterior interval estimation; HMC, Hamiltonian Monte Carlo.
\end{tablenotes}
\end{threeparttable}
\end{table}

\emph{Real Data Analysis.}
Following the discussion of the read data analysis in the main manuscript of HMM, we compare the accuracy of block filtered posterior distribution and its divide-and-conquer competitors for posterior inference on $Q$ (Table \ref{tab:par-real-q}). For $K = \log n, n^{1/4}$, and $n^{1/3}$, the block filtered posterior distribution outperforms its divide-and-conquer competitors, and its accuracy is close to that of Hamiltonian Monte Carlo. As for the efficiency comparison (Table \ref{tab:t-real}), the block filtered posterior distribution is more efficient than data augmentation for all choices of $K$. For $K = n^{1/3}$, computing block filtered posterior distribution is more efficient than Hamiltonian Monte Carlo.

\begin{table}[ht]
  \caption{\it Accuracy of the approximate posterior distributions of $Q$. The entries in the table are the median of the accuracies computed across the $Q$ dimensions. The maximum Monte Carlo error for the block filtered posterior distribution is 0.01}
\label{tab:par-real-q}  
\centering
\begin{threeparttable}
{\centering
  \begin{tabular}{|l|c|c|c|c|c|}
    \hline
    $K$ & BFP & WASP & DPMC & PIE & HMC \\
    \hline
    $\log n$ & 0.89 & 0.05 & 0.05 & 0.05 &0.96\\
    $n^{1/4}$ & 0.87 & 0.05 & 0.05 & 0.05 &\\
    $n^{1/3}$ & 0.84 & 0.06 & 0.06 & 0.06 &\\
    \hline
  \end{tabular}
}
\begin{tablenotes}
\small
\item[] BFP, block filtered posterior distribution; WASP, Wasserstein posterior; DPMC, double parallel Monte Carlo; PIE, posterior interval estimation; HMC, Hamiltonian Monte Carlo.
\end{tablenotes}
\end{threeparttable}
\end{table}

\begin{table}[ht]
  \caption{\it Time (in hours) for computing the posterior distributions of $\theta$ }
\label{tab:t-real}    
\centering
\begin{threeparttable}
  \begin{tabular}{|c|c|c|c|c|}
    \hline
    DA &\multicolumn{3}{c|}{BFP} & HMC \\
    \hline
     & $\log n$ & {$n^{1/4}$} & {$n^{1/3}$}& \\ 
    \hline
    4.53 & 2.27 & 1.53 & 0.70 & 1.11\\ 
    \hline
  \end{tabular}
 \begin{tablenotes}
\small
\item[]
 DA, data augmentation; BFP, block filtered posterior distribution; HMC, Hamiltonian Monte Carlo.
 \end{tablenotes}
\end{threeparttable}
\end{table}

\section*{Acknowledgements}
Chunlei Wang and Sanvesh Srivastava are partially supported by grants from the Office of Naval Research (ONR-BAA N000141812741) and the National Science Foundation (DMS-1854667/1854662).

\bibliographystyle{Chicago}
\bibliography{papers}

\end{document}